\documentclass[12pt, a4paper]{article}%
\usepackage{amsmath,amssymb,eucal}
\usepackage{amsfonts}
\usepackage{mathrsfs}
\usepackage{slashed}
\numberwithin{equation}{section} 
\def\beq{\begin{equation}}
\def\eeq{\end{equation}}
\def\bE{ {{\mathbb{E}}}}
\def\bC{ {{\mathbb{C}}}}
\def\bR{ {{\mathbb{R}}}}
\def\Tr{ {{\rm{Tr}}} }
\def\ad{ {{\rm{ad}}} }

\newcommand{\dB}{{\mathbb{B}}}

\newcommand{\sgn}{{\rm sgn}}
\newcommand{\pk}[1]{p_{\kappa}}
\newcommand{\lba}{\lceil}
\newcommand{\rba}{\rceil}
\newcommand{\lca}{\lfloor}
\newcommand{\rca}{\rfloor}
\newcommand\myfrac[2]{\genfrac{}{}{0pt}{}{#1}{#2}}

\def\fd{ {\mathfrak{d}} }

\newtheorem{defn}{{\bf Definition}}[section]
\newtheorem{thm}[defn]{{\bf Theorem}}
\newtheorem{cor}[defn]{{\bf Corollary}}
\newtheorem{lem}[defn]{{\bf Lemma}}
\newtheorem{prop}[defn]{{\bf Proposition}}
\newtheorem{rem}[defn]{{\bf Remark}}
\newtheorem{example}[defn]{Example}
\newtheorem{notation}[defn]{Notation}

\newenvironment{proof}[1][Proof]{\textbf{#1.} }{\hfill \rule{0.5em}{0.5em}}
\begin{document}

\title{Positive mass gap of quantum Yang-Mills Fields}
\author{Adrian P. C. Lim \\
Email: ppcube@gmail.com
}

\date{}

\maketitle

\begin{abstract}
We construct a 4-dimensional quantum field theory on a Hilbert space, dependent on a simple Lie Algebra of a compact Lie group, that satisfies Wightman's axioms. This Hilbert space can be written as a countable sum of non-separable Hilbert spaces, each indexed by a non-trivial, inequivalent irreducible representation of the Lie Algebra.

In each component Hilbert space, a state is given by a triple, a space-like rectangular surface $S$ in $\mathbb{R}^4$, a measurable section of the Lie Algebra bundle over this surface $S$, represented irreducibly as a matrix, and a Minkowski frame. The inner product is associated with the area of the surface $S$.

In our previous work, we constructed a Yang-Mills measure for a compact semi-simple gauge group. We will use a Yang-Mills path integral to quantize the momentum and energy in this theory. During the quantization process, renormalization techniques and asymptotic freedom will be used. Each component Hilbert space is the eigenspace for the momentum operator and Hamiltonian, and the corresponding Hamiltonian eigenvalue is given by the quadratic Casimir operator. The eigenvalue of the corresponding momentum operator will be shown to be strictly less than the eigenvalue of the Hamiltonian, hence showing the existence of a positive mass gap in each component Hilbert space. We will further show that the infimum of the set containing positive mass gaps, each indexed by an irreducible representation, is strictly positive.

In the last section, we will show how the positive mass gap will imply the Clustering Theorem.
\end{abstract}

\hspace{.35cm}{\small {\bf MSC} 2020: 81T13, 81T08, 81T70} \\
\indent \hspace{.35cm}{\small {\bf Keywords}: Mass gap, Yang-Mills, Wightman's axioms, compact simple Lie group,  \\
\indent \hspace{2.4cm} renormalization, asymptotic freedom, clustering, space-like surface, \\
\indent \hspace{2.4cm} time-like surface, Casimir operator, Lorentz transformation, ${\rm SL}(2,\bC)$}, \\
\indent \hspace{2.4cm} Callan-Symanzik Equation, path integral, spinor representation



\tableofcontents

\section{Preliminaries}\label{s.pre}

Let $M$ be a 4-manifold, with $\Lambda^q(T^\ast M)$ being the $q$-th exterior power of the cotangent bundle over the manifold $M$. Fix a Riemannian metric $g$ on $M$ and this in turn defines an inner product $\langle \cdot, \cdot \rangle_q$ on $\Lambda^q(T^\ast M)$, for which we can define a volume form $d\omega$ on $M$. This allows us to define a Hodge star operator $\ast$ acting on $k$-forms, $\ast: \Lambda^k(T^\ast M) \rightarrow \Lambda^{4-k}(T^\ast M)$ such that for $u, v \in \Lambda^k(T^\ast M)$, we have \beq u \wedge \ast v = \langle u, v \rangle_k\ d\omega. \label{e.x.7} \eeq  An inner product on the set of smooth sections $\Gamma(\Lambda^k(T^\ast M))$ is then defined as \beq \langle u, v \rangle = \int_M u \wedge \ast v = \int_M \langle u, v \rangle_k\ d\omega. \label{e.x.7a} \eeq See \cite{MR1312606}.

Introduce a compact and simple gauge group $G$. Without loss of generality, we will assume that $G$ is a Lie subgroup of ${\rm U}(\bar{N})$, $\bar{N} \in \mathbb{N}$. We will identify the (real) Lie Algebra $\mathfrak{g}$ of $G$ with a Lie subalgebra of the Lie Algebra $\mathfrak{u}(\bar{N})$ of ${\rm U}(\bar{N})$ throughout this article. Suppose we write the trace as $\Tr_{{\rm Mat}(\bar{N}, \mathbb{C})}$, which we will abbreviate as $\Tr$ in future. Then we can define a positive, non-degenerate bilinear form by \beq \langle A, B \rangle = -\Tr_{{\rm Mat}(\bar{N}, \mathbb{C})}[AB] \label{e.i.2} \eeq for $A,B \in \mathfrak{g}$. Its Lie bracket will be denoted by $[A, B]\equiv \ad(A)B$.

Let $P \rightarrow M$ be some trivial vector bundle, with structure group $G$. The vector space of all smooth $\mathfrak{g}$-valued 1-forms on the manifold $M$ will be denoted by $\mathcal{A}_{M, \mathfrak{g}}$.
Denote the group of all smooth $G$-valued mappings on $M$ by $\mathcal{G}$, called the gauge group. The gauge group induces a gauge transformation on $\mathcal{A}_{M, \mathfrak{g}}$,  $\mathcal{A}_{M, \mathfrak{g}} \times \mathcal{G} \rightarrow \mathcal{A}_{M, \mathfrak{g}}$ given by \beq A \cdot \Omega := A^{\Omega} = \Omega^{-1}d\Omega + \Omega^{-1}A\Omega \nonumber \eeq for $A \in \mathcal{A}_{M, \mathfrak{g}}$, $\Omega \in \mathcal{G}$. The orbit of an element $A \in \mathcal{A}_{M, \mathfrak{g}}$ under this operation will be denoted by $[A]$ and the set of all orbits by $\mathcal{A}_{M, \mathfrak{g}}/\mathcal{G}$.

For $A \in \mathcal{A}_{M, \mathfrak{g}}$, the curvature $dA + A \wedge A$ is a smooth $\mathfrak{g}$-valued 2-form on $M$, whereby $dA$ is the differential of $A$ and $A \wedge A$ is computed using the Lie Bracket of $\mathfrak{g}$ and the wedge product on $\Lambda^1(T^\ast M)$, at each fibre of the tensor bundle $\Lambda^1(T^\ast M) \otimes (M \times \mathfrak{g}  \rightarrow M)$. The Yang-Mills Lagrangian is given by \beq S_{{\rm YM}}(A) = \int_{M}  \left|dA + A \wedge A \right|^2\ d\omega. \nonumber \eeq Here, the induced norm $|\cdot|$ is defined from the tensor product of $\langle \cdot, \cdot \rangle_2$ and the inner product on $\mathfrak{g}$, computed on each fiber of the bundle $\Lambda^2(T^\ast M) \otimes (M \times \mathfrak{g} \rightarrow M)$. The integral over $M$ is then defined using Equation (\ref{e.x.7a}). Note that this Lagrangian is invariant under gauge transformations.

The 4-manifold we will consider in this article is $\bR \times \bR^3\equiv \bR^4$, with tangent bundle $T\bR^4$. Note that $\bR$ will be referred to as the time-axis and $\bR^3$ is the spatial 3-dimensional Euclidean space. We will fix the coordinate axes with corresponding coordinates $\vec{x} \equiv (x^0, x^1, x^2, x^3)^T$, $x^0$ is the time coordinate and spatial coordinates $(x^1, x^2, x^3)^T$, hence defining an orthonormal basis $\{e_a\}_{a=0}^3$ on $\bR^4 \equiv \bR \times \bR^3$. We will also choose the standard Riemannian metric (Euclidean metric) on $T\mathbb{R}^4$, denoted as $\langle \cdot, \cdot \rangle$.

Let $\Lambda^q(\bR^4)$ denote the fiber of the $q$-th exterior power of the cotangent bundle over $\bR^4$, and we choose the canonical basis
$\{dx^0, dx^1, dx^2, dx^3\}$ for $\Lambda^1(\bR^4)$. Let $\Lambda^1(\bR^3)$ denote the subspace in $\Lambda^1(\bR^4)$ spanned by $\{dx^1, dx^2, dx^3\}$. There is an obvious inner product defined on $\Lambda^1(\bR^4)$, i.e. $\langle dx^a, dx^b \rangle = 0$ if $a \neq b$, 1 otherwise. Finally, a basis for $\Lambda^2(\bR^4)$ is given by \beq \{dx^0 \wedge dx^1, dx^0 \wedge dx^2, dx^0 \wedge dx^3, dx^1 \wedge dx^2, dx^3 \wedge dx^1, dx^2\wedge dx^3 \}. \nonumber \eeq

Using the volume form $d\omega = dx^0 \wedge dx^1\wedge dx^2 \wedge dx^3$, the Hodge star operator $\ast$ is a linear isomorphism between $\Lambda^2(\bR^4)$ and $\Lambda^2(\bR^4)$, i.e.
\begin{align*}
\ast(dx^0 \wedge dx^1) = dx^2 \wedge dx^3,\quad \ast(dx^0 \wedge dx^2) = dx^3 \wedge dx^1,\quad  \ast(dx^0 \wedge dx^3) = dx^1\wedge dx^2, \\
\ast(dx^2 \wedge dx^3) = dx^0 \wedge dx^1,\quad \ast(dx^3 \wedge dx^1) = dx^0 \wedge dx^2,\quad  \ast(dx^1 \wedge dx^2) = dx^0\wedge dx^3.
\end{align*}

We adopt Einstein's summation convention, i.e. we sum over repeated superscripts and subscripts. We set the speed of light $\mathbf{c} = 1$. We can define the Minkowski metric, given by \beq \vec{x}\cdot \vec{y} = -x^0y^0 + \sum_{i=1}^3 x^iy^i. \label{e.inn.4} \eeq Note that our Minkowski metric is negative of the one used by physicists. A vector $\vec{x}$ is time-like (space-like) if $\vec{x} \cdot \vec{x} < 0$ ($\vec{x} \cdot \vec{x} > 0$). It is null if $\vec{x} \cdot \vec{x} = 0$. When $\vec{x}$ and $\vec{y}$ are space-like separated, it means that \beq -(x^0-y^0)^2 + \sum_{i=1}^3(x^i-y^i)^2 \geq 0. \nonumber \eeq

A Lorentz transformation $\Lambda$ is a linear transformation mapping space-time $\bR^4$ onto itself, which preserves the Minkowski metric given in Equation (\ref{e.inn.4}). Indeed, the Lorentz transformations form a group, referred to as Lorentz group $L$. It has 4 components, and we will call the component containing the identity, as the restricted Lorentz group, denoted $L_+^\uparrow$.

\subsection{Why should one read this article}

At the time of this writing, an online search will reveal that several authors have attempted to solve the Yang-Mills mass gap problem, which is one of the millennium problems, as described in \cite{jaffe}. The problem is to construct a Hilbert space satisfying Wightman's axioms for a compact, simple Yang-Mills gauge theory in 4-dimensional Euclidean space. See \cite{streater} for a complete description of the axioms. The axioms are also stated in \cite{glimm1981}. Furthermore, this theory has a minimum positive mass gap. This means that besides the zero eigenvalue, the Hamiltonian has a minimum positive eigenvalue. The momentum operator is also a non-negative operator, and the difference between the squares of their eigenvalues, is known as the mass gap squared. Despite all these attempts, no solution has been widely accepted by the scientific community. So, why should one spend time reading this article?

It is known that the set of inequivalent, non-trivial, irreducible representation $\{\rho_n: n\in \mathbb{N}\}$ of a simple Lie Algebra $\mathfrak{g}$ is indexed by highest weights, hence countable. See \cite{hall2015}. Using all the inequivalent, non-trivial, irreducible  representations of $\mathfrak{g}$, we can construct a Hilbert space \beq \mathbb{H}_{{\rm YM}}(\mathfrak{g}) = \{1\} \oplus \bigoplus_{n \geq 1}\mathscr{H}(\rho_n), \nonumber \eeq for which Wightman's axioms are satisfied. The vacuum state will be denoted by 1 and $\{1\}$ will denote its linear span. The Hilbert space $\mathscr{H}(\rho_n)$ is defined using the non-trivial irreducible representation $\rho_n$.

Experiments have shown that quantum fields are highly singular. Hence, we need to smear the field, which we will represent it as a $\rho(\mathfrak{g})$-valued vector field, over a rectangular surface. This will define a state in $\mathscr{H}(\rho)$.

Each state in $\mathscr{H}(\rho)$ is described by a space-like rectangular surface $S$ equipped with a Minkowski frame, and a measurable section of $S \times [\rho(\mathfrak{g}) \otimes \bC] \rightarrow S$ defined on it. See Definition \ref{d.a.1}. This is where the geometry comes in, as surfaces have a well-defined physical quantity, which is the area. Incidentally, when we compute the average flux of non-abelian Yang-Mills gauge fields through a surface in \cite{YM-Lim02}, we will obtain a formula for the area of the surface. Further justification for its connection with the Yang-Mills action will be given in Remark \ref{r.f.1}.

But area is not invariant under the action of the Poincare group. To define an unitary action by the Poincare group, we will consider time to be purely imaginary and hence define a physical quantity on the surface $S$, associated with the area of the surface. See Definition \ref{d.r.2}.

We can now allow a test function to act on this `smeared' field. This test function will be represented as a field operator, to be defined later in Section \ref{s.qfo}, and it acts on a dense subset in the Hilbert space $\mathbb{H}_{{\rm YM}}(\mathfrak{g})$, as required in Wightman axioms.

A non-abelian Yang-Mills measure was constructed in \cite{YM-Lim02}. Using a Yang-Mills path integral, we will proceed in Section \ref{s.iop} to quantize momentum, which will yield its quantum eigenvalues, associated with the quadratic Casimir operator. A path integral approach to quantize the Yang-Mills theory was described in \cite{Skript2}. During the quantization process, we will use renormalization techniques and asympotic freedom. Incidentally, $\mathscr{H}(\rho_n)$ is the eigenspace for both the Hamiltonian and quantized momentum operator. We will now list down the following reasons why we think this is a correct approach to prove the existence of a mass gap.

The quadratic Casimir operator is dependent on the non-trivial irreducible representation of the simple Lie Algebra $\mathfrak{g}$. It is a positive operator and proportional to the identity. If our gauge group is ${\rm SU}(2)$, then the quadratic Casimir operator is given by $j(j+1)$, for a given representation $\rho: \mathfrak{su}(2) \rightarrow {\rm End}(\bC^{2j+1})$, $j$ is a non-negative half integer or integer. Its square root is well known to be the eigenvalue of total momentum in quantum mechanics.

The square of the Hamiltonian will be defined later to be proportional to the dimension of the representation times the quadratic Casimir operator. Hence, we will correlate the Casimir operator directly with the energy levels squared. Large values of the Casimir operator mean high energy levels. The momentum eigenvalues will be computed via a path integral. The set containing the eigenvalues will be discrete, and because the eigenvalues of both operators go to infinity, both are unbounded. See subsection \ref{ss.pt}. By comparing the eigenvalues of the squares of the Hamiltonian and the quantized momentum operator, we will prove the existence of a mass gap. See subsection \ref{ss.mg}. Note that it is not enough to just show that the Hamiltonian has a strictly positive minimum eigenvalue, besides the zero eigenvalue. Our construction will also show that the the vacuum state is an eigenstate of the Hamiltonian and quantized momentum operator with eigenvalue 0 respectively, implying the vacuum state is massless.

A successful quantum Yang-Mills theory should explain the following:
\begin{description}
  \item[1.] There is a mass gap, which will explain why the strong force is short range.
  \item[2.] The theory should incorporate asymptotic freedom, i.e. at high energies and short distances, the theory is like a free theory.
\end{description}

In the case of gauge group ${\rm SU}(3)$, which describes the strong force, it must also explain the following:
\begin{description}
  \item[3.] The theory should demonstrate quark confinement, i.e. the potential between a quark and anti-quark grows linearly;
  \item[4.] The theory should incorporate chiral symmetry breaking, which means that the vacuum is potentially invariant only under a certain subgroup of the full symmetry group that acts on the quark fields.
\end{description}

Item 2 was proved in a prequel \cite{YM-Lim02}, but we only showed asymptotic freedom in the context for short distances. In that prequel, we derived the Wilson Area Law formula of a time-like surface $S$ using a (non-abelian) Yang-Mills path integral. This will show quark confinement and that the potential energy between quarks grows linearly, which is Item 3. We also like to mention that the Area Law formula does not hold in the abelian gauge group case, as was shown in \cite{YM-Lim01}.

In this article, our main goal will be to prove Item 1, and also show that asymptotic freedom holds at high energies. Item 4 can be demonstrated in Example \ref{ex.f.1}. If one assumes that asymptotic freedom holds for non-abelian simple compact gauge group, then a positive mass gap is implied, which we will furnish the details later in this article. On pages 541-543 in \cite{peskin1995}, and also in \cite{APP1977231}, the authors gave a qualitative explanation of asymptotic freedom, using gauge group ${\rm SU}(2)$. See also \cite{PhysRevLett31851}.

The weak interaction is described by a ${\rm SU}(2)$ gauge theory. An unification of the weak and electromagnetic interaction is described by the electroweak theory, which is ${\rm SU}(2) \times {\rm U}(1)$ gauge theory, first put forth by Sheldon Glashow in 1961, then completed later by Abdus Salam and Steven Weinberg in 1967. Mathematically formulated as Yang-Mills fields, the gauge bosons in this theory have to be massless. But experimentally, it was shown that the gauge bosons responsible for weak interaction are massive, hence short range; whereas the photons, responsible for electromagnetic interaction, are massless, hence long range. To resolve this issue, the Higgs mechanism was introduced. See page 330 in \cite{grif2008}. But to date, there is no experimental evidence to suggest that gluons have any physical mass, even though strong interaction is short range.

A successful quantum Yang-Mills theory that satisfies Wightman's axioms and exhibits a positive mass gap, will then imply and explain the short range nature of the weak and strong interaction, without assuming the existence of a physical non-zero mass of these gauge bosons. This will be mathematically formulated later as the Clustering Theorem in Section \ref{s.cluster}.

This article is focused on a construction of a 4-dimensional Yang-Mills quantum field theory that satisfies Wightman's axioms, without assuming the existence of a positive mass. The proof that a positive mass gap exists, requires a construction of a Yang-Mills path integral, which is detailed in \cite{YM-Lim01, YM-Lim02}.  An abelian Yang-Mills path integral was constructed in \cite{YM-Lim01}. The former details the construction of an infinite dimensional Gaussian probability space using Abstract Wiener space formalism, developed by Gross. See \cite{MR0461643}. The construction of an abelian Yang-Mills path integral will then allow us to construct a non-abelian Yang-Mills path integral in \cite{YM-Lim02}. The latter is instrumental in proving the mass gap.

These two prequals are technical in nature. By removing all the technical aspects of the construction, we hope that this article, will be more palatable to both physicists and mathematicians, who only want to understand a construction of a 4-dimensional quantum field, satisfying Wightman's axioms.

To further convince the reader that the construction of a 4-dimensional quantum field theory is correct, we will proceed to prove the Clustering Theorem, that will imply that the vacuum expectation has an exponential decay along a space-like separation, which implies that the force represented by the Lie Algebra $\mathfrak{g}$ is short ranged in nature. Because our construction of a 4-dimensional quantum theory for Yang-Mills fields satisfies a modified version of Wightman's axioms, we will give an alternative proof of the Clustering Theorem, that fits into our context.

\section{A Description of Quantum Hilbert space}\label{s.bdhs}

To construct our Hilbert space, we need a compact Lie group $G$, with its real simple Lie Algebra $\mathfrak{g}$, of which $\{E^\alpha\}_{\alpha=1}^N$ is an orthonormal basis using the inner product defined in Equation (\ref{e.i.2}), fixed throughout this article.

Extend the inner product defined in Equation (\ref{e.i.2}) to be a sesquilinear complex inner product, over the complexification of $\mathfrak{g}$, denoted as $\mathfrak{g}_{\bC} \equiv \mathfrak{g} \otimes_{\bR} \bC$. Hence, it is linear in the first variable, conjugate linear in the second.

A finite dimensional representation of $\mathfrak{g}$ is a Lie Algebra homomorphism $\rho$ of $\mathfrak{g}$ into ${\rm End}(\bC^{\tilde{N}})$ or into ${\rm End}(V)$ for some complex vector space $V$. All our representations will be considered to be non-trivial and irreducible. The dimension of the representation is given by $\tilde{N}$. Because $G$ is a compact Lie group, every finite dimensional representation of $G$ is equivalent to a unitary representation. See Theorems 9.4 and 9.5 in \cite{hall2015}.

Hence, we can always assume that our Lie Algebra representation $\rho: \mathfrak{g} \rightarrow {\rm End}(\bC^{\tilde{N}})$ is represented as skew-Hermitian matrices. Thus, the eigenvalues of $\rho(E)$ will be purely imaginary, $E \in \mathfrak{g}$.

Suppose each non-trivial irreducible representation $\rho_n: \mathfrak{g} \rightarrow {\rm End}(\bC^{\tilde{N}_n})$ is indexed by $n$, whereby $n \in \mathbb{N}$ and no two representations are equivalent. We will order them as described in subsection \ref{ss.mg}.

Introduce state 1, which will be referred to as the vacuum state, and $\{1\}$ refers to the linear span of $1$, over the complex numbers. Let $\mathscr{H}(\rho_0) := \{1\} $, which is an one dimensional complex inner product space, with a complex sesquilinear inner product defined by $\langle 1, 1 \rangle = 1$.

Define \beq \mathbb{H}_{{\rm YM}}(\mathfrak{g}) := \bigoplus_{n=0}^\infty \mathscr{H}(\rho_n), \label{e.h.2} \eeq whereby $\rho_n : \mathfrak{g} \rightarrow {\rm End}(\bC^{\tilde{N}_n})$, $n \geq 1$.  The inner product defined on this direct sum is given by \beq \left\langle \sum_{n=0}^\infty v_n, \sum_{n=0}^\infty u_n \right\rangle := \sum_{n=0}^\infty \langle v_n, u_n \rangle, \nonumber \eeq whereby $\langle v_n, u_n \rangle$ is the inner product defined on $\mathscr{H}(\rho_n)$.

We will now give the full description of each Hilbert space $\mathscr{H}(\rho)$, $\rho$ is an irreducible representation. Later, we will see that $\mathscr{H}(\rho)$ is an eigenspace for the momentum operator and Hamiltonian.

\subsection{Time-like and space-like surfaces}\label{ss.tsl}

\begin{notation}
We will let $I = [0,1]$ be the unit interval, and $I^2 \equiv I \times I$. Denote $\hat{s} = (s,\bar{s}),\ \hat{t} = (t,\bar{t})\in I^2$, $d\hat{s} \equiv ds d\bar{s}$, $d\hat{t} \equiv dt d\bar{t}$. Typically, $s, \bar{s}, t, \bar{t}$ will be reserved as the variable for some parametrization, i.e. $\rho: s \in I \mapsto \rho(s) \in \bR^4$.
\end{notation}

\begin{defn}\label{d.ts.1}(Time-like and space-like)\\
Let $S$ be a bounded rectangular surface in $\bR^4$, contained in some plane. By rotating the spatial axes if necessary, we may assume without any loss of generality, a parametrization of $S$ is given by \beq \left\{(a^0 + sb^0, a^1 , a^2 + sb^2, a^3 + tb^3 )^T \in \bR^4:\ s, t \in I\right\}, \label{e.d.2} \eeq for constants $a^\alpha, b^\alpha \in \bR$. Now, the surface $S$ is spanned by two directional vectors $(b^0, 0, b^2, 0)^T$ and $(0,0, 0, b^3)^T$. Note that $(b^0, 0, b^2, 0)^T$ lie in the $x^0-x^2$ plane and is orthogonal to $(0,0, 0, b^3)^T$.

We say a rectangular surface is space-like, if $|b^0| < |b^2|$, i.e. the acute angle which the vector $(b^0, 0, b^2, 0)^T$ makes with the $x^2$-axis in the $x^0-x^2$-plane is less than $\pi/4$.

We say a rectangular surface is time-like, if $|b^0| > |b^2|$, i.e. the acute angle which the vector $(b^0, 0, b^2, 0)^T$ makes with the $x^0$-axis in the $x^0-x^2$-plane is less than $\pi/4$.

Let $S$ be a rectangular surface in $\bR^4$ contained in some plane, and $TS$ denote the set of directional vectors that lie inside $S$.

Write $\vec{v} = (v^0, v) \equiv (v^0, v^1, v^2, v^3)^T \in TS$, and define $|v|^2 = v^{1,2} + v^{2,2} + v^{3,2}$. An equivalent way to say that $S$ is time-like is \beq \inf_{\vec{0} \neq \vec{v} \in TS} \frac{|v|^2}{v^{0,2}} < 1. \label{e.inf.1} \eeq

And we say that $S$ is space-like if \beq \inf_{\vec{0} \neq \vec{v} \in TS} \frac{|v|^2}{v^{0,2}} > 1. \label{e.inf.2} \eeq
\end{defn}

\begin{rem}
\begin{enumerate}
  \item By definition, a time-like surface must contain a time-like directional vector in it. Since under Lorentz transformation, a time-like vector remains time-like, we see that a time-like surface remains time-like under Lorentz transformation. Similarly, a surface is space-like means all its directional vectors in the surface are space-like. Under Lorentz transformation, all its directional vectors spanning $S$ remain space-like, hence a space-like surface remains space-like under Lorentz transformation.
  \item By a boost, any space-like rectangular surface contained in a plane can be transformed into a surface lying strictly inside $\{c\} \times \bR^3$, for some constant $c$.
  \item Recall $e_0$ spans the time-axis. By a boost, any time-like rectangular surface contained in a plane can be transformed into a surface which is spanned by $e_0$ and an orthogonal directional vector $v \in \bR^3$.
  \item Any two distinct points on a space-like surface is space-like separated; but two distinct points on a time-like surface may not be time-like separated.
\end{enumerate}
\end{rem}

When we say surface $S$ in this article, we mean it is a countable, disjoint union of rectangular surfaces in $\bR^4$, which are space-like, containing none, some or all of its boundary points.

\begin{defn}(Surface)\label{d.sur.1}\\
Any surface $S \equiv \{S_u\}_{u\geq 1} \subset \bR^4$ satisfies the following conditions:
\begin{itemize}
  \item each component $S_u$ is a space-like rectangular surface contained in some plane;
  \item each connected component $S_u$ may contain none, some or all of its boundary;
  \item $S_u \cap S_v = \emptyset$ if $u \neq v$;
  \item $S_u$ is contained in some bounded set in $\bR^4$.
\end{itemize}
\end{defn}

\begin{defn}\label{d.a.1}
Let $S_0$ be a compact rectangular space-like surface inside the $x^2-x^3$ plane. From Equation (\ref{e.d.2}), we see that any rectangular space-like surface $S$ contained in a plane, can be transformed to $S_0$ by Lorentz transformations and translation.

Recall $\{e_a\}_{a=0}^3$ is an orthonormal basis on $\bR^4$. We say that $\{\hat{f}_a\}_{a=0}^3$ is a Minkowski frame for a compact space-like surface $S$ contained in some plane, if there exists a sequence of Lorentz transformations $\Lambda_1, \cdots, \Lambda_n$ and a translation by $\vec{a} \in \bR^4$, such that
\begin{itemize}
  \item $S = \Lambda_n \cdots \Lambda_1 S_0 + \vec{a}$;
  \item $\hat{f}_a = \Lambda_n \cdots \Lambda_1 e_a \in \bR^4$, $a=0, \cdots, 3$.
\end{itemize}
\end{defn}

\begin{rem}
Observe that $\hat{f}_0$ is time-like and for $i=1,2,3$, $\hat{f}_i$ is space-like, satisfying the following properties:
\begin{itemize}
  \item $\hat{f}_a \cdot \hat{f}_b = 0$ if $a \neq b$; and
  \item $\hat{f}_0\cdot \hat{f}_0 = -1$ and $\hat{f}_i \cdot \hat{f}_i = 1$.
\end{itemize}
Note that $\{\hat{f}_2, \hat{f}_3\}$ spans $S$. Later we will let $S^\flat$ be a time-like plane, spanned by $\{\hat{f}_0, \hat{f}_1\}$. Clearly, $\{\hat{f}_a\}_{a=0}^3$ is a basis on $\bR^4$.
\end{rem}

Each component Hilbert space $\mathscr{H}(\rho)$ will consists of vectors of the form \newline $\sum_{u=1}^\infty\left(S_u, f_\alpha^u \otimes \rho(E^\alpha), \{\hat{f}_a^u\}_{a=0}^3\right)$, whereby $S= \bigcup_{u=1}^\infty S_u$ is some countable union of compact, rectangular surfaces in $\bR^4$, $f_\alpha^u$ will be some (measurable) complex-valued function, which is defined on the surface $S_u$. And $\{\hat{f}_a^u\}_{a=0}^3$ is a Minkowski frame for each compact rectangular surface $S_u$ contained in a plane as described in Definition \ref{d.a.1}.

Let $\sigma: [0,1] \times [0,1] \rightarrow \bR^4$ be a parametrization of a compact $S$. The complex-valued function $f_\alpha$ is measurable on $S$, if $f_\alpha \circ \sigma: [0,1]^2 \rightarrow \bC$ is measurable.

Each of these surfaces is assumed to be space-like, as defined above. We sum over repeated index $\alpha$, from $\alpha=1$ to $N$. One should think of $f_\alpha^u \otimes \rho(E^\alpha)$ as a section of the vector bundle $S_u \times \rho(\mathfrak{g})_\bC \rightarrow S_u$, defined over the surface $S_u$. In terms of the parametrization $\sigma$, the section at $\sigma(\hat{s})$ is given by $f_\alpha^u(\sigma(\hat{s}))\otimes \rho(E^\alpha)$.

\begin{rem}
Another way to view this vector $\left(S_u, f_\alpha^u \otimes \rho(E^\alpha), \{\hat{f}_a\}_{a=0}^3\right)$ is via taking its Fourier Transform into energy-momentum space. We will refer the reader to Section \ref{s.cluster}.
\end{rem}

Let $S$ and $\tilde{S}$ be rectangular space-like surfaces contained in a plane. Given complex scalars $\lambda$ and $\mu$, we define the addition and scalar multiplication as
\begin{align}
\lambda&\left(S, f_\alpha \otimes \rho(E^\alpha), \{\hat{f}_a\}_{a=0}^3\right) + \mu\left(\tilde{S}, g_\alpha \otimes \rho(E^\alpha), \{\hat{f}_a\}_{a=0}^3 \right) \nonumber \\
&:= \left(S \cup \tilde{S}, \left(\lambda\tilde{f}_\alpha + \mu\tilde{g}_\alpha\right) \otimes \rho(E^\alpha), \{\hat{f}_a\}_{a=0}^3\right). \label{e.v.1}
\end{align}
Here, we extend $f_\alpha$ to be $\tilde{f}_\alpha: S \cup \tilde{S} \rightarrow \bC$ by $\tilde{f}_\alpha(p) = f_\alpha(p)$ if $p \in S$; $\tilde{f}(p) = 0$ otherwise. Similarly, $\tilde{g}_\alpha$ is an extension of $g_\alpha$, defined as $\tilde{g}_\alpha(p) = g_\alpha(p)$, if $p \in \tilde{S}$, $\tilde{g}_\alpha(p)=0$ otherwise.

\begin{rem}
For the above addition to hold, we require that the Minkowski frame $\{\hat{f}_a\}_{a=0}^3$ on $S$ and $\tilde{S}$ to be identical.
\end{rem}

Given 2 surfaces, $S$ and $\tilde{S}$, we need to take the intersection and union of these surfaces. Now, the union of these 2 surfaces can always be written as a disjoint union of connected sets, each such set is a space-like surface, containing none, some or all of its boundary points. However, the intersection may not be a surface. For example, the two surfaces may intersect to give a line. In such a case, we will take the intersection to be the empty set $\emptyset$.

Given a surface $S$, let $\sigma$ be any parametrization of $S$. We can define $\int_Sd\rho$ using this parametrization $\sigma$ as given in Definition \ref{d.r.1}. Now, replace $\sigma \equiv (\sigma_0, \sigma_1, \sigma_2, \sigma_3)^T$ with $\acute{\sigma} = (i\sigma_0, \sigma_1, \sigma_2, \sigma_3)^T$ and hence define $\int_S d|\acute{\rho}|$ as given in Definition \ref{d.r.2}. With this, define the following inner product on $\mathscr{H}(\rho)$.

\begin{defn}\label{d.inn.2}
Given a surface $S= \bigcup_{u=1}^n S_u$ equipped with a collection of frames $\{\hat{f}_a^u: a=0,\cdots, 3\}_{u\geq 1}$, and a set of bounded and continuous complex-valued functions $\{f_\alpha^u: \alpha=1, \cdots, N\}_{u \geq 1}$ for $u \geq 1$, defined on $S $, form a vector \\
$\sum_{u=1}^n \left(S_u, f_\alpha^u \otimes \rho(E^\alpha), \{\hat{f}_a^u\}_{a=0}^3 \right)$, also referred to as Yang-Mills field. Note that for each $u$, $f_\alpha^u \otimes \rho(E^\alpha)$ is a section of $S_u \times [\rho(\mathfrak{g})\otimes \bC] \rightarrow S_u$, with $S_u$ contained in some plane, equipped with a Minkowski frame $\{\hat{f}_a^u\}_{a=0}^3$.

Let $V$ be a (complex) vector space containing such vectors, with addition and scalar multiplication defined in Equation (\ref{e.v.1}). The zero vector can be written as $\left(S, 0, \{\hat{f}_a\}_{a=0}^3 \right)$ for any space-like rectangular surface $S$ contained in a plane, equipped with any suitable Minkowski frame $\{\hat{f}_a\}_{a=0}^3$.

Refer to Definition \ref{d.r.2}. Assume that $S$ and $\tilde{S}$ be space-like surfaces, contained in a plane. Define an inner product $\langle \cdot, \cdot \rangle$ for $\left(S, f_\alpha \otimes \rho(E^\alpha), \{\hat{f}_a\}_{a=0}^3 \right) \in V$, and $\left(\tilde{S}, g_\beta \otimes \rho(E^\beta), \{\hat{g}_a\}_{a=0}^3 \right) \in V$, given by
\begin{align}
&\left\langle \left(S, f_\alpha \otimes \rho(E^\alpha) , \{\hat{f}_a\}_{a=0}^3\right), \left(\tilde{S}, g_\beta \otimes \rho(E^\beta), \{\hat{g}_a\}_{a=0}^3 \right) \right\rangle \nonumber \\
&:=\int_{S \cap \tilde{S}} [f_\alpha \overline{g_\beta}] \cdot d|\acute{\rho}| \cdot \Tr[-\rho(E^\alpha) \rho(E^\beta)]\label{e.i.3} \\
&\equiv   \sum_{\alpha=1}^N C(\rho)\int_{I^2}  [f_{\alpha}\cdot\overline{g_{\alpha}}](\sigma(\hat{s}))\left|
\sum_{0\leq a < b \leq 3}\acute{\rho}_\sigma^{ab}(\hat{s})\left[\det \acute{J}_{ab}^\sigma(\hat{s})\right]\right| d\hat{s}, \nonumber
\end{align}
provided $\hat{f}_a = \hat{g}_a$, for $a=0, \ldots , 3$. Otherwise, it is defined as zero.
Here, $\sigma: I^2 \rightarrow \bR^4$ is a parametrization of $S \cap \tilde{S}$ and $C(\rho)$ is defined later in Notation \ref{n.co.1}.

Denote its norm by $|\cdot|$. Let $\mathscr{H}(\rho)$ denote the Hilbert space containing $V$.
\end{defn}

\begin{rem}\label{r.ym.1}
The quantity $\int_S d\rho$ was first derived in \cite{YM-Lim01}, which motivates the definition of this quantity $\int_S d|\acute{\rho}|$. It was obtained from computing an abelian Yang-Mills path integral, by computing the average square of the flux of a Yang-Mills gauge field, over a time-like or space-like surface, using an infinite-dimensional Gaussian measure.

This quantity can also be derived from the Wilson Area Law formula in \cite{YM-Lim02}, computed using a non-abelian gauge group. The reader will later observe that this quantity $\int_S d|\acute{\rho}|$, plays an important role in all the 4 Wightman's axioms.

Using this quantity is one reason why we call the states as Yang-Mills fields, as the above inner product has its origins from a Yang-Mills path integral. There is also a second reason for calling them Yang-Mills fields. Refer to Remark \ref{r.f.1}.
\end{rem}

\begin{prop}
The Hilbert space $\mathscr{H}(\rho)$ is non-separable.
\end{prop}

\begin{proof}
Consider a compact rectangular surface $S_0$ contained in the $x^2-x^3$ plane, with $\{e_a\}_{a=0}^3$ as its Minkowski frame. Then, we see that \beq \left\{\left(S_0 + \vec{a}, \rho(E^\alpha), \{e_a\}_{a=0}^3\right):\ \vec{a} \in \mathbb{R} \right\} \nonumber \eeq is an uncountable set of orthogonal vectors in $\mathscr{H}(\rho)$, since \beq \left\langle \Big(S_0 + \vec{a}, \rho(E^\alpha), \{e_a\}_{a=0}^3\Big), \left(S_0 + \vec{b}, \rho(E^\alpha), \{e_a\}_{a=0}^3\right) \right\rangle = 0, \nonumber \eeq if $\vec{a} \neq \vec{b}$. Hence the Hilbert space is non-separable.
\end{proof}

\subsection{Unitary representation of inhomogeneous ${\rm SL}(2,\bC)$}\label{ss.ur}

Given a continuous group acting on $\bR^4$, we can consider its corresponding inhomogeneous group, whose elements are pairs consisting of a translation and a homogeneous transformation. For example, the Poincare group $\mathscr{P}$ containing the Lorentz group $L$, will have elements $\{\vec{a}, \Lambda\}$, where $\Lambda \in L$ and $\vec{a}$ will represent translation in the direction $\vec{a}$. The multiplication law for the Poincare group is given by \beq \{\vec{a}_1, A_1\}\{\vec{a}_2, A_2\} = \{\vec{a}_1 + A_1\vec{a}_2, A_1A_2\}. \nonumber \eeq
Associated with the restricted Lorentz group $L_+^\uparrow$ is the group of $2 \times 2$ complex matrices of determinant one, denoted by ${\rm SL}(2, \bC)$. There is an onto homomorphism $Y: {\rm SL}(2,\bC) \rightarrow L_+^\uparrow$. Thus, given $\Lambda \in {\rm SL}(2,\bC)$, $Y(\Lambda) \in L_+^\uparrow \subset L$. See \cite{streater}.

Instead of the Poincare group, we can consider the inhomogeneous ${\rm SL}(2,\bC)$ in its place, which we will also denote by ${\rm SL}(2, \bC)$, and use it to construct unitary representations. Its elements will consist of $\{\vec{a}, \Lambda\}$ and its multiplication law is given by \beq \{\vec{a}_1, \Lambda_1\}\{\vec{a}_2, \Lambda_2 \} = \{\vec{a}_1+Y(\Lambda_1)\vec{a}_2, \Lambda_1 \Lambda_2\}. \nonumber \eeq By abuse of notation, for any $\Lambda \in {\rm SL}(2,\bC)$, we will write $\Lambda \vec{a}$ to mean $\Lambda$ being represented as a $4 \times 4$ matrix, acting on $\vec{a} \in \bR^4$. This means we will write $\Lambda_1 \vec{a}_2 \equiv Y(\Lambda_1)\vec{a}_2$.

Given a vector $\vec{x} \in \bR^4$, $\{\vec{a}, \Lambda\}$ acts on $\vec{x}$ by $\vec{x} \mapsto \Lambda\vec{x} + \vec{a}$. By abuse of notation, for a surface $S$, $\{\vec{a}, \Lambda\}$ acts on $S$ by $S \mapsto \Lambda S + \vec{a}$, which means apply a Lorentz transformation $Y(\Lambda)$ to every position vector on the surface $S$, followed by translation in the direction $\vec{a}$. In terms of some parametrization $\sigma: I^2 \rightarrow \bR^4$ for $S$, the surface $\Lambda S + \vec{a}$ is parametrized by $\Lambda \sigma + \vec{a}: I^2 \rightarrow \bR^4$.

In general, the only finite dimensional unitary representation of ${\rm SL}(2, \bC)$ is the trivial representation. See Theorem 16.2 in \cite{knapp2016}. Thus, to construct a unitary representation, we must consider an infinite dimensional space. See \cite{MR1968551}.

\begin{defn}(Unitary Representation of the inhomogeneous ${\rm SL}(2,\bC)$)\\
Let $\hat{H}(\rho)$, $\hat{P}(\rho)$ be positive numbers, dependent on the representation $\rho$, to be defined later in Definition \ref{d.ma.1}. Let $\Lambda$ be in ${\rm SL}(2,\bC)$.

There is an unitary representation of the inhomogeneous ${\rm SL}(2,\bC)$, $\{\vec{a}, \Lambda\} \mapsto U(\vec{a}, \Lambda)$. Now, $U(\vec{a},\Lambda)$ acts on the Hilbert space $\mathscr{H}(\rho)$, \beq \left(S, f_\alpha \otimes \rho(E^\alpha), \{\hat{f}_a\}_{a=0}^3\right) \longmapsto U(\vec{a},\Lambda)\left(S, f_\alpha \otimes \rho(E^\alpha), \{\hat{f}_a\}_{a=0}^3\right), \nonumber \eeq by
\begin{align}
U&(\vec{a},\Lambda)\left(S, f_\alpha \otimes \rho(E^\alpha), \{\hat{f}_a\}_{a=0}^3\right) \nonumber \\
&:= \left(\Lambda S + \vec{a}, e^{-i[\vec{a} \cdot (\hat{H}(\rho_n)\Lambda\hat{f}_0  + \hat{P}(\rho_n)\Lambda\hat{f}_1 )]}
f_{\alpha}(\Lambda^{-1}(\cdot-\vec{a}))\otimes \rho(E^\alpha), \{\Lambda\hat{f}_a\}_{a=0}^3\right). \label{e.u.4}
\end{align}
Here, $S$ is a space-like surface contained in some plane.
\end{defn}

\begin{rem}
\begin{enumerate}
 \item Notice that $U(\vec{a}, \Lambda)$ acts trivially on $\rho(\mathfrak{g})$. In classical Yang-Mills equation, the fields over $\bR^4$, are $\mathfrak{g}$-valued. The Lie group $G$ describes the internal symmetry, on each fiber of the vector bundle. See \cite{GUT}. When we apply a Lorentz transformation, we expect that $G$ remains invariant under Lorentz transformation. After all, $G$ acts fiberwise on the vector bundle over $\bR^4$.
  \item Let us explain the formula on the RHS of Equation (\ref{e.u.4}). Suppose $\sigma: I^2 \rightarrow S$ is a parametrization for $S$. Then $\Lambda \sigma + \vec{a} \equiv Y(\Lambda)\sigma + \vec{a}$ will be a parametrization for $Y(\Lambda)S+ \vec{a}$. And, the field at the point $\vec{x} := Y(\Lambda)\sigma(\hat{s}) + \vec{a} \in Y(\Lambda) S + \vec{a}$, is given by
      \begin{align*}
      e^{-i[\vec{a}\cdot (\hat{H}(\rho_n)Y(\Lambda)\hat{f}_0 + \hat{P}(\rho_n)Y(\Lambda)\hat{f}_1)]}& f_\alpha[Y(\Lambda^{-1})(\vec{x}- \vec{a})] \otimes \rho(E^{\alpha}) \\
      &\equiv e^{-i[\vec{a}\cdot (\hat{H}(\rho_n)Y(\Lambda)\hat{f}_0 + \hat{P}(\rho_n)Y(\Lambda)\hat{f}_1)]} f_\alpha[\sigma(\hat{s})]\otimes \rho(E^\alpha).
      \end{align*}
  \item When there is no translation, the vector field over $Y(\Lambda) S$ is the pushforward of the vector field $f_\alpha \otimes \rho(E^\alpha)$ over $S$.
  \item When there is only translation, the unitary operator can be simplified to be
\begin{align*}
U(\vec{a},1)&\left(S, f_\alpha \otimes \rho(E^\alpha), \{\hat{f}_a\}_{a=0}^3\right)\\
\equiv& e^{-i[\vec{a} \cdot (\hat{H}(\rho_n)\hat{f}_0  + \hat{P}(\rho_n)\hat{f}_1 )]}\left( S + \vec{a}, f_\alpha(\cdot- \vec{a}) \otimes \rho\left(E^{\alpha}\right), \{\hat{f}_a\}_{a=0}^3\right),
\end{align*}
directly from Equation (\ref{e.u.4}).
\item In general, the directional derivative for $\left(S, f_\alpha \otimes \rho(E^\alpha), \{\hat{f}_a\}_{a=0}^3\right)$ does not exist for arbitrary direction $\vec{a}$. Thus, in the computation of the generators for translation in the $\hat{f}_0$ and $\hat{f}_1$ directions, the derivative does not appear. The author in \cite{hepph6330} also talked about problems when taking the derivative of a quantum field, at short distances. The differences in fields at different spatial points actually diverges as the separation gets smaller. The fields become infinitely rough at small distance scales and it means that it is impossible experimentally to probe the field at a single point.
\end{enumerate}
\end{rem}

It is straightforward to check that this map $\{\vec{a}, \Lambda\} \mapsto U(\vec{a}, \Lambda)$ is a representation, to be left to the reader.

\begin{lem}\label{l.b.2}
The map $U(\vec{a}, \Lambda)$ defined on $\mathscr{H}(\rho)$ is unitary, using the inner product $\langle \cdot, \cdot \rangle$ as defined in Definition \ref{d.inn.2}.
\end{lem}

\begin{proof}
It suffices to show for $S$ and $\tilde{S}$, both contained in some space-like plane respectively. Let $\{\hat{f}_a\}_{a=0}^3$, $\{\hat{g}_a\}_{a=0}^3$ be as defined in Definition \ref{d.a.1} for $\Lambda S$ and $\Lambda\tilde{S}$ respectively.

By Definition \ref{d.inn.2}, it suffices to prove when $\Lambda S \cap \Lambda\tilde{S}$ has non-zero area. We only consider the case when $\hat{f}_a = \hat{g}_a$, $a=0, \cdots, 3$, since the result is trivial otherwise. Thus,
\begin{align*}
&\left\langle U(\vec{a}, \Lambda)\left(S, f_\alpha \otimes \rho(E^\alpha), \{\hat{f}_a\}_{a=0}^3\right), U(\vec{a}, \Lambda)\left(\tilde{S}, g_\beta \otimes \rho(E^\beta), \{\hat{f}_a\}_{a=0}^3 \right) \right\rangle \\
&:=\int_{[\Lambda S + \vec{a}] \cap [\Lambda \tilde{S} + \vec{a}]}\ d|\acute{\rho}|\  e^{-i[\vec{a} \cdot (\hat{H}(\rho_n)\Lambda \hat{f}_0 + \hat{P}(\rho_n)\Lambda \hat{f}_1)]}e^{i[\vec{a} \cdot (\hat{H}(\rho_n)\Lambda \hat{f}_0 + \hat{P}(\rho_n)\Lambda \hat{f}_1)]}\\
&\hspace{5.5cm} \times [f_\alpha \overline{g_\beta}](\Lambda^{-1}(\cdot-\vec{a})) \cdot \Tr[-\rho(E^\alpha) \rho(E^\beta)] \\
&=\int_{\Lambda (S \cap \tilde{S})+ \vec{a}} [f_\alpha \overline{g_\beta}](\Lambda^{-1}(\cdot-\vec{a})) \cdot d|\acute{\rho}| \cdot \Tr[-\rho(E^\alpha) \rho(E^\beta)] \\
&=\int_{S \cap  \tilde{S}} [f_\alpha \overline{g_\beta}](\cdot) \cdot d|\acute{\rho}| \cdot \Tr[-\rho(E^\alpha) \rho(E^\beta)],
\end{align*}
after applying Lemma \ref{l.b.5}.
\end{proof}

We have just described the Hilbert space $\mathscr{H}(\rho)$. Now let us focus on the vacuum state, which we denoted it as 1. Recall our Yang-Mills fields are described by a triple $\sum_{u=1}^\infty \left(S_u, f_\alpha^u \otimes \rho(E^\alpha), \{\hat{f}_a^u\}_{a=0}^3 \right)$, whereby $S_u$ is some non-empty surface.

\begin{defn}
The vacuum state 1 is synonymous with the empty set $\emptyset$, i.e. $1 \equiv (\emptyset)$. We define $\langle 1, 1 \rangle := 1$.
\end{defn}

\begin{rem}\label{r.v.1}
\begin{enumerate}
  \item On the empty set, it does not make sense to have a measurable function or Minkowski frame defined on it.
  \item Recall that $\{1\}$ is a one-dimensional subspace spanned by the vacuum state. The Lie Algebra $\mathfrak{g}$ acts trivially on $\{1\}$, via the trivial representation $\rho_0: \mathfrak{g} \rightarrow \bC$. Refer to Equation (\ref{e.h.2}).
\end{enumerate}
\end{rem}

Clearly the empty set is invariant under the Poincare action. Hence, it is invariant under $U(\vec{a}, \Lambda)$.

Given a Schwartz function $g$ on $\bR^4$, we can define a quantum field operator $\phi^{\alpha,n}(g)$, acting on the vacuum state and $\sum_{u=1}^\infty\left(S, f_\alpha^u \otimes \rho_n(E^\alpha), \{\hat{f}_a^u\}_{a=0}^3\right)$ in Section \ref{s.qfo}. We will show later that $\phi^{\alpha,n}(g)$ is densely defined on $\mathbb{H}_{{\rm YM}}(\mathfrak{g})$. Furthermore, $\{\phi^{\alpha,n}(g):\ 1 \leq \alpha \leq \underline{N}\}$ defines a spinor of dimension $\underline{N}$. Using this definition and the definition of unitary transformation, one can prove the transformation law given in Proposition \ref{p.w.1}, and causality in Section \ref{s.c} using this action.

\begin{rem}
In \cite{glimm1981}, the function $g$ is interpreted as an observable, which is elevated to be a quantum field operator $\phi^{\alpha,n}(g)$. One can also understand it as applying canonical quantization to a classical field $g$.
\end{rem}

In our description of the Hilbert space containing Yang-Mills fields, nowhere did we use the Yang-Mills action, so it is not clear if $\mathbb{H}_{{\rm YM}}(\mathfrak{g})$ is a Hilbert space for a quantum Yang-Mills gauge theory.

The set $\{\hat{H}(\rho_n), \hat{P}(\rho_n)\}_{n \geq 1}$, indexed by the non-trivial irreducible, inequivalent representations, will give us a discrete set of eigenvalues. The choices of $\hat{H}(\rho)$ and $\hat{P}(\rho)$ will be given later, referred to as the eigenvalues of the Hamiltonian $\hat{H}$ and momentum operator $\hat{P}$ respectively. Indeed, these discrete eigenvalues will give us a countable spectrum for the translation operator $U(\vec{a}, 1)$, provided $\vec{a} = a^0 \hat{f}_0 + a^1 \hat{f}_1$. There are many choices of $\{\hat{H}(\rho), \hat{P}(\rho)\}$ and it is not clear how we should choose these numbers.

It is only in Wightman's zeroth axiom, where the mass of the theory appears. This axiom requires that $\hat{H}(\rho)^2 - \hat{P}(\rho)^2 = m(\rho)^2$, for some mass gap $m(\rho) \geq 0$.\footnote{See Remark \ref{r.m.1}.} Furthermore, these eigenvalues are required to be unbounded. The mass gap problem is equivalent to show that $m_0 := \inf_{n \in \mathbb{N}} m(\rho_n) > 0$, and this is only true for a compact simple gauge group.

For each eigenstate in $\mathscr{H}(\rho_n)$, $\hat{H}$ and $\hat{P}$ will be multiplication by scalars $\hat{H}(\rho_n)$ and $\hat{P}(\rho_n)$ respectively. Then $\hat{H}^2 - \hat{P}^2 = m^2$ translates to \beq \frac{\hat{P}(\rho_n)^2}{\hat{H}(\rho_n)^2}-1 = -\frac{m(\rho_n)^2}{\hat{H}(\rho_n)^2}. \label{e.mg.3} \eeq To prove that the Hamiltonian and momentum operators are unbounded, and there is a positive mass gap $m_0$, it suffices to show that for all $n \geq 1$, \beq 0 > \frac{\hat{P}(\rho_n)^2}{\hat{H}(\rho_n)^2}-1 \longrightarrow 0, \nonumber \eeq as $n \rightarrow \infty$, and $\lim_{n \rightarrow \infty}m(\rho_n) = \infty$.

To show that the operators are unbounded and the existence of a positive mass gap, we need to make use of the Yang-Mills path integral to quantize, so that we will obtain Equation (\ref{e.mg.3}). During the quantization process, we will use renormalization techniques and asymptotic freedom. Note that asymptotic freedom only holds for a non-abelian gauge group. The compact simple gauge group will give us a quadratic Casimir operator, dependent on the representation $\rho_n$. As the set containing all Casimir operators, each corresponding to a non-equivalent irreducible representation of $\mathfrak{g}$, is countably infinite and unbounded from above, the Hamiltonian will be shown to be an unbounded operator.

The existence of a positive mass gap is a consequence of Equation (\ref{e.mg.1}), first proved in \cite{YM-Lim02}. To prove this equation, we need
\begin{itemize}
  \item renormalization techniques,
  \item asymptotic freedom,
  \item the compactness of the gauge group allows us to represent the Lie Algebra as skew-Hermitian matrices,
  \item the structure constants and the quadratic Casimir operator of the simple Lie Algebra, and
  \item the quartic term in the Yang-Mills action,
\end{itemize}
all of which are collectively responsible for the existence of a mass gap. We also need to impose the Callan-Symanzik Equation to prove the existence of a mass gap $m_0$.

\begin{rem}
The idea that the quartic term in the Yang-Mills action might be responsible for the mass gap was suggested in \cite{fey1981479}.
\end{rem}

Using a Yang-Mills path integral to define the Hamiltonian and momentum operator eigenvalues justify our construction as a 4-dimensional Yang-Mills quantum gauge theory. In Section \ref{s.cluster}, we will show how the positive mass gap, will imply the Clustering Theorem.

As for the rest of the axioms, we see that it does not make use of the Yang-Mills action. We will postpone the proof of the mass gap till Section \ref{s.iop}. For now, we will move on to the remaining Wightman's axioms.

\section{Quantum Field Operators}\label{s.qfo}

We will now begin our discussion on the field operators that act on the Hilbert space $\bigoplus_{n=0}^\infty \mathscr{H}(\rho_n)$, which contains our Yang-Mills fields.

Every irreducible finite dimensional representation of ${\rm SL}(2, \bC)$ is denoted by $D^{(j, k)}$, $j, k$ are non-negative integers or half integers. This representation is known as the spinor representation of ${\rm SL}(2, \bC)$. Under this representation, one sees that $D^{(j,k)}(-1) = (-1)^{2(j+k)}$. Incidently, when restricted to ${\rm SU}(2)$, $\Lambda \in {\rm SU}(2) \mapsto D^{(j,0)}(\Lambda)$ is equivalent to an irreducible representation of ${\rm SU}(2)$. See \cite{streater}.

Recall we have a Minkowski frame $\{\hat{f}_a\}_{a=0}^3$ defined on a rectangular space-like surface $S$ in Definition \ref{d.a.1}. Now, $\hat{f}_a = \Lambda e_a \equiv Y(\Lambda)e_a$ for some $\Lambda \in {\rm SL}(2, \bC)$. Suppose we have another Minkowski frame $\{\hat{g}_a\}_{a=0}^3$ on $S$ such that $\hat{f}_a = \hat{g}_a$ for all $a$. Then, we have that $Y(\Lambda)e_a = Y(\tilde{\Lambda})e_a$ for some $\tilde{\Lambda} \in {\rm SL}(2, \bC)$. Hence, $Y(\Lambda) = Y(\tilde{\Lambda})$, which can be shown that $\Lambda = \pm \tilde{\Lambda}$. When $j+k$ is an integer, we see that $D^{(j,k)}(\pm1) = 1$, thus $D^{(j,k)}$ is a representation describing vector bosons.

\begin{defn}(Test functions)\\
We let $\mathscr{P}$ denote the Schwartz space consisting of infinitely differentiable complex-valued functions on $\bR^4$, which converge to 0 at infinity faster than any powers of $|\vec{x}|$. We will refer $f \in \mathscr{P}$ as a test function, which is bounded.
\end{defn}

\begin{notation}
For $\vec{k} = (k^0, k^1, k^2, k^3)$, $k^a \in \{0\} \cup \mathbb{N}$, we will write \beq D^{\vec{k}} = \left(\frac{\partial}{\partial x^0}\right)^{k^0}\left(\frac{\partial}{\partial x^1}\right)^{k^1}\left(\frac{\partial}{\partial x^2}\right)^{k^2}\left(\frac{\partial}{\partial x^3}\right)^{k^3}, \quad \vec{x}^{\vec{k}} = (x^0)^{k^0}(x^1)^{k^1}(x^2)^{k^2}(x^3)^{k^3}. \nonumber \eeq And $|\vec{k}| = \sum_{a=0}^3 |k^a|$.
\end{notation}

\begin{defn}(Norm on $\mathscr{P}$)\label{d.n.1}\\
Let $r,s$ be whole numbers. Suppose $f \in \mathscr{P}$. With the above notation, define a norm $\parallel \cdot \parallel_{r,s}$ on $\mathscr{P}$ as \beq \parallel f \parallel_{r,s} := \sum_{ |\vec{k}| \leq r}\sum_{|\vec{l}| \leq s}\sup_{\vec{x}\in \bR^4} |\vec{x}^{\vec{k}} D^{\vec{l}}f(\vec{x})|. \nonumber \eeq
\end{defn}

\subsection{Creation operators}

\begin{defn}\label{d.p.1}(Time-like plane)\\
Refer to Definition \ref{d.a.1}. Let $S$ be a connected space-like rectangular surface contained in some plane.

Let $S^\flat$ be a time-like plane spanned by $\{\hat{f}_0, \hat{f}_1\}$, parametrized by \beq  \hat{s}= (s,\bar{s}) \mapsto \sigma(s,\bar{s}) = s\hat{f}_0 + \bar{s}\hat{f}_1 ,\ s, \bar{s} \in \bR. \label{e.p.1} \eeq And we will write $d\hat{s} = ds d\bar{s}$.

Refer to Definition \ref{d.r.2}. Define using a Minkowski metric, $\tilde{\eta}: \vec{v}\in \bR^4 \mapsto \vec{v} \cdot (\hat{H}(\rho_n) \hat{f}_0 + \hat{P}(\rho_n)\hat{f}_1) \in \bR$. Suppose we are given a $\tilde{f} \in \mathscr{P}$. We will define a new function $\tilde{f}^{\{\hat{f}_0, \hat{f}_1\}}: \bR^4 \rightarrow \bC$ by
\begin{align}
\vec{x}\in \bR^4 &\longmapsto \tilde{f}^{\{\hat{f}_0, \hat{f}_1\}}(\hat{H}(\rho_n), \hat{P}(\rho_n))(\vec{x}) :=
\int_{S^\flat} \frac{e^{-i\tilde{\eta}(\cdot)}}{2\pi}\tilde{f}(\vec{x} + \cdot)\ d|\acute{\rho}| \nonumber \\
&=\int_{\hat{s} \in \bR^2}\frac{e^{-i[\sigma(\hat{s}) \cdot (\hat{H}(\rho_n) \hat{f}_0 + \hat{P}(\rho_n)\hat{f}_1)]}}{2\pi}\tilde{f}\left( \vec{x} + \sigma(\hat{s}) \right)\cdot |\acute{\rho}_\sigma|(\hat{s})\ d\hat{s}, \label{e.p.3}
\end{align}
integration over a time-like plane $S^\flat$, using the parametrization given in Equation (\ref{e.p.1}), for representation $\rho_n$.
\end{defn}

\begin{rem}\label{r.f.3}
\begin{enumerate}
  \item Note that $\tilde{f}^{\{\hat{f}_0, \hat{f}_1\}} \notin \mathscr{P}$, unless $\tilde{f} \equiv 0$.
  \item Even though we compute the integral using a given parametrization $\sigma$, it is actually independent of the parametrization and only depends on the time-like plane span by $\{\hat{f}_0, \hat{f}_1\}$.
      Because we are doing a Fourier Transform on time-like and one space-like variable, evaluated at $\hat{H}(\rho_n)\hat{f}_0 + \hat{P}(\rho_n)\hat{f}_1$, we see that the transformed function depends on $\{\hat{f}_0, \hat{f}_1\}$, not just on $S^\flat$.
  \item\label{r.f.3c} If $\vec{x} = \sum_{a=0}^3x^a\hat{f}_a$, then
  \begin{align*}
  \tilde{f}^{\{\hat{f}_0, \hat{f}_1\}}&(\hat{H}(\rho_n), \hat{P}(\rho_n))(\vec{x}) \\
  &= e^{-i[x^0\hat{H}(\rho_n) - x^1\hat{P}(\rho_n)]}\tilde{f}^{\{\hat{f}_0, \hat{f}_1\}}(\hat{H}(\rho_n), \hat{P}(\rho_n))(x^2\hat{f}_2 + x^3\hat{f}_3).
  \end{align*}
  Hence, $\hat{H}(\rho_n)$ ($\hat{P}(\rho_n)$) is the generator for translation, in the $\hat{f}_0$ ($\hat{f}_1$) direction.
\end{enumerate}
\end{rem}

\begin{defn}\label{d.ac.2}
Let $\mathcal{A} := \{F^\alpha: 1 \leq \alpha \leq \underline{N}\} \subset \mathfrak{g}$ be a finite set, and define recursively for $n \geq 2$, $\mathcal{A}^n := [\mathcal{A}^{n-1}, \mathcal{A}]$, $\mathcal{A}^1 := \mathcal{A}$, such that $\mathfrak{g}$ can be spanned by $\bigcup_{j=1}^{\underline{n}}\mathcal{A}^j$ for some $\underline{n} \geq 1$.
\end{defn}

Given $\tilde{f} \in \mathscr{P}$, we now wish to describe the field operator $\phi^{\alpha,n}(\tilde{f})$, $1 \leq \alpha \leq \underline{N}$. Suppose we have a spinor representation $A: {\rm SL}(2, \bC) \rightarrow {\rm End}(\bC^{\underline{N}})$. For each $n \in \mathbb{N}$, the field operator $\{\phi^{\alpha,n}(\tilde{f}):\ 1\leq \alpha \leq \underline{N}\}$ transforms like a spinor under the action of $A(\Lambda)$, for any $\Lambda \in {\rm SL}(2, \bC)$. Hence, a Schwartz function $\tilde{f}$ will be `promoted' to be some spinor $\sum_{\alpha=1}^{\underline{N}}c_\alpha \phi^{\alpha, n}(\tilde{f})$. But first, how does $\phi^{\alpha,n}(\tilde{f})$ act on $1$, for some Schwartz function $\tilde{f}$?

\begin{defn}\label{d.co.1}(Creation operators)\\
Recall we indexed our non-trivial irreducible representation by a natural number $n$. Fix a connected space-like plane $S_0 \subset \bR^4$. We will choose $S_0$ to be the $x^2-x^3$ plane. Using spatial rotation, we can rotate $S_0$ to be the $x^i-x^j$ plane, for $i,j =1, 2, 3$. Together with translation and boost, we can transform any surface contained inside $S_0$, to be any space-like surface, using the unitary representation of $SL(2, \bC)$. By Definition \ref{d.a.1}, we will choose $\{e_a\}_{a=0}^3$ to be a Minkowski frame for $S_0$.

For any $\tilde{f} \in \mathscr{P}$, we define an operator $\phi^{\alpha,n}(\tilde{f})$, $\alpha=1,2, \cdots, \underline{N}$, $n \in \mathbb{N}$, which acts on the vacuum state 1 by
\begin{align*}
\phi^{\alpha,n}(\tilde{f}) 1 :=& \left(S_0,   \tilde{f}^{\{e_0, e_1\}}  \otimes \rho_n(F^\alpha), \{e_a\}_{a=0}^3\right) \\
\equiv&\left(S_0,   \tilde{f}^{\{e_0, e_1\}}(\hat{H}(\rho_n), \hat{P}(\rho_n))  \otimes \rho_n(F^\alpha), \{e_a \}_{a=0}^3\right) \in \mathscr{H}(\rho_{n}),
\end{align*}
for $F^\alpha \in \mathcal{A} \subset \mathfrak{g}$. See Definition \ref{d.ac.2}.

In the notation above, it is understood that $\tilde{f}^{\{e_0, e_1\}} \equiv \tilde{f}^{\{e_0, e_1\}}(\hat{H}(\rho_n), \hat{P}(\rho_n))$. And we restrict the domain of
$\tilde{f}^{\{e_0, e_1\}}(\hat{H}(\rho_n), \hat{P}(\rho_n)): \bR^4 \rightarrow \bC$ to be on the surface $S_0$.
\end{defn}

\begin{rem}
The field operators can be indexed by a countably infinite set. See \cite{Ruelle}.
\end{rem}

\subsection{Domain and continuity}

\begin{defn}
Let $\mathscr{S}$ denote the set consisting of countable union of space-like rectangular surfaces in $\bR^4$, each component surface is compact or a plane. For $S \in \mathscr{S}$, let $\mathscr{P}_S$ denote the set of Schwartz functions defined on $S$.

In particular, if $\sigma: I^2 \rightarrow \bR^4$ is a parametrization of a compact space-like surface $S$, we say $f$ is a Schwartz function on $S$ when $f\circ \sigma$ is an infinitely differentiable function on $I^2$, and is compactly supported in the interior of $I^2$, i.e. it decays to zero at the boundary.

If $S$ is a space-like plane, then for some parametrization $\sigma: \bR^2 \rightarrow S$, we say $f$ is a Schwartz function on $S$ when $f\circ \sigma$ is a Schwartz function on $\bR^2$.
\end{defn}

\begin{defn}(Domain of field operators)\\
Define a domain $\mathscr{D} \subset  \bigoplus_{n=0}^\infty \mathscr{H}(\rho_n)$ as \beq \left\{a_01 + \sum_{n,u=1}^\infty \left(S_{n,u}, f_{n,\alpha}^u\otimes \rho_{n}(E^{\alpha}), \{\hat{f}_a^{n,u}\}_{a=0}^3 \right):\ a_0 \in \bC,\ f_{n,\alpha}^u \in \mathscr{P}_{S_{n,u}},\ S_{n,u} \in \mathscr{S}  \right\}. \nonumber \eeq
\end{defn}

Given a parametrization $\sigma: I^2 \rightarrow \bR^4$ for a surface $S$, we say that $f \in L^2(S)$ if $f$ is measurable on $S$ and \beq \int_{I^2} |f\circ \sigma|^2(\hat{s}) \left|
\sum_{0\leq a < b \leq 3}\acute{\rho}_\sigma^{ab}(\hat{s})\left[\det \acute{J}_{ab}^\sigma(\hat{s})\right]\right| d\hat{s} < \infty. \nonumber \eeq
By construction of $\mathscr{H}(\rho)$, we only consider space-like surfaces, equipped with a Minkowski frame.  For any surface $S \in \mathscr{S}$, $\mathscr{P}_S$ is dense inside $L^2(S)$. So we see that $\mathscr{D}$ is actually a dense set inside $\mathbb{H}_{{\rm YM}}(\mathfrak{g})$, and it contains the vacuum state.

\begin{rem}
In the proof of Proposition \ref{p.d.1}, we will show that $\mathscr{P}_S$ is dense inside $L^2(S)$, when $S$ is compact.
\end{rem}

From Equation (\ref{e.d.2}), a compact space-like surface in the $x^2-x^3$ plane can be transformed to any compact space-like surface under translation, spatial rotation or boost. Thus, $\mathscr{S}$ remains invariant under the action of ${\rm SL}(2, \bC)$. From Definition \ref{d.a.1}, a Minkowski frame associated with a rectangular surface, is generated by Lorentz transformations of $\{e_a\}_{a=0}^3$. Hence $U(\vec{a}, \Lambda)\mathscr{D} \subset \mathscr{D}$.

We can now define the field operator, whose domain is given by $\mathscr{D}$, as follows. Recall from Definition \ref{d.p.1}, how we can define a new function $\tilde{f}^{\{\hat{f}_0, \hat{f}_1\}}$, from $\tilde{f} \in \mathscr{P}$, using $\{\hat{f}_0, \hat{f}_1\}$ contained in a Minkowski frame, associated with a space-like surface $S$.

From the opening paragraph in Section \ref{s.qfo}, we saw that $\{\hat{f}_a\}_{a=0}^3$ uniquely determines $\Lambda \in {\rm SL}(2, \bC)$, up to $\pm 1$. When $j+k$ is an integer, we see that $D^{(j, k)}(\pm \Lambda) = D^{(j, k)}( \Lambda)$.

\begin{defn}\label{d.co.2}
Let $\sum_{u=1}^\infty\left(S_u, g_\beta^u \otimes \rho(E^\beta), \{\hat{f}_a^u\}_{a=0}^3\right) \in \mathscr{H}(\rho)$, whereby each $S_u$ is a connected rectangular space-like surface contained inside some plane, and $g_\beta^u \in \mathscr{P}_{S_u}$. Refer to Definition \ref{d.p.1}. For each Minkowski frame $\{\hat{f}_a^u\}_{a=0}^3$, let $\Lambda^u e_a = \hat{f}_a^u$ for some $\Lambda^u \in {\rm SL}(2, \bC)$. The adjoint representation $\ad$ of $\rho(\mathfrak{g})$ is defined as  $\ad(\rho(E^\alpha))\rho(E^\beta) = [\rho(E^\alpha), \rho(E^\beta)]$.

Recall from Definition \ref{d.ac.2}, we defined a set $\mathcal{A} \subset \mathfrak{g}$ with cardinality $\underline{N}$. Suppose we have a spinor representation $A: {\rm SL}(2, \bC) \rightarrow {\rm End}(\bC^{\underline{N}})$. This representation can be written as a sum of irreducible representations of the form $\bigoplus_{\alpha=1}^m D^{(j_\alpha, k_\alpha)}$, $D^{(j_\alpha, k_\alpha)}$ is an irreducible representation of ${\rm SL}(2, \bC)$ as described earlier. We will further assume that for each $1 \leq \alpha \leq m$, $j_\alpha+k_\alpha$ is an integer.

For $\Lambda \in {\rm SL}(2, \bC)$, let $A(\Lambda)_\alpha^\beta$ denote the entry at the $\beta$-th row, $\alpha$-th column. Given a test function $f \in \mathscr{P}$, define a field operator $\phi^{\alpha,n}(f)$ as
\begin{align*}
\phi^{\alpha,n}&(f)\sum_{u=1}^\infty \left(S_u, g_\beta^u \otimes \rho_m(E^\beta), \{\hat{f}_a^u\}_{a=0}^3\right) \\
&:=
\left\{
  \begin{array}{ll}
    \sum_{u=1}^\infty\left(S_u, f^{\{\hat{f}_0^u, \hat{f}_1^u\}}_n A(\Lambda^u)_\gamma^\alpha\cdot g_\beta^u \otimes  \rho_{n}\left([F^\gamma, E^\beta]\right), \{\hat{f}_a^u\}_{a=0}^3 \right), & \hbox{$m=n$;} \\
    0, & \hbox{$m\neq n$.}
  \end{array}
\right.
\end{align*}
There is an implied sum over $\gamma$ from 1 to $\underline{N}$, and over $\beta$ from 1 to $N$. Here, $f^{\{\hat{f}_0^u, \hat{f}_1^u\}}_n \equiv f^{\{\hat{f}_0^u, \hat{f}_1^u\}}(\hat{H}(\rho_n), \hat{P}(\rho_n))$ and the restricted vector field $f^{\{\hat{f}_0^u, \hat{f}_1^u\}}_n\Big|_{S_u} A(\Lambda^u)_\gamma^\alpha \otimes \rho_n(F^\gamma)$ acts on $g_\beta^u \otimes \rho_n(E^\beta)$, by
$\ad\left( f^{\{\hat{f}_0^u, \hat{f}_1^u\}}_n|_{S_u} A(\Lambda^u)_\gamma^\alpha\otimes \rho_n(F^\gamma)\right)$, over the surface $S_u$ fiberwise, $\ad(\rho_n(E))$ is the adjoint representation of $\rho_n(E)$.
\end{defn}

\begin{rem}\label{r.ub.1}
\begin{enumerate}
 \item The alpha in $\phi^{\alpha,n}(f)$ is referred to as the spinor index. Indeed, for each $n \in \mathbb{N}$ and $f \in \mathscr{P}$, we have that \beq \{\phi^{\alpha,n}(f):\ 1 \leq \alpha \leq \underline{N}\} \nonumber \eeq is a spinor of dimension $\underline{N}$. If $A = D^{(s,0)}$, then we must have that $s$ is an integer, i.e. the spinors are vector bosons in this case.
 \item\label{i.m.1} When we have $\{e_a\}_{a=0}^3$ as a Minkowski frame, then $\Lambda$ is the identity matrix. Thus, $A(\Lambda)$ will be the identity.
 \item By linearity, we define
  \begin{align*}
  \phi^{\alpha,n}(f) \sum_{m=0}^\infty v_m &:= \sum_{m=0}^\infty \phi^{\alpha,n}(f) v_m \\
  &= a_0\Big( S_0, f^{\{e_0, e_1\}}_n \otimes \rho_n(E^\alpha), \{e_a\}_{a=0}^3 \Big) +  \phi^{\alpha,n}(f) v_n ,
  \end{align*}
  $v_0= a_01$ is a scalar multiple of the vacuum state and $v_n \in \mathscr{H}(\rho_n)$. The domain for $\phi^{\alpha,n}(f)$ will be $\mathscr{D}$. Note that it is a bounded operator.
  \item Suppose $g_\beta$ is measurable, or $L^2$ integrable on a space-like rectangular surface $S$, contained in some plane. Since $f^{\{\hat{f}_0, \hat{f}_1\}}$ is bounded and continuous, we see that $f^{\{\hat{f}_0, \hat{f}_1\}}\cdot g_\beta$ is measurable and $L^2$ integrable. So,
      \begin{align*}
      \phi^{\alpha,n}(f)&\left(S, g_\beta \otimes \rho_n(E^\beta), \{\hat{f}_a\}_{a=0}^3\right) \\
      &:= \left(S, f^{\{\hat{f}_0, \hat{f}_1\}}_n\cdot A(\Lambda)_\gamma^\alpha g_\beta \otimes \rho_n([F^\gamma, E^\beta]), \{\hat{f}_a\}_{a=0}^3\right),
      \end{align*}
      and $f^{\{\hat{f}_0, \hat{f}_1\}}_n\cdot g_\beta$ is defined almost everywhere on $S$. But we will run into problems later, when proving Proposition \ref{p.td.1}. This is because multiplying a measurable function with a tempered distribution, may not be a tempered distribution. It is thus necessary to restrict the domain for the field operators to be on $\mathscr{D}$, to avoid technical difficulties later on.
\end{enumerate}
\end{rem}

We showed how $\phi^{\alpha,n}(f)$ is defined on $\mathscr{D}$. We can now define its adjoint.

\begin{defn}(Annihilation operators)\label{d.ado}\\
Using the inner product in Definition \ref{d.inn.2}, we define the adjoint $\phi^{\alpha,n}(g)^\ast$ on a space-like surface $S$ contained in some plane, as
\begin{align*}
\phi^{\alpha,n}(g)^\ast &\left(S, f_\beta \otimes \rho_{n}(E^\beta), \{\hat{f}_a\}_{a=0}^3 \right) \\
=& - \left(S, \overline{g^{\{\hat{f}_0, \hat{f}_1\}}} \overline{A(\Lambda)_\gamma^\alpha}\cdot f_\beta \otimes  \rho_n([F^\gamma,E^\beta]), \{\hat{f}_a\}_{a=0}^3 \right) \\
&+ \left\langle \left(S, f_\beta \otimes \rho_{n}(E^\beta), \{\hat{f}_a\}_{a=0}^3\right), \phi^{\alpha,n}(g) 1 \right\rangle 1,
\end{align*}
whereby $\hat{f}_a = \Lambda e_a$ for $a=0, \cdots, 3$, and we sum over repeated indices $\gamma$ from 1 to $\underline{N}$, and over $\beta$ from 1 to $N$. This is because \beq \left\langle \ad(\rho(E^\alpha)) \rho(E^\beta), \rho(E^\gamma) \right\rangle = -\left\langle
\rho(E^\beta), \ad(\rho(E^\alpha)) \rho(E^\gamma) \right\rangle. \nonumber \eeq

And
\begin{itemize}
  \item $\phi^{\alpha,n}(g)^\ast$ will send $\left(S, f_\beta \otimes \rho_{m}(E^\beta), \{\hat{f}_a\}_{a=0}^3 \right)$ to 0 if $m \neq n$;
  \item $\phi^{\alpha,n}(g)^\ast 1 = 0$.
\end{itemize}
\end{defn}

\begin{rem}
Clearly, we can choose the domain to be $\mathscr{D}$.
\end{rem}

\begin{example}
Consider the Lie group ${\rm SU}(2)$. Its Lie algebra can be generated by the three Pauli matrices. Suppose $\mathcal{A}$ as described in Definition \ref{d.ac.2} is linearly independent. Hence, $2\leq |\mathcal{A}| \leq 3$. If $A = D^{(j,k)}$, then $j =1$, $k=0$, since $j+k$ must be an integer. Thus, $|\mathcal{A}|=3$ and the Pauli matrices, which represent $W^\pm$ and $Z$ bosons responsible for weak force interactions, transform like spin 1 vector bosons.
\end{example}

\begin{example}\label{ex.f.1}
Consider the Lie group ${\rm SU}(3)$. Its Lie algebra can be generated by the Gell-Mann matrices
\begin{align*}
\lambda_1 =&
\left(
  \begin{array}{ccc}
    0 &\ 1 &\ 0 \\
    -1 &\ 0 &\ 0 \\
    0 &\ 0 &\ 0 \\
  \end{array}
\right), \quad \lambda_2 =
\left(
  \begin{array}{ccc}
    0 &\ i &\ 0 \\
    i &\ 0 &\ 0 \\
    0 &\ 0 &\ 0 \\
  \end{array}
\right), \quad \lambda_3 =
\left(
  \begin{array}{ccc}
    i &\ 0 &\ 0 \\
    0 &\ -i &\ 0 \\
    0 &\ 0 &\ 0 \\
  \end{array}
\right), \\
\lambda_4 =&
\left(
  \begin{array}{ccc}
    0 &\ 0 &\ i \\
    0 &\ 0 &\ 0 \\
    i &\ 0 &\ 0 \\
  \end{array}
\right), \quad
\lambda_5 =
\left(
  \begin{array}{ccc}
    0 &\ 0 &\ 1 \\
    0 &\ 0 &\ 0 \\
    -1 &\ 0 &\ 0 \\
  \end{array}
\right), \\
\lambda_6 =&
\left(
  \begin{array}{ccc}
    0 &\ 0 &\ 0 \\
    0 &\ 0 &\ i \\
    0 &\ i &\ 0 \\
  \end{array}
\right), \quad \lambda_7 =
\left(
  \begin{array}{ccc}
    0 &\ 0 &\ 0 \\
    0 &\ 0 &\ 1 \\
    0 &\ -1 &\ 0 \\
  \end{array}
\right), \quad \lambda_8 = \frac{1}{\sqrt3}
\left(
  \begin{array}{ccc}
    i &\ 0 &\ 0 \\
    0 &\ i &\ 0 \\
    0 &\ 0 &\ -2i \\
  \end{array}
\right),
\end{align*}
each representing the gluons responsible for strong force interaction.

Observe that $\{\lambda_1, \lambda_2, \lambda_3\}$ is the Lie algebra of a subgroup $H_1 \subset {\rm SU}(3)$, isomorphic to ${\rm SU}(2)$. Furthermore, $\lambda_8$ generates an abelian subgroup $H_2$, such that the elements in $H_1$ commutes with elements in $H_2$, because $\lambda_8$ commutes with $\{\lambda_1, \lambda_2, \lambda_3\}$. Let $\tilde{H}$ be a Lie subgroup in ${\rm SU}(3)$, generated by $\{\lambda_1, \lambda_2, \lambda_3, \lambda_8\}$.

Suppose the vacuum state $1$ is invariant under this unbroken subgroup $\tilde{H}$ and $\mathcal{A}$ as described in Definition \ref{d.ac.2} is linearly independent. This means that $3 \leq |\mathcal{A}|\leq 4$ and cannot contain any elements that are linear combination of $\{\lambda_1, \lambda_2, \lambda_3, \lambda_8\}$, also referred to as unbroken generators. Therefore, this ${\rm SU}(3)$ gauge theory is spontaneously broken. Furthermore, if $A = D^{(j,k)}$, then either $j = 1$, $k = 0$ or $j = k = 1/2$.

Choose $\mathcal{A} = \{\lambda_4, \lambda_5, \lambda_7\}$. A direct computation shows that
\begin{align*}
[\lambda_7, \lambda_4] =& \lambda_2, \quad [\lambda_7, \lambda_5] = \lambda_1, \quad
\ad(\lambda_4)\ad(\lambda_5) \lambda_7 = -\lambda_6, \\
\ad(\lambda_7)\ad(\lambda_4)\ad(\lambda_5) \lambda_7 =& \lambda_3 - \sqrt3 \lambda_8, \quad [\lambda_5, \lambda_4] = \lambda_3 + \sqrt3 \lambda_8.
\end{align*}
Thus, we see that $\bigcup_{j=1}^4 \mathcal{A}^j$ spans $\mathfrak{su}(3)$. In this case, we can choose $A = D^{(1,0)}$ as an irreducible representation and the vectors in the span of $\mathcal{A}$ are called spin 1 vectors.

If we choose $\mathcal{A} = \{\lambda_4, \lambda_5, \lambda_6, \lambda_7\}$, which are the broken generators in $\mathfrak{su}(3)$, then $A = D^{(1/2,1/2)}$ is equivalent to $Y$, and the vectors in the span of $\mathcal{A}$ transform like 4-vectors.
\end{example}

\subsection{Cyclicity}

\begin{notation}
Let $S = [0,1] \times [0,1] \equiv I^2$ and $\mathscr{P}_S(n) \subset \mathscr{P}_S$, whereby $g \in  \mathscr{P}_S(n)$ if $g = f_1 f_2 \cdots f_n$, each $f_i \in \mathscr{P}_S$.

Let $C_c(S, \bC)$ and $C(S, \bC)$ denote the set of compactly supported continuous functions and the set of continuous functions on $S$ respectively.

We also write $\parallel \cdot \parallel_{L^2}$ to denote the $L^2$ norm on $S$, i.e. $\parallel f \parallel_{L^2} = \left[\int_{I^2} |f(\hat{s})|^2\ d\hat{s}\right]^{1/2}$.
\end{notation}

\begin{lem}\label{l.b.6}
We have that $\mathscr{P}_S$ is dense in $L^2(S)$. Let $\mathcal{G}(n)$ be the smallest algebra containing $\mathscr{P}_S(n)$, $n \in \mathbb{N}$ fixed. Then,
$\mathcal{G}(n)$ is dense in $\mathscr{P}_S$ using the $L^2$ norm. Hence, $\mathcal{G}(n)$ is dense in in $L^2(S)$.
\end{lem}

\begin{proof}
Now $\mathscr{P}_S$ is a complex algebra and clearly it separates interior points in $S$. Unfortunately, it does not contain the unit 1 on $S \equiv I^2$. Let $\mathcal{C}$ be the smallest algebra containing $1$ and $\mathscr{P}_S$.

By complex Stone Weierstrass Theorem, $\mathscr{P}_S$ and $\mathcal{C}$ are respectively dense in $C_c(S, \bC)$ and $C(S, \bC)$. Furthermore, since continuous functions are dense in $L^2(S)$, we see that $\mathcal{C}$ will be dense in $L^2(S)$. To show that the space containing polynomials of functions in $\mathscr{P}_S$ will generate $L^2(S)$, we will show that we can approximate 1 via the $L^2$ norm, using a sequence of functions in $\mathscr{P}_S$.

Define $\breve{\varphi}_\delta: [0,1] \rightarrow \bR$ by \beq \breve{\varphi}_\delta(t):= \left\{
  \begin{array}{ll}
    \frac{1}{\delta}t, & \hbox{$0 \leq t \leq \delta$;} \\
    1, & \hbox{$\delta < t \leq 1-\delta$;} \\
    \frac{1}{\delta}(1 - t) , & \hbox{$1-\delta < t \leq 1$.}
  \end{array}
\right. \nonumber \eeq Then, $\hat{t} = (t, \bar{t}) \in I^2 \mapsto \phi_\delta(\hat{t}):= \breve{\varphi}_\delta(t)\breve{\varphi}_\delta(\bar{t})$ is continuous. Let $\epsilon > 0$. We can find a $\delta > 0$ such that $\parallel \phi_\delta - 1 \parallel_{L^2} < \epsilon/2$.

Since $\phi_\delta \in C_c(S, \bC)$, we can find a $0 \leq g_\epsilon \in \mathscr{P}_S$ such that $\parallel \phi_\delta - g_\epsilon\parallel_{L^2} < \epsilon/2$. Thus, $\parallel 1 - g_\epsilon\parallel_{L^2} < \epsilon$. This proves that $\mathscr{P}_S$ is dense in $L^2(S)$.

To prove the third statement, let $f \in \mathscr{P}_S$ and let $\epsilon > 0$. There exists a $M > 0$ such that $|f|(\hat{s}) < M$ for all $\hat{s} \in S$. Choose a $\delta > 0$ such that $\parallel 1 - g_\delta^{n-1} \parallel_{L^2} < \epsilon/M$, $g_\delta \in \mathscr{P}_S$. Let $\tilde{g}_\epsilon := g_\delta^{n-1}f$. Then,
\begin{align*}
\parallel f - \tilde{g}_\epsilon \parallel_{L^2} &= \left[\int_{I^2} |1 - g_\delta^{n-1}|^2(\hat{t}) |f|^2(\hat{t})\ d\hat{t} \right]^{1/2} \\
&\leq M \parallel 1  - g_\delta^{n-1} \parallel_{L^2} < \epsilon.
\end{align*}
Thus, $\mathcal{G}(n)$ is dense in $\mathscr{P}_S$ using the $L^2$ norm.
\end{proof}

Because $\mathscr{P}_S$ is an algebra, we see that \beq \psi^{\alpha_1,m_1}(g_1)\cdots \psi^{\alpha_k, m_k}(g_k) \mathscr{D} \subset \mathscr{D}, \nonumber \eeq whereby $\psi^{\alpha_i,k_i}(g_i) = \phi^{\alpha_i, k_i}(g_i)$ or its adjoint $\phi^{\alpha_i, k_i}(g_i)^\ast$.

Let $\mathscr{D}_0$ be a subspace inside $\mathbb{H}_{{\rm YM}}(\mathfrak{g})$, generated by the action of $U(\vec{a}, \Lambda)$ and polynomials containing $\phi^{\alpha_1,n}(f_1), \cdots ,\phi^{\alpha_r,n}(f_r)$, acting on the vacuum state 1. Clearly, $\mathscr{D}_0 \subset \mathscr{D}$.

\begin{prop}\label{p.d.1}
The set $\mathscr{D}_0$ is dense inside $\mathbb{H}_{{\rm YM}}(\mathfrak{g})$.
\end{prop}

\begin{proof}
Under $U(\vec{a}, \Lambda)$, we can transform any space-like rectangular surface, into other space-like rectangular surface, via translation, spatial rotation and boost. See Equation (\ref{e.d.2}). Thus it suffices to show that for any compact space-like rectangular surface $S \subset \bR^4$ contained in some plane, we can approximate any measurable section in $S \times \rho_n(\mathfrak{g})_\bC \rightarrow S$, for each fixed $n$, using the field operators. Without loss of generality, we assume that $S$ is $I^2$, lying inside the $x^2-x^3$ plane.

Refer to Definition \ref{d.ac.2}. First we assume that the span of $\mathcal{A}$ is $\mathfrak{g}$. In this case, we see that $\underline{N} \geq N$. Recall $\ad(E)F = [E, F] \in \mathfrak{g}$. Since $\mathfrak{g}$ is simple, we see that for a fixed $k \geq 1$, \beq {\rm span}\ \left\{\ad( F^{\alpha_1})\cdots \ad( F^{\alpha_k})F^{\beta}:\ 1\leq \beta\leq \underline{N},\ 1\leq \alpha_i \leq \underline{N},\ i=1,2, \cdots, k \right\} = \mathfrak{g}. \nonumber \eeq

Therefore, for each $1 \leq \gamma \leq N$, we can write \beq E^\gamma = \sum_{\beta=1}^{N(\gamma)} d_{m,\beta}^\gamma\ad( F^{\alpha_1^{\gamma, \beta}})\cdots \ad( F^{\alpha_{m-1}^{\gamma, \beta}})F^{\alpha_m^{\gamma, \beta}},
\nonumber \eeq for real coefficients $d_{m,\beta}^\gamma$ and natural numbers $1 \leq \alpha_i^{\gamma, \beta} \leq \underline{N}$.

Let $p_\kappa(x) = \frac{\kappa}{\sqrt{2\pi}}e^{-\kappa^2 x^2/2}$ be an one-dimensional Gaussian function, mean 0, variance $1/\kappa^2$. Its Fourier Transform is $\hat{p}_\kappa(q) = \kappa p_{1/\kappa}(q)$. Let $c_n = \hat{p}_1(\hat{H}(\rho_n)) \neq 0$, $d_n = \hat{p}_1(\hat{P}(\rho_n)) \neq 0$ be fixed.

For a given set of Schwartz functions $\{f_1, \cdots, f_m\}$ defined on $S$, extend each one to be a function \beq F_i \in \mathscr{P}: \vec{x} \in \bR^4 \mapsto \frac{1}{c_n}p_1(x^0)\frac{1}{d_n}p_1(x^1)f_i(x^2, x^3),\ i=1, \cdots, m. \nonumber \eeq We also have \beq F_i^{\{e_0, e_1\}}(\hat{H}(\rho_n), \hat{P}(\rho_n))(0, 0, x^2, x^3) = \frac{1}{c_n}\hat{p}_1(\hat{H}(\rho_n))\frac{1}{d_n}\hat{p}_1(\hat{P}(\rho_n))f_i(x^2, x^3) = f_i(x^2, x^3). \nonumber \eeq

Hence,
\begin{align*}
[F_1^{\{e_0, e_1\}}& \cdot \cdots F_m^{\{e_0, e_1\}}](\hat{H}(\rho_n), \hat{P}(\rho_n))(0, 0, x^2, x^3)
= \prod_{i=1}^m f_i(x^2, x^3) .
\end{align*}

For any set of Schwartz functions $\{f_1, \cdots, f_m\}$ on $S$, we can extend them to be Schwartz functions on $\bR^4$ as described above, and thus we have
\begin{align*}
\sum_{\beta=1}^{N(\gamma)} d_{m,\beta}^\gamma \phi^{\alpha_1^{\gamma,\beta}, n}(F_1)\cdots \phi^{\alpha_m^{\gamma, \beta}, n}(F_m)1
&= \left(S, \prod_{i=1}^m f_i \otimes \rho_n(E^\gamma), \{e_a\}_{a=0}^3 \right).
\end{align*}
See Item \ref{i.m.1} in Remark \ref{r.ub.1}.

If we let \beq C_m = {\rm span}\ \left\{ \phi^{\alpha_1,n}(F_1)\cdots \phi^{\alpha_m,n}(F_m)1:\ F_i\in \mathscr{P},\ 1\leq \alpha_i \leq \underline{N} \right\},  \nonumber \eeq we see that the sum of subspaces, $\sum_{m=1}^\infty C_m $, is dense in $L^2(S) \otimes \rho_n(\mathfrak{g})$. This follows from Lemma \ref{l.b.6}.

This will show that we can find a sequence of vectors in the sum $\sum_{m=1}^\infty C_m $ and approximate any vector of the form $\left(S, f_\alpha \otimes \rho_n(E^\alpha), \{e_a\}_{a=0}^3 \right)$, whereby $f_\alpha$ is $L^2$ integrable on a compact rectangular surface $S$ inside $x^2-x^3$ plane.

Now suppose span of $\mathcal{A}$ is not equal to $\mathfrak{g}$. Thus, $\underline{n} \geq 2$.
By definition of $\mathcal{A}$, we see that \beq {\rm span}\ \left\{\ad( F^{\alpha_1})\cdots \ad( F^{\alpha_{\tilde{n}-1}})F^{\alpha_{\tilde{n}}}:\ 1\leq \alpha_i \leq \underline{N},\ 1 \leq \tilde{n} \leq \underline{n} \right\} = \mathfrak{g}. \nonumber \eeq Thus, for each $1 \leq \gamma \leq N$, we can write for some $1 \leq \tilde{n} \leq \underline{n}$, \beq \tilde{E}^\gamma = \sum_{\xi=1}^{N(\gamma)} d_{\tilde{n}}(\gamma,\xi)\ad( F^{\alpha_1(\gamma,\xi)})\cdots \ad( F^{\alpha_{\tilde{n}-1}(\gamma,\xi)})
F^{\alpha_{\tilde{n}}(\gamma,\xi)}, \nonumber \eeq
for real coefficients $d_{\tilde{n}}(\gamma,\xi)$ and natural numbers $1 \leq \alpha_i(\gamma,\xi) \leq \underline{N}$. Here, $\{ \tilde{E}^\alpha:\ 1 \leq \alpha \leq N\}$ is a basis containing unit vectors for $\mathfrak{g}$.

Let $\epsilon > 0$ and $f \in \mathscr{P}_S$. From the proof of Lemma \ref{l.b.6}, we can find $g_1, \cdots, g_{\tilde{n}} \in \mathscr{P}_S$ such that \beq \parallel f - g_1\cdots g_{\tilde{n}}\parallel_{L^2} < \frac{\epsilon}{C(\rho_n)}. \nonumber \eeq Using an earlier argument, there exists $G_1, \cdots, G_{\tilde{n}}$ such that $G_i^{\{e_0, e_1\}} \equiv g_i$ for $1 \leq i \leq \tilde{n}$.

Hence, we have that
\begin{align*}
\sum_{\xi=1}^{N(\gamma)}& d_{\tilde{n}}(\gamma,\xi) \phi^{\alpha_1(\gamma,\xi),n}(G_1)\cdots \phi^{\alpha_{\tilde{n}-1}(\gamma,\xi),n}(G_{\tilde{n}-1})
\phi^{\alpha_{\tilde{n}}(\gamma,\xi),n}
(G_{\tilde{n}})1 \\
=& \left( S, \prod_{i=1}^{\tilde{n}}g_i \otimes \rho_n(\tilde{E}^\gamma), \{e_a\}_{a=0}^3 \right).
\end{align*}

A direct computation will show that \beq \left| \Big( S, f \otimes \rho(\tilde{E}^\gamma), \{e_a\}_{a=0}^3 \Big) - \left( S, \prod_{i=1}^{\tilde{n}}g_i \otimes \rho_n(\tilde{E}^\gamma), \{e_a\}_{a=0}^3 \right) \right| < \epsilon. \nonumber \eeq This completes the proof.
\end{proof}

\begin{rem}\label{r.f.1}
Given any $\left(S, f_\alpha \otimes \rho(E^\alpha), \{\hat{f}_a\}_{a=0}^3\right)$, we can find a sequence of \\ Schwartz functions $\{g_\alpha^m: \bR^4 \rightarrow \bC\ |\ m \in \mathbb{N}\}$, for which its partial Fourier Transform approximates $f_\alpha$, for each $1\leq \alpha\leq N$. Here, for each $g_\alpha^m: \bR^4 \rightarrow \bC$, we can take the Fourier Transform on the time-like and space-like variables, i.e. $g_\alpha^m \mapsto g_\alpha^{m,\{\hat{f}_0, \hat{f}_1\}}(\hat{H}(\rho), \hat{P}(\rho))(\vec{x})$, $\vec{x} \in S$. As explained at the end of subsection \ref{ss.ur}, the eigenvalues $\{\hat{H}(\rho), \hat{P}(\rho) \}$ will be defined using a Yang-Mills path integral in subsection \ref{ss.mg}. Therefore, $\left(S, f_\alpha \otimes \rho(E^\alpha), \{\hat{f}_a\}_{a=0}^3\right)\in \mathscr{H}(\rho)$ can be written as the limit of $\left\{\left(S, g_\alpha^{m, \{\hat{f}_0, \hat{f}_1\}} \otimes \rho(E^\alpha), \{\hat{f}_a\}_{a=0}^3\right)\right\}_{m=1}^\infty$, using the inner product given in Definition \ref{d.inn.2}, and hence they are referred to as Yang-Mills fields.
\end{rem}

\subsection{Tempered Distribution}

\begin{prop}\label{p.td.1}
Let $S, \tilde{S}$ be bounded space-like surfaces lying inside some plane. Let $\left(S, g_\beta \otimes \rho_n(E^\beta), \{\hat{f}_a\}_{a=0}^3\right), \left(\tilde{S}, \tilde{g}_\beta \otimes \rho_n(E^\beta), \{\hat{f}_a\}_{a=0}^3 \right) \in \mathscr{D}$. Suppose $\hat{f}_a = \Lambda e_a$ for $\Lambda \in {\rm SL}(2, \bC)$.

Write \beq c_\alpha^{\gamma, \beta} = A(\Lambda)_\delta^\alpha\Tr\left[-[\ad(\rho_n(F^\delta)) \rho_n(E^\gamma)] \rho_n(E^\beta)\right], \nonumber \eeq with an implied sum over $\delta$. Given a test function $f \in \mathscr{P}$, we define for a smooth parametrization $\sigma: I^2 \rightarrow S \cap \tilde{S} \subset \bR^4$,
\begin{align*}
T(f) :=& \left\langle \phi^{\alpha,n}(f)\left(\tilde{S}, \tilde{g}_\gamma \otimes \rho_n(E^\gamma), \{\hat{f}_a\}_{a=0}^3\right), \left(S, g_\beta \otimes \rho_n(E^\beta), \{\hat{f}_a\}_{a=0}^3\right)\right\rangle \\
=& c_\alpha^{\gamma, \beta} \int_{I^2} d\hat{t}\ \left[f^{\{\hat{f}_0, \hat{f}_1\}}\cdot \tilde{g}_\gamma \cdot \overline{g_\beta}\right](\sigma(\hat{t}))
\cdot|\acute{\rho}_\sigma|(\hat{t}).
\end{align*}
Then $T$ is a linear functional on $\mathscr{P}$. Furthermore, it is a tempered distribution.
\end{prop}

\begin{proof}
It is clear that it is a linear functional on $\mathscr{P}$. It remains to show that it is a tempered distribution.

Let $\sigma: I^2 \rightarrow \bR^4$ be a parametrization of $\hat{S} := S \cap \tilde{S}$, and write \beq h  =   [\tilde{g}_\gamma\cdot\overline{g_\beta}] \circ \sigma \cdot \left|\acute{\rho}_{\sigma}  \right|c_\alpha^{\gamma, \beta}
. \nonumber \eeq Then,
\begin{align*}
T(f) =& \int_{I^2} f^{\{\hat{f}_0, \hat{f}_1\}}(\sigma(\hat{t}))h(\hat{t})\ d\hat{t}.
\end{align*}

Let $\tilde{\sigma} : \hat{s} \mapsto \tilde{\sigma}(\hat{s}) = s\hat{f}_0 + \bar{s}\hat{f}_1$, $s,\bar{s} \in \bR$, whereby $\{\hat{f}_a\}_{a=0}^3$ is a Minkowski frame for $\hat{S}$. By definition,
\begin{align*}
\vec{x} \longmapsto f^{\{\hat{f}_0, \hat{f}_1\}}(\vec{x}) &\equiv f^{\{\hat{f}_0, \hat{f}_1\}}(\hat{H}(\rho_n), \hat{P}(\rho_n))(\vec{x}) \\
&= \int_{\hat{s}\in \bR^2} \frac{e^{-i[\tilde{\sigma}(\hat{s}) \cdot (\hat{H}(\rho_n)\hat{f}_0 + \hat{P}(\rho_n)\hat{f}_1)]}}{2\pi}f\left(\vec{x} + \tilde{\sigma}(\hat{s})\right)\left|\acute{\rho}_{\tilde{\sigma}}\right|(\hat{s})\ d\hat{s}.
\end{align*}

Write $\vec{\alpha} = \hat{H}(\rho_n)\hat{f}_0 + \hat{P}(\rho_n)\hat{f}_1$. Thus
\begin{align*}
T(f) =& \int_{I^2} f^{\{\hat{f}_0, \hat{f}_1\}}(\sigma(\hat{t}))h(\hat{t})\ d\hat{t} \\
=& \int_{\hat{s} \in \bR^2, \hat{t} \in I^2}d\hat{t}d\hat{s}\ f(\sigma(\hat{t}) + \tilde{\sigma}(\hat{s}))  \left|\acute{\rho}_{\tilde{\sigma}}\right|(\hat{s})
\left|\acute{\rho}_{\sigma}  \right|(\hat{t})\cdot \frac{e^{-i[\tilde{\sigma}(\hat{s}) \cdot \vec{\alpha}]}}{2\pi} [\tilde{g}_\gamma\cdot\overline{g_\beta}] \circ \sigma (\hat{t}) \cdot c_\alpha^{\gamma, \beta}.
\end{align*}

Note that
\begin{align*}
(\hat{s}, \hat{t})\in \bR^2 \times I^2 &\longmapsto \left|\acute{\rho}_{\tilde{\sigma}}\right|(\hat{s})
\left|\acute{\rho}_{\sigma}  \right|(\hat{t})\cdot \frac{e^{-i[\tilde{\sigma}(\hat{s}) \cdot \vec{\alpha}]}}{2\pi}
[\tilde{g}_\gamma\cdot\overline{g_\beta}] \circ \sigma (\hat{t}) \cdot c_\alpha^{\gamma, \beta},
\end{align*}
is a tempered distribution, because $\{g_\beta, \tilde{g}_\gamma\}_{\beta, \gamma=1}^N$ are Schwartz functions on $\hat{S}$. Therefore, $T$ is a tempered distribution.
\end{proof}

\begin{rem}
The map
\begin{align*}
f &\in \mathscr{P} \\
&\longmapsto \int_{\hat{s}\in \bR^2, \hat{t} \in I^2}d\hat{t}d\hat{s}\ f(\sigma(\hat{t}) + \tilde{\sigma}(\hat{s})) \left|\acute{\rho}_{\tilde{\sigma}}\right|(\hat{s})
\left|\acute{\rho}_{\sigma}  \right|(\hat{t})\cdot \frac{e^{-i[\tilde{\sigma}(\hat{s}) \cdot \vec{\alpha}]}}{2\pi}
[\tilde{g}_\gamma\cdot\overline{g_\beta}] \circ \sigma (\hat{t})  \\
&\hspace{8.4cm}\times \Tr\left[-[\ad(\cdot) \rho_n(E^\gamma)] \rho_n(E^\beta)\right],
\end{align*}
defines a $\rho_n(\mathfrak{g})$-valued distribution, using the inner product on $\rho_n(\mathfrak{g})$ defined in Equation (\ref{e.ck.2}).
\end{rem}

\begin{cor}
Suppose $S = \tilde{S} = S_0$. For any $f \in \mathscr{P}$, we can write
\begin{align*}
&T(f) \\
&= \int_{\bR^4}f(\vec{x})\frac{e^{i[x^0 \hat{H}(\rho_n)- x^1\hat{P}(\rho_n)]}}{2\pi}[\tilde{g}_\gamma \cdot \overline{g_\beta}](x^2, x^3)\ d\vec{x} \cdot \Tr\left[-[\ad(\rho_n(F^\alpha)) \rho_n(E^\gamma)] \rho_n(E^\beta)\right].
\end{align*}
\end{cor}

\begin{proof}
We note that $\{e_a\}_{a=0}^3$ is the default Minkowski frame. In this case, $\Lambda = \pm1 \in {\rm SL}(2, \bC)$, and $A(\pm 1) = 1$. Furthermore, we can choose the parametrizations \\
$\sigma(x^2, x^3) = x^2e_2 + x^3e_3$, $\tilde{\sigma}(x^0, x^1) = x^0 e_0 + x^1 e_1$. A direct calculation shows $|\acute{\rho}_\sigma|= |\acute{\rho}_{\tilde{\sigma}}| = 1$. The result hence follows.
\end{proof}

\begin{rem}
In this corollary, we see that the function that maps \beq \vec{x}=(x^0, x^1, x^2, x^3) \longmapsto \frac{e^{i[x^0 \hat{H}(\rho_n)- x^1\hat{P}(\rho_n)]}}{2\pi}[\tilde{g}_\gamma \cdot \overline{g_\beta}](x^2, x^3), \nonumber \eeq is not in $\mathscr{P}$, but rather it is a tempered distribution.
\end{rem}

\section{Transformation Law of the Field Operator}

Recall $\mathfrak{g}$ has dimension $N$ and has an irreducible representation $\rho: \mathfrak{g} \rightarrow {\rm End}(\bC^{\tilde{N}})$. Without any loss of generality, consider a space-like surface $S$, contained inside some plane. From Definition \ref{d.a.1}, we have a Minkowski frame $\{\hat{f}_a\}_{a=0}^3$ assigned to it.

In Definition \ref{d.ac.2}, we have a finite set $\mathcal{A} \subset \mathfrak{g}$, which defines a spinor indexed by $1 \leq \alpha \leq \underline{N}$. This spinor transforms under the action $A(\Lambda) : \phi^{\alpha,n} \mapsto A(\Lambda)_\beta^\alpha\phi^{\beta,n}$ for $\Lambda \in {\rm SL}(2,\bC)$.

For $\Lambda \in {\rm SL}(2, \bC)$, we consider $\Lambda^{-1}( S - \vec{a})$, which is also a space-like surface, with $\{\hat{g}_a\}_{a=0}^3 = \{\Lambda^{-1} \hat{f}_a\}_{a=0}^3$ assigned as a Minkowski frame to it by Definition \ref{d.a.1}.

\begin{prop}\label{p.w.1}
We have the transformation law for the field operators acting on $\mathscr{H}(\rho)$, i.e.
\begin{align*}
U(\vec{a},\Lambda)&\phi^{\alpha,n}(f) U(\vec{a},\Lambda)^{-1}\left( S,  g_\beta  \otimes \rho_n(E^\beta), \{\hat{f}_a\}_{a=0}^3\right)  \\
&= A(\Lambda^{-1})_\gamma^\alpha\phi^{\gamma,n}(f(\Lambda^{-1}(\cdot - \vec{a}) ))\left( S,  g_\beta  \otimes \rho_n(E^\beta), \{\hat{f}_a\}_{a=0}^3\right),
\end{align*}
whereby $S$ is some rectangular space-like surface contained in some plane.

Recall $S_0^\flat$ is the $x^0-x^1$ plane. However,
\begin{align*}
U&(\vec{a},\Lambda)\phi^{\alpha,n}(f) U(\vec{a},\Lambda)^{-1}1 \\
&=
\left(\Lambda S_0 + \vec{a}, e^{-i[\vec{a} \cdot (\hat{H}(\rho_n)\hat{g}_0 + \hat{P}(\rho_n)\hat{g}_1)]}
f^{\{\hat{g}_0, \hat{g}_1\}}(\Lambda^{-1}(\cdot-\vec{a}))\otimes \rho_n(F^\alpha), \{\hat{g}_a\}_{a=0}^3 \right).
\end{align*}
Here, $\{\hat{g}_a\}_{a=0}^3 = \{\Lambda e_a\}_{a=0}^3$ is a Minkowski frame for $\Lambda S_0$.
\end{prop}

\begin{proof}
By Definition \ref{d.a.1}, there is a $\tilde{\Lambda} \in {\rm SL}(2, \bC)$ such that $\hat{f}_a = Y(\tilde{\Lambda})e_a$. Thus, $\{\Lambda^{-1}\tilde{\Lambda} e_a\}_{a=0}^3 = \{\Lambda^{-1}\hat{f}_a\}_{a=0}^3$ is a Minkowski frame for $\Lambda^{-1}(S - \vec{a})$.

Write $d_\beta^\alpha = A(\Lambda^{-1}\tilde{\Lambda})_\beta^\alpha \equiv A(\Lambda^{-1})_\gamma^\alpha A(\tilde{\Lambda})^\gamma_\beta$,
\begin{align*}
T(\rho_n, \vec{a}) &= e^{-i[\vec{a} \cdot (\hat{H}(\rho_n)\hat{f}_0 + \hat{P}(\rho_n)\hat{f}_1)]} , \quad
T(\rho_n, \vec{a})^{-1} = e^{i[\vec{a} \cdot (\hat{H}(\rho_n)\hat{f}_0 + \hat{P}(\rho_n)\hat{f}_1)]} ,
\end{align*}
and $f^{\mathcal{C}} = f^{\{\Lambda^{-1}\hat{f}_0, \Lambda^{-1}\hat{f}_1\}}$, with
$\mathcal{C} = \{\Lambda^{-1}\hat{f}_0, \Lambda^{-1}\hat{f}_1\}$, $\mathcal{D} = \{\Lambda^{-1}\hat{f}_a\}_{a=0}^3$.

We have
\begin{align}
U&(\vec{a},\Lambda)\phi^{\alpha,n}(f) U(\vec{a},\Lambda)^{-1}\left( S,  g_\beta  \otimes \rho_n(E^\beta), \{\hat{f}_a\}_{a=0}^3\right) \nonumber \\
=& U(\vec{a},\Lambda) \phi^{\alpha,n}(f) \Big( \Lambda^{-1}(S - \vec{a}), T(\rho_n, \vec{a})^{-1} g_\beta(\Lambda\cdot + \vec{a})  \otimes \rho_n\left(E^\beta\right), \mathcal{D} \Big)  \nonumber \\
=& U(\vec{a},\Lambda)\Big( \Lambda^{-1}(S - \vec{a}), [T(\rho_n, \vec{a})^{-1}f^{\mathcal{C}}](\cdot) d_\gamma^\alpha \cdot g_\beta(\Lambda\cdot + \vec{a})  \otimes \ad(\rho_n(F^\gamma)) \rho_n\left(E^\beta\right), \mathcal{D}\Big) \nonumber \\
=& \Big(S,  T(\rho_n, \vec{a})T(\rho_n, \vec{a})^{-1} f^{\mathcal{C}}(\Lambda^{-1}(\cdot - \vec{a}))d_\gamma^\alpha g_\beta(\cdot)\otimes \ad(\rho_n(F^\gamma)) \rho_n\left(E^\beta\right), \{\hat{f}_a\}_{a=0}^3\Big). \label{e.t.2}
\end{align}

Let $[\Lambda^{-1}(S - \vec{a})]^\flat$ be the span of  $\{\Lambda^{-1}\hat{f}_0, \Lambda^{-1}\hat{f}_1\}$. By definition, for any point $\vec{x} \in S$, we have that \beq f^{\mathcal{C}}(\Lambda^{-1}(\vec{x} - \vec{a})) \equiv f^{\{\Lambda^{-1}\hat{f}_0, \Lambda^{-1}\hat{f}_1\}}(\Lambda^{-1}(\vec{x} - \vec{a})), \nonumber \eeq means we do a partial integration using Equation (\ref{e.p.3}) on $f$ in the time-like plane $\vec{y} + [\Lambda^{-1}(S - \vec{a})]^\flat$, for $\vec{y} = \Lambda^{-1}(\vec{x} - \vec{a})$.

Let $\sigma(s,\bar{s}) = s\hat{f}_0 + \bar{s}\hat{f}_1$, $s,\bar{s} \in \bR$. Let $\hat{g}_0=\Lambda^{-1}\hat{f}_0$ and $\hat{g}_1=\Lambda^{-1}\hat{f}_1$. We will let $\hat{\sigma}(\hat{s}) = \Lambda^{-1}\sigma(\hat{s}) = s\hat{g}_0 + \bar{s}\hat{g}_1$, $s,\bar{s} \in \bR$.

Now, $\Lambda^{-1} \sigma \cdot \Lambda^{-1} \hat{f}_a = \sigma \cdot \hat{f}_a$ and $\acute{\rho}_{\hat{\sigma}} = \acute{\rho}_{\sigma}$. See Remark \ref{r.a.3}.  Hence,
\begin{align*}
&f^{\{\hat{g}_0, \hat{g}_1\}}(\hat{H}(\rho_n), \hat{P}(\rho_n))(\Lambda^{-1}(\vec{x} - \vec{a})) \\
&:= \int_{\hat{s}\in \bR^2}\frac{e^{-i[\hat{\sigma}(\hat{s}) \cdot (\hat{H}(\rho_n)\hat{g}_0 + \hat{P}(\rho_n)\hat{g}_1)]}}{2\pi}f(\vec{y} + \hat{\sigma}(\hat{s}))|\acute{\rho}_{\hat{\sigma}}|(\hat{s})\ d\hat{s} \\
&= \int_{\hat{s}\in \bR^2}\frac{e^{-i[\sigma(\hat{s}) \cdot (\hat{H}(\rho_n)\hat{f}_0 + \hat{P}(\rho_n)\hat{f}_1)]}}{2\pi}f(\vec{y} + \Lambda^{-1}\sigma(\hat{s}))|\acute{\rho}_{\sigma}|(\hat{s})\ d\hat{s} \\
&=\int_{\hat{s}\in \bR^2}\frac{e^{-i[\sigma(\hat{s}) \cdot (\hat{H}(\rho_n)\hat{f}_0 + \hat{P}(\rho_n)\hat{f}_1)]}}{2\pi}f(\Lambda^{-1}(\vec{x} + \sigma(\hat{s}) - \vec{a}))|\acute{\rho}_{\sigma}|(\hat{s})\ d\hat{s} \\
&= f(\Lambda^{-1}(\cdot - \vec{a}))^{\{\hat{f}_0, \hat{f}_1\}}(\hat{H}(\rho_n), \hat{P}(\rho_n))(\vec{x}).
\end{align*}

Therefore, Equation (\ref{e.t.2}) is equal to
\begin{align*}
\Big( S,& f(\Lambda^{-1}(\cdot - \vec{a})  )^{\{\hat{f}_0, \hat{f}_1\}}A(\Lambda^{-1})_\gamma^\alpha A(\tilde{\Lambda})_\delta^\gamma\cdot g_\beta  \otimes \ad(\rho_n(F^\delta))\rho_n(E^\beta), \{\tilde{\Lambda}e_a\}_{a=0}^3\Big) \\
&=A(\Lambda^{-1})_\gamma^\alpha \phi^{\gamma, n}[f(\Lambda^{-1}(\cdot - \vec{a}))]\Big( S, g_\beta  \otimes \rho_n(E^\beta), \{\tilde{\Lambda}e_a\}_{a=0}^3\Big).
\end{align*}

To prove the second statement, note that $U(\vec{a}, \Lambda)1 = 1$. Let $\hat{g}_a = \Lambda e_a$, $a=0, \cdots, 3$. Then,
\begin{align*}
U(\vec{a},\Lambda)&\phi^{\alpha,n}(f) U(\vec{a},\Lambda)^{-1}1 =U(\vec{a},\Lambda)\phi^{\alpha,n}(f)1 \\
=& U(\vec{a},\Lambda)\left(S_0, f^{\{e_0, e_1\}} \otimes \rho_n(F^\alpha), \{e_a\}_{a=0}^3 \right) \\
=& \left(\Lambda S_0 + \vec{a}, e^{-i[\vec{a} \cdot (\hat{H}(\rho_n)\hat{g}_0 + \hat{P}(\rho_n)\hat{g}_1)]}
f^{\{e_0, e_1\}}(\Lambda^{-1}(\cdot-\vec{a}) )\otimes \rho_n(F^\alpha), \{\hat{g}_a\}_{a=0}^3\right),
\end{align*}
by definitions.
\end{proof}

\begin{rem}
Relative to $\Lambda S_0$ and for $\Lambda \vec{a} = a^0\hat{g}_0 + a^1\hat{g}_1$, we see that multiplication by $ e^{-i[\Lambda \vec{a} \cdot (\hat{H}(\rho_n)\hat{g}_0 + \hat{P}(\rho_n)\hat{g}_1)]} = e^{i[a^0\hat{H}(\rho_n) - a^1\hat{P}(\rho_n)]}$, corresponds to a shift \beq f(\cdot)^{\{\hat{g}_0, \hat{g}_1\}} \mapsto f(\cdot - \Lambda \vec{a})^{\{\hat{g}_0, \hat{g}_1\}} , \nonumber \eeq when we take Fourier Transform.
\end{rem}

\section{Causality}\label{s.c}

We are now down to the final Wightman's axiom. Recall that we have a countable set $\{\hat{H}(\rho_n), \hat{P}(\rho_n)\}_{n \in \mathbb{N}}$ to be defined later in Definition \ref{d.ma.1}. We will see in this section that to satisfy local commutativity, we must have $\hat{H}(\rho_n)^2 - \hat{P}(\rho_n)^2 > 0$, for each $n$.

\begin{defn}(Space-like separation)\label{d.sl.1}\\
Let $f, g \in \mathscr{P}$. The support of $f$, denoted ${\rm supp}\ f$, is the closed set obtained by taking the complement of the largest open set in which $f$ vanishes. We say that ${\rm supp}\ f$ and ${\rm supp}\ g$ are space-like (time-like) separated if $f(\vec{x})g(\vec{y}) = 0$ for all pairs of points $\vec{x} = (x^0, x),\ \vec{y}=(y^0, y) \in \bR^4$ such that \beq (\vec{x} -\vec{y})\cdot (\vec{x} - \vec{y}) = -(x^0-y^0)^2 + \sum_{i=1}^3 (x^i - y^i)^2 \leq\ (\geq)\  0. \nonumber \eeq
\end{defn}

\begin{rem}\label{r.sl.1}
\begin{itemize}
  \item Given a connected space-like rectangular surface $S$ contained in a plane, any two distinct points $\vec{x}, \vec{y} \in S$ are actually space-like separated. By Definition \ref{d.sl.1}, if $f$ and $g$ have supports which are space-like separated, then we must have $f(\vec{x})g(\vec{x}) = 0$ for any $\vec{x} \in S$.
  \item Note that on a time-like rectangular surface, two distinct points in it may not be time-like separated.
\end{itemize}
\end{rem}

\begin{notation}
For this section, we only consider a space-like surface $S$ contained in a plane. It is equipped with a Minkowski frame $\{\hat{f}_a\}_{a=0}^3$, which will be assumed throughout. To ease our notations, we will drop this Minkowski frame from our notation. This means \beq (S, f_\alpha \otimes \rho(E^\alpha)) \equiv (S, f_\alpha \otimes \rho(E^\alpha), \{\hat{f}_a\}_{a=0}^3). \nonumber \eeq
\end{notation}

\begin{defn}
Write the commutators of $f$ and $g$ in $\mathscr{P}$ as
\begin{align*}
\lba \phi^{\alpha,n}(f), \phi^{\beta,n}(g) \rba &:= \phi^{\alpha,n}(f)\phi^{\beta,n}(g) - \phi^{\alpha,n}(g)\phi^{\beta,n}(f), \\
\lba \phi^{\alpha,n}(f)^\ast, \phi^{\beta,n}(g)^\ast \rba &:= \phi^{\alpha,n}(f)^\ast\phi^{\beta,n}(g)^\ast - \phi^{\alpha,n}(g)^\ast\phi^{\beta,n}(f)^\ast.
\end{align*}
\end{defn}

\begin{rem}
We are using the version of Wightman's last axiom taken from \cite{glimm1981}, and not from \cite{streater}.
\end{rem}

\begin{lem}\label{l.w.1}
Let $S$ be a space-like surface contained in a plane.
Then, we have
\begin{align*}
\lba \phi^{\alpha,n}(f), \phi^{\beta,n}(g) \rba \left(S, h_\gamma \otimes \rho_n(E^\gamma)\right) = 0,
\end{align*}
for any  $\left(S, h_\gamma \otimes \rho_n(E^\gamma)\right) \in \mathscr{D}$, and $[\phi^{\alpha,n}(f), \phi^{\beta,n}(g)] 1 = 0$.

We also have
\begin{align*}
\lba\phi^{\alpha,n}(f)^\ast, \phi^{\beta,n}(g)^\ast\rba&1 = 0, \quad \lba\phi^{\alpha,n}(f)^\ast, \phi^{\beta,n}(g)^\ast \rba \left(S, h_\gamma \otimes \rho_n(E^\gamma)\right) = 0,
\end{align*}
for any  $\left(S, h_\gamma \otimes \rho_n(E^\gamma)\right) \in \mathscr{D}$.
\end{lem}

\begin{proof}
The first two commutator relations follow from Definitions \ref{d.co.1}, \ref{d.co.2}, and that $f^{\{\hat{f}_0, \hat{f}_1\}} \cdot g^{\{\hat{f}_0, \hat{f}_1\}} = g^{\{\hat{f}_0, \hat{f}_1\}} \cdot f^{\{\hat{f}_0, \hat{f}_1\}}$. By taking the adjoint, the last two follow immediately.
\end{proof}

\begin{rem}
These commutation relations hold, regardless of whether the supports of $f$ and $g$ are space-like separated or not.
\end{rem}

\subsection{CPT Theorem}\label{ss.cpt}

\begin{defn}\label{d.ac.1}
Define the anti-commutators and commutators of $f$ and $g$ in $\mathscr{P}$, as
\begin{align*}
\lca \phi^{\alpha,n}(f), \phi^{\beta,n}(g)^\ast \rca_\pm &:= \phi^{\alpha,n}(f)\phi^{\beta,n}(g)^\ast \pm \phi^{\alpha,n}(g)\phi^{\beta,n}(f)^\ast, \\
\lca \phi^{\alpha,n}(f)^\ast, \phi^{\beta,n}(g) \rca_\pm &:= \phi^{\alpha,n}(f)^\ast\phi^{\beta,n}(g) \pm \phi^{\alpha,n}(g)^\ast\phi^{\beta,n}(f).
\end{align*}
\end{defn}

\begin{lem}\label{l.w.2}
Recall $\ad(\rho(E^\alpha))$ refers to its adjoint representation on $\rho(\mathfrak{g})$. Suppose the Minkowski frame on $S$ is $\hat{f}_a = \Lambda e_a$, $a=0,1,2,3$. Write $\mathcal{C} = \{\hat{f}_0, \hat{f}_1\}$.

For any  $\left(S, h_\gamma \otimes \rho_n(E^\gamma) \right) \in \mathscr{D}$,
we have
\begin{align}
\lca \phi^{\alpha,n}&(f), \phi^{\beta,n}(g)^\ast \rca_\pm \left(S, h_\gamma \otimes \rho_n(E^\gamma) \right) \nonumber \\
= -&A(\Lambda)_\delta^\alpha \overline{A(\Lambda)_\mu^\beta}\left(S, {\rm B}^\pm[f^{\mathcal{C}}\cdot  \overline{g^{\mathcal{C}}} \pm g^{\mathcal{C}}\cdot  \overline{f^{\mathcal{C}}} ] \cdot h_\gamma \otimes \ad(\rho_n(F^{\delta}))\ad(\rho_n(F^{\mu}))\rho_n(E^\gamma) \right) \nonumber \\
&+ \left\langle (S, h_\gamma \otimes \rho_n(E^\gamma)), \phi^{\beta,n}(g)1 \right\rangle\ \phi^{\alpha,n}(f)1 \nonumber \\
&\pm \left\langle (S, h_\gamma \otimes \rho_n(E^\gamma)), \phi^{\beta,n}(f)1 \right\rangle\ \phi^{\alpha,n}(g)1, \label{e.co.3}
\end{align}
whereby ${\rm B}^+ = {\rm Re}$ and ${\rm B}^- = {\rm Im}$ for anti-commutation and commutation relations respectively.

And
\begin{align}
\lca\phi^{\alpha,n}& (f)^\ast, \phi^{\beta,n}(g) \rca_\pm \left(S, h_\gamma \otimes \rho_n(E^\gamma) \right) \nonumber \\
= -&\overline{A(\Lambda)_\delta^\alpha} A(\Lambda)_\mu^\beta \left(S,  {\rm B}^\pm[\overline{f^{\mathcal{C}}}\cdot  g^{\mathcal{C}} \pm \overline{g^{\mathcal{C}}}\cdot f^{\mathcal{C}} ]\cdot h_\gamma \otimes \ad(\rho_n(F^{\delta}))\ad(\rho_n(F^{\mu})) \rho_n(E^\gamma) \right) \nonumber \\
&+ \left\langle \left(S, g^{\mathcal{C}} A(\Lambda)_\mu^\beta\cdot h_\gamma \otimes \ad(\rho_n(F^{\mu}))\rho_n(E^\gamma)\right), \phi^{\alpha,n}(f)1 \right\rangle 1 \nonumber \\
&\pm \left\langle \left(S, f^{\mathcal{C}} A(\Lambda)_\mu^\beta\cdot h_\gamma \otimes \ad(\rho_n(F^{\mu}))\rho_n(E^\gamma) \right), \phi^{\alpha,n}(g)1 \right\rangle 1, \label{e.co.4}
\end{align}
whereby ${\rm B}^+ = {\rm Re}$ and ${\rm B}^- = {\rm Im}$ for anti-commutation and commutation relations respectively.
\end{lem}

\begin{proof}
From Definitions \ref{d.co.1}, \ref{d.co.2} and \ref{d.ado}, we see that
\begin{align*}
\phi^{\alpha,n}(f)\phi^{\beta,n}(g)^\ast &(S, h_\gamma \otimes \rho_n(E^\gamma)) \\
=& \phi^{\alpha,n}(f) \overline{A(\Lambda)_\mu^\beta}\bigg[\left(S, -\overline{g^{\mathcal{C}}} \cdot h_\gamma \otimes \ad(\rho_n(F^{\mu}))\rho_n(E^\gamma)\right) \\
&\hspace{3.6cm} + \left\langle \left(S, h_\gamma \otimes \rho_n(E^\gamma)\right), \phi^{\beta,n}(g)1 \right\rangle 1\bigg] \\
=& - A(\Lambda)_\delta^\alpha \overline{A(\Lambda)_\mu^\beta}\left(S, f^{\mathcal{C}} \cdot \overline{g^{\mathcal{C}}} \cdot h_\gamma \otimes \ad(\rho_n(F^{\delta}))\ad(\rho_n(F^{\mu}))\rho_n(E^\gamma)\right) \\
&+ \left\langle \left(S, h_\gamma \otimes \rho_n(E^\gamma)\right), \phi^{\beta,n}(g)1 \right\rangle\ \phi^{\alpha,n}(f)1.
\end{align*}

Similarly,
\begin{align*}
\phi^{\alpha,n}(g)\phi^{\beta,n}(f)^\ast &\left(S, h_\gamma \otimes \rho_n(E^\gamma) \right) \\
=& - A(\Lambda)_\delta^\alpha \overline{A(\Lambda)_\mu^\beta}\left(S, g^{\mathcal{C}}\cdot \overline{f^{\mathcal{C}}}  \cdot h_\gamma \otimes \ad(\rho_n(F^{\delta}))\ad(\rho_n(F^{\mu}))\rho_n(E^\gamma) \right) \\
&+ \left\langle \left(S, h_\gamma \otimes \rho_n(E^\gamma)\right), \phi^{\beta,n}(f)1 \right\rangle\ \phi^{\alpha,n}(g)1.
\end{align*}

Take the sum or difference, and we will obtain
\begin{align*}
 -A&(\Lambda)_\delta^\alpha \overline{A(\Lambda)_\mu^\beta}\Big(S,  [f^{\mathcal{C}} \cdot \overline{g^{\mathcal{C}}}\pm g^{\mathcal{C}} \cdot \overline{f^{\mathcal{C}}}  ]\cdot h_\gamma \otimes \ad(\rho_n(F^{\delta}))\ad(\rho_n(F^{\mu}))\rho_n(E^\gamma) \Big) \\
+& \left\langle (S, h_\gamma \otimes \rho_n(E^\gamma)), \phi^{\beta,n}(g)1 \right\rangle \phi^{\alpha,n}(f)1 \pm \left\langle (S, h_\gamma \otimes \rho_n(E^\gamma)), \phi^{\beta,n}(f)1 \right\rangle \phi^{\alpha,n}(g)1.
\end{align*}

Since $f^{\mathcal{C}}\cdot  \overline{g^{\mathcal{C}}} \pm g^{\mathcal{C}}\cdot  \overline{f^{\mathcal{C}}}$ is real and purely imaginary respectively, this proves Equation (\ref{e.co.3}). The proof for Equation (\ref{e.co.4}) is similar, hence omitted.
\end{proof}

\begin{rem}\label{r.cpt.1}
Without any loss of generality, we assume that a time-like plane $S^\flat$ spanned by $\{\hat{f}_0, \hat{f}_1\}$, is parametrized by $\vec{y}(\hat{s}) \equiv \vec{y}(s, \bar{s}) :=  s\hat{f}_0 + \bar{s}\hat{f}_1$, $s, \bar{s} \in \bR$, whereby $\hat{f}_0 \cdot \hat{f}_0 = -1$, $\hat{f}_1 \cdot \hat{f}_1 = 1$, $\hat{f}_0 \cdot \hat{f}_1 = 0$.

Suppose we now assume that ${\rm supp}\ f$ and ${\rm supp}\ g$ are disjoint compact sets. Write $\hat{H} = \hat{H}(\rho_n)$, $\hat{P} = \hat{P}(\rho_n)$. Let $\vec{x} \in S$. By definition, for any $\vec{x} \in S$,
\begin{align*}
g^{\{\hat{f}_0, \hat{f}_1\}}(\hat{H}, \hat{P})(\vec{x}) =& \int_{\hat{s} \in \bR^2}\frac{e^{-i[\vec{y}(\hat{s}) \cdot (\hat{H}\hat{f}_0+\hat{P}\hat{f}_1)]}}{2\pi}g(\vec{x} + \vec{y}(\hat{s}))|\acute{\rho}_{\vec{y}}|(\hat{s})\ d\hat{s}, \\
\overline{f^{\{\hat{f}_0, \hat{f}_1\}}}(\hat{H}, \hat{P})(\vec{x}) =& \int_{\hat{t} \in \bR^2}\frac{e^{i[\vec{y}(\hat{t}) \cdot (\hat{H}\hat{f}_0+\hat{P}\hat{f}_1)]}}{2\pi}\bar{f}(\vec{x} + \vec{y}(\hat{t}) )|\acute{\rho}_{\vec{y}}|(\hat{t})\ d\hat{t}.
\end{align*}

Write $g_{\vec{x}}(\cdot) = g(\vec{x} + \cdot)$, $\bar{f}_{\vec{x}}(\cdot) = \bar{f}(\vec{x} + \cdot)$. Thus,
\begin{align*}
\Big[g^{\{\hat{f}_0, \hat{f}_1\}}&\cdot  \overline{f^{\{\hat{f}_0, \hat{f}_1\}}}\Big](\hat{H}, \hat{P})(\vec{x}) \\
&= \int_{\hat{s}, \hat{t} \in \bR^2}\frac{e^{-i[\vec{y}(\hat{s}) \cdot (\hat{H}\hat{f}_0+\hat{P}\hat{f}_1)]}}{(2\pi)^2}g_{\vec{x}}( \vec{y}(\hat{s})+\vec{y}(\hat{t}))\bar{f}_{\vec{x}}( \vec{y}( \hat{t}))|\acute{\rho}_{\vec{y}}|(\hat{s})|\acute{\rho}_{\vec{y}}|(\hat{t})\ d\hat{s}d\hat{t}  \\
&= \int_{\hat{t} \in \bR^2, \hat{s} \in D}\frac{e^{-i[\vec{y}(\hat{s}) \cdot (\hat{H}\hat{f}_0+\hat{P}\hat{f}_1)]}}{(2\pi)^2}g_{\vec{x}}( \vec{y}(\hat{s})+\vec{y}(\hat{t}))\bar{f}_{\vec{x}}( \vec{y}( \hat{t}))|\acute{\rho}_{\vec{y}}|(\hat{s})|\acute{\rho}_{\vec{y}}|(\hat{t})\ d\hat{s}d\hat{t}.
\end{align*}

Similarly,
\begin{align*}
\Big[f^{\{\hat{f}_0, \hat{f}_1\}}&\cdot  \overline{g^{\{\hat{f}_0, \hat{f}_1\}}}\Big](\hat{H}, \hat{P})(\vec{x})  \\
&= \int_{\hat{t} \in \bR^2, \hat{s} \in -D }\frac{e^{-i[\vec{y}(\hat{s}) \cdot (\hat{H}\hat{f}_0+\hat{P}\hat{f}_1)]}}{(2\pi)^2}f_{\vec{x}}( \vec{y}(\hat{s})+\vec{y}(\hat{t}))\bar{g}_{\vec{x}}( \vec{y}( \hat{t}))|\acute{\rho}_{\vec{y}}|(\hat{s})|\acute{\rho}_{\vec{y}}|(\hat{t})\ d\hat{s}d\hat{t}.
\end{align*}

From both expressions, we see that the integrals depend on the relative displacement between pairs of positions in their respective supports. When the region of integration is on $D$, it is clear that we are referring to $g^{\{\hat{f}_0, \hat{f}_1\}}\cdot  \overline{f^{\{\hat{f}_0, \hat{f}_1\}}}$; when the region of integration is on $-D$, then we are referring to $f^{\{\hat{f}_0, \hat{f}_1\}}\cdot  \overline{g^{\{\hat{f}_0, \hat{f}_1\}}}$. Note that the vectors in the set $\vec{y}(-D)$, are in the opposite direction of those vectors in $\vec{y}(D)$.

Thus, by reversing time direction and taking space inversion (parity), we can obtain its complex conjugate. The CPT theorem is being implied by these two expressions. See Remark \ref{r.cpt.2}. Note that here, `C' refers to complex conjugation, not charge conjugation.
\end{rem}

In general, the anti-commutators and commutators will not be equal to zero, even when the supports are space-like separated.

\begin{lem}\label{l.w.6}
Let $f, g \in \mathscr{P}$ for which their supports are space-like separated, and let $S$ be a space-like plane, equipped with a Minkowski frame $\{\hat{f}_a\}_{a=0}^3$. Let $S^\flat$ be the span of $\{\hat{f}_0, \hat{f}_1\}$.

Fix a $\vec{x} \in S$. Write $g_{\vec{x}}(\cdot) = g(\vec{x} + \cdot)$, $\bar{f}_{\vec{x}}(\cdot) = \bar{f}(\vec{x} + \cdot)$. Suppose $\hat{H}(\rho_n)^2 - \hat{P}(\rho_n)^2 > 0$.

If ${\rm supp}\ f \cap (\vec{x} + S^\flat) = \emptyset$ or ${\rm supp}\ g \cap (\vec{x} + S^\flat) = \emptyset$, then for any $u, v \in S^\flat$, we have
\beq \frac{e^{-i[(u-v) \cdot (\hat{H}(\rho_n)\hat{f}_0 + \hat{P}(\rho_n)\hat{f}_1)]}}{(2\pi)^2}
g_{\vec{x}}(u)\bar{f}_{\vec{x}}(v) \mp \frac{e^{-i[(v-u) \cdot (\hat{H}(\rho_n)\hat{f}_0+\hat{P}(\rho_n)\hat{f}_1)]}}{(2\pi)^2}f_{\vec{x}}(v)
\bar{g}_{\vec{x}}(u) = 0. \label{e.a.2} \eeq

Suppose both sets are non-empty. If $u-v$ is parallel to $\hat{P}(\rho_n)\hat{f}_0 + \hat{H}(\rho_n)\hat{f}_1$, then the commutation and anti-commutation relations in Equation (\ref{e.a.2}) hold when $f\cdot \bar{g}$ is real and imaginary respectively.
\end{lem}

\begin{proof}
When one of the sets is empty, then $f_{\vec{x}}(v)\cdot \bar{g}_{\vec{x}}(u) = 0$, so clearly Equation (\ref{e.a.2}) holds.

Now consider when both are non-empty. Since \beq (u-v) \cdot (\hat{H}(\rho_n)\hat{f}_0 + \hat{P}(\rho_n)\hat{f}_1) = c (\hat{P}(\rho_n)\hat{f}_0 + \hat{H}(\rho_n)\hat{f}_1)\cdot (\hat{H}(\rho_n)\hat{f}_0 + \hat{P}(\rho_n)\hat{f}_1) =0, \nonumber \eeq we have \beq e^{-i[(u-v) \cdot (\hat{H}(\rho_n)\hat{f}_0 + \hat{P}(\rho_n)\hat{f}_1)]} = \cos\left[(u-v) \cdot (\hat{H}(\rho_n)\hat{f}_0 + \hat{P}(\rho_n)\hat{f}_1) \right]=1. \nonumber \eeq When $f_{\vec{x}}(v)\cdot \bar{g}_{\vec{x}}(u)$ is real, then the LHS of Equation (\ref{e.a.2}) becomes \beq
\left[\frac{e^{-i[(u-v) \cdot (\hat{H}(\rho_n)\hat{f}_0 + \hat{P}(\rho_n)\hat{f}_1)]}}{(2\pi)^2} - \frac{e^{-i[(v-u) \cdot (\hat{H}(\rho_n)\hat{f}_0+\hat{P}(\rho_n)\hat{f}_1)]}}{(2\pi)^2}\right]f_{\vec{x}}(v)
\bar{g}_{\vec{x}}(u), \nonumber \eeq which is zero.

When $f_{\vec{x}}(v)\cdot \bar{g}_{\vec{x}}(u)$ is purely imaginary, then the LHS of Equation (\ref{e.a.2}) becomes \beq
\left[-\frac{e^{-i[(u-v) \cdot (\hat{H}(\rho_n)\hat{f}_0 + \hat{P}(\rho_n)\hat{f}_1)]}}{(2\pi)^2} + \frac{e^{-i[(v-u) \cdot (\hat{H}(\rho_n)\hat{f}_0+\hat{P}(\rho_n)\hat{f}_1)]}}{(2\pi)^2}\right]f_{\vec{x}}(v)
\bar{g}_{\vec{x}}(u), \nonumber \eeq which is zero.
\end{proof}

\begin{rem}\label{r.ex.4}
Suppose both ${\rm supp}\ f \cap (\vec{x} + S^\flat)$ and ${\rm supp}\ g \cap (\vec{x} + S^\flat)$ are non-empty.
The lemma says that there exists a space-like line in $S^\flat$, such that the LHS of Equation (\ref{e.a.2}) is zero.

If $\hat{H}(\rho_n)\hat{f}_0 + \hat{P}(\rho_n)\hat{f}_1$ is space-like or null, then the LHS of Equation (\ref{e.a.2}) cannot be zero, on any space-like line in $S^\flat$. This can be inferred from Lemma \ref{l.b.4}.

Thus, it is essential that a positive mass gap exists in $\mathscr{H}(\rho_n)$, for the lemma to hold true.
\end{rem}

%
%

Consider a bilinear map that sends
\begin{align*}
(f,  g)&\in \mathscr{P} \times \mathscr{P} \\
\longmapsto& \left\langle  \phi^{\alpha,n}(f) \phi^{\beta, n}(g)^\ast \left(S, h_\gamma \otimes \rho_n(E^\gamma) \right), (\tilde{S}, \tilde{h}_\gamma \otimes \rho_n(E^\gamma) )\right\rangle \\
&- \left\langle \left(S, h_\gamma \otimes \rho_n(E^\gamma)\right), \phi^{\beta, n}(g)1 \right\rangle\left\langle \phi^{\alpha,n}(f)1, (\tilde{S}, \tilde{h}_\gamma \otimes \rho_n(E^\gamma) )\right\rangle .
\end{align*}
From Proposition \ref{p.td.1}, we have a tempered distribution $W(\vec{x}, \vec{y})$ such that
\begin{align*}
\int_{\vec{x} \in \bR^4}&\int_{\vec{y} \in \bR^4} W(\vec{x}, \vec{y}) f(\vec{x}) \otimes_\bR g(\vec{y})\ d\vec{x} d\vec{y} \\
=& \left\langle  \phi^{\alpha,n}(f) \phi^{\beta, n}(g)^\ast \left(S, h_\gamma \otimes \rho_n(E^\gamma)\right), (\tilde{S}, \tilde{h}_\gamma \otimes \rho_n(E^\gamma))\right\rangle \\
&- \left\langle \left(S, h_\gamma \otimes \rho_n(E^\gamma)\right), \phi^{\beta, n}(g)1 \right\rangle\left\langle \phi^{\alpha,n}(f)1, (\tilde{S}, \tilde{h}_\gamma \otimes \rho_n(E^\gamma) )\right\rangle .
\end{align*}
See Remark \ref{r.c.2a}.

Indeed, writing $\vec{x} = \sum_{a=0}^3x^a \hat{f}_a$ and $\vec{y} = \sum_{a=0}^3y^a \hat{f}_a$, we have that
\begin{align}
&\left( x^0\hat{f}_0 + x^1\hat{f}_1, y^0\hat{f}_0 + y^1\hat{f}_1 \right) \longmapsto \nonumber \\
&\int_{(x^2, x^3) \in \bR^2}\int_{(y^2, y^3) \in \bR^2} W(\vec{x}, \vec{y}) f(\vec{x}) \otimes_\bR g(\vec{y}) \ dx^2dx^3 dy^2dy^3, \label{e.f.2}
\end{align}
defines a continuous function on $S^\flat \times S^\flat$.

\begin{notation}\label{n.co.3}
Suppose $f \cdot \bar{g}$ is real. Then we will define tempered distribution ${\rm Re}\  W$, such that
\begin{align*}
\int_{\vec{x} \in \bR^4}\int_{\vec{y} \in \bR^4}&{\rm Re}\  W(\vec{x}, \vec{y})\left[ f(\vec{x})\otimes_\bR g(\vec{y})\right]\ d\vec{x} d\vec{y} \\
&:= \int_{\vec{x} \in \bR^4}\int_{\vec{y} \in \bR^4} W(\vec{x}, \vec{y})\left[ f(\vec{x})\otimes_\bR g(\vec{y})\right]\ d\vec{x} d\vec{y}.
\end{align*}

When $f \cdot \bar{g}$ is purely imaginary, we will define tempered distribution ${\rm Im}\  W$, such that
\begin{align*}
\int_{\vec{x} \in \bR^4}\int_{\vec{y} \in \bR^4}&{\rm Im}\  W(\vec{x}, \vec{y})\left[ f(\vec{x})\otimes_\bR g(\vec{y})\right]\ d\vec{x} d\vec{y} \\
&:= \int_{\vec{x} \in \bR^4}\int_{\vec{y} \in \bR^4} W(\vec{x}, \vec{y})\left[ f(\vec{x})\otimes_\bR g(\vec{y})\right]\ d\vec{x} d\vec{y}.
\end{align*}
\end{notation}

\begin{rem}\label{r.c.2a}
Suppose we write $f = \underline{f} + i\overline{f}$, $g = \underline{g} + i\overline{g}$, whereby $\underline{f} = {\rm Re}\ f$, $\overline{f} = {\rm Im}\ f$, $\underline{g} = {\rm Re}\ g$, $\overline{g} = {\rm Im}\ g$. From Notation \ref{n.co.3}, we understand
\begin{align*}
\int_{\vec{x} \in \bR^4}&\int_{\vec{y} \in \bR^4} W(\vec{x}, \vec{y})\left[ f(\vec{x})\otimes_\bR g(\vec{y})\right]\ d\vec{x} d\vec{y} \\
=& \int_{\vec{x} \in \bR^4}\int_{\vec{y} \in \bR^4} {\rm Re}\ W(\vec{x}, \vec{y})\left[\ \underline{f}(\vec{x})\otimes_\bR \underline{g}(\vec{y}) + i\overline{f}(\vec{x}) \otimes_\bR i\overline{g}(\vec{y})\right]\ d\vec{x} d\vec{y} \\
&+ \int_{\vec{x} \in \bR^4}\int_{\vec{y} \in \bR^4} {\rm Im}\ W(\vec{x}, \vec{y})\left[\ i\overline{f}(\vec{x})\otimes_\bR \underline{g}(\vec{y})+ \underline{f}(\vec{x}) \otimes_\bR i\overline{g}(\vec{y}) \right]\ d\vec{x} d\vec{y}.
\end{align*}
\end{rem}

\begin{lem}\label{l.w.3}
Fix a space-like plane $S$ and let the Minkowski frame on $S$ be $\hat{f}_a = \Lambda e_a$, $a=0,1,2,3$, as defined in Definition \ref{d.a.1}. Suppose $\hat{H}(\rho_n)^2 - \hat{P}(\rho_n)^2 > 0$.

We have
\begin{align*}
{\rm Re}\ W(\vec{x}, \vec{y}) - {\rm Re}\ W(\vec{y}, \vec{x})  &= 0, \\
{\rm Im}\ W(\vec{x}, \vec{y}) + {\rm Im}\ W(\vec{y}, \vec{x}) &= 0,
\end{align*}
provided $\vec{0} \neq \vec{x} - \vec{y}$ can be written as $c_1 (\hat{P}(\rho_n)\hat{f}_0 + \hat{H}(\rho_n)\hat{f}_1) + \sum_{i=2}^3 c_i \hat{f}_i$ for some constants $c_i$'s.
\end{lem}

\begin{proof}
Choose any $f, g \in \mathscr{P}$ such that their compact supports are space-like separated. Write $g_{\vec{x}}(\cdot) = g(\vec{x} + \cdot)$, $\bar{f}_{\vec{x}}(\cdot) = \bar{f}(\vec{x} + \cdot)$, and $\hat{H} = \hat{H}(\rho_n)$, $\hat{P} = \hat{P}(\rho_n)$.

From the proof in Lemma \ref{l.w.2},
\begin{align*}
\phi^{\alpha,n}(f) & \phi^{\beta, n}(g)^\ast \left(S, h_\gamma \otimes \rho_n(E^\gamma) \right) -
\left\langle \left(S, h_\gamma \otimes \rho_n(E^\gamma) \right), \phi^{\beta, n}(g)1 \right\rangle\ \phi^{\alpha,n}(f)1 \\
=& -A(\Lambda)_\delta^\alpha \overline{A(\Lambda)_\mu^\beta}\Big(S, [f^{\{\hat{f}_0, \hat{f}_1\}} \cdot \overline{g^{\{\hat{f}_0, \hat{f}_1\}}} ]\cdot h_\gamma \otimes \ad(\rho_n(F^{\delta}))\ad(\rho_n(F^{\mu}))\rho_n(E^\gamma) \Big).
\end{align*}
Refer to the calculations in Remark \ref{r.cpt.1}. If we swap $f$ with $g$ in the above expression and take the sum or difference, we will have
\beq -A(\Lambda)_\delta^\alpha \overline{A(\Lambda)_\mu^\beta}\Big(S, A^\pm \cdot h_\gamma \otimes \ad(\rho_n(F^{\delta}))\ad(\rho_n(F^{\mu}))\rho_n(E^\gamma)\Big),
\nonumber \eeq
whereby
\begin{align*}
A^\pm(\vec{z}) = &
\int_{\hat{s}, \hat{t} \in \bR^2}\frac{e^{-i[(\vec{y}(\hat{s})- \vec{y}(\hat{t})) \cdot (\hat{H}\hat{f}_0+\hat{P}\hat{f}_1)]}}{(2\pi)^2}f_{\vec{z}}( \vec{y}(\hat{s}))\bar{g}_{\vec{z}}( \vec{y}( \hat{t}))|\acute{\rho}_{\vec{y}}|(\hat{s})|\acute{\rho}_{\vec{y}}|(\hat{t})\ d\hat{s}d\hat{t} \\
&\pm \int_{\hat{s}, \hat{t} \in \bR^2}\frac{e^{-i[(\vec{y}(\hat{t})- \vec{y}(\hat{s})) \cdot (\hat{H}\hat{f}_0+\hat{P}\hat{f}_1)]}}{(2\pi)^2}g_{\vec{z}}( \vec{y}(\hat{t}))\bar{f}_{\vec{z}}( \vec{y}( \hat{s}))|\acute{\rho}_{\vec{y}}|(\hat{s})|\acute{\rho}_{\vec{y}}|(\hat{t})\ d\hat{s}d\hat{t}.
\end{align*}
Since their supports are disjoint, by definition of $W(\vec{y}, \vec{x})$, we note that swapping the arguments $\vec{x}$ and $\vec{y}$, is equivalent to swapping $f$ and $g$, i.e. $(\vec{x}, \vec{y}) \longleftrightarrow (\vec{y}, \vec{x})$ is the same as $f(\vec{x})\bar{g}(\vec{y}) \longleftrightarrow g(\vec{y})\bar{f}(\vec{x})$, in their respective integrals.

Consider when $\vec{0} \neq \vec{x} - \vec{y} = c_1 (\hat{P}\hat{f}_0 + \hat{H}\hat{f}_1) + \sum_{i=2}^3 c_i \hat{f}_i$, whereby not all the $c_i$'s are zero. By Lemma \ref{l.w.6}, we see that if $c_2$ or $c_3$ is non-zero, then we must have $W(\vec{x}, \vec{y}) = W(\vec{y}, \vec{x}) = 0$.

Consider when $c_1 \neq 0$ and $f\bar{g}$ is real. By Equation (\ref{e.a.2}), we have that
\begin{align*}
{\rm Re}\ W(\vec{x}, \vec{y}) - {\rm Re}\ W(\vec{y}, \vec{x})
= 0,
\end{align*}
for the real part.

Now consider when $c_1 \neq 0$ and $f\bar{g}$ is purely imaginary. By Equation (\ref{e.a.2}), we have that
\begin{align*}
{\rm Im}\ W(\vec{x}, \vec{y}) + {\rm Im}\ W(\vec{y}, \vec{x}) = 0,
\end{align*}
for the imaginary part.
\end{proof}

\begin{rem}\label{r.ex.3}
In the Wightman's axiom for local commutativity, it is required that $\lca \phi^{\alpha,n}(f), \phi^{\beta,n}(g)^\ast \rca_\pm$ is zero when their supports are space-like separated. This is not true in general, even if $f\cdot \bar{g}$ is real or purely imaginary. The same remark applies to $\lca \phi^{\alpha,n}(f)^\ast, \phi^{\beta,n}(g) \rca_\pm$, as we will see in Lemmas \ref{l.w.4} and \ref{l.w.5}.

Indeed, when $(\vec{x} + S^\flat) \cap\ {\rm supp}\ f$ or $(\vec{x} + S^\flat) \cap\ {\rm supp}\ g$ is empty for every $\vec{x} \in S$, then we have $\lca \phi^{\alpha,n}(f), \phi^{\beta,n}(g)^\ast \rca_\pm = \lca \phi^{\alpha,n}(f)^\ast, \phi^{\beta,n}(g) \rca_\pm = 0$, acting on $(S, f_\alpha \otimes E^\alpha)$.

In the proof, when $c_1 \neq 0$ but $c_2 = c_3 = 0$, we see that ${\rm Re}\ W(\vec{x}, \vec{y}) ={\rm Im}\ W(\vec{x}, \vec{y}) = 1$. In the case for ${\rm Re}\ W(\vec{x}, \vec{y})$, the arguments $(\vec{x}, \vec{y})$ will not serve any purpose as the distribution is symmetric. But in the case for ${\rm Im}\ W(\vec{x}, \vec{y})$, the arguments $(\vec{x}, \vec{y})$ will dictate whether the corresponding first and second input, should be real or purely imaginary function.

\end{rem}

Consider another bilinear map that sends
\begin{align*}
(f,g)&\in \mathscr{P} \times \mathscr{P}
\longmapsto \left\langle  \phi^{\alpha,n}(f)^\ast \phi^{\beta, n}(g) \left(S, h_\gamma \otimes \rho_n(E^\gamma)\right), (\tilde{S}, \tilde{h}_\gamma \otimes \rho_n(E^\gamma) )\right\rangle .
\end{align*}
From Proposition \ref{p.td.1}, we have a tempered distribution $\tilde{W}(\vec{x}, \vec{y})$ such that
\begin{align*}
\int_{\vec{x} \in \bR^4}&\int_{\vec{y} \in \bR^4} \tilde{W}(\vec{x}, \vec{y}) f(\vec{x})\otimes_\bR g(\vec{y})\ d\vec{x} d\vec{y} \\
=& \left\langle  \phi^{\alpha,n}(f)^\ast \phi^{\beta, n}(g) \left(S, h_\gamma \otimes \rho_n(E^\gamma)\right), (\tilde{S}, \tilde{h}_\gamma \otimes \rho_n(E^\gamma) )\right\rangle .
\end{align*}
Similar to Equation (\ref{e.f.2}), we can define a continuous complex-valued function on $S^\flat \times S^\flat$ from it.

\begin{rem}
Suppose we write $f = \underline{f} + i\overline{f}$, $g = \underline{g} + i\overline{g}$, whereby $\underline{f} = {\rm Re}\ f$, $\overline{f} = {\rm Im}\ f$, $\underline{g} = {\rm Re}\ g$, $\overline{g} = {\rm Im}\ g$. Define tempered distributions ${\rm Re}\ \tilde{W}$ and ${\rm Im}\ \tilde{W}$, similar to how we defined ${\rm Re}\ W$ and ${\rm Im}\ W$ in Notation \ref{n.co.3}. We understand
\begin{align*}
\int_{\vec{x} \in \bR^4}&\int_{\vec{y} \in \bR^4} \tilde{W}(\vec{x}, \vec{y})\left[ f(\vec{x})\otimes_\bR g(\vec{y})\right]\ d\vec{x} d\vec{y} \\
=& \int_{\vec{x} \in \bR^4}\int_{\vec{y} \in \bR^4} {\rm Re}\ \tilde{W}(\vec{x}, \vec{y})\left[\ \underline{f}(\vec{x}) \otimes_\bR \underline{g}(\vec{y}) + i\overline{f}(\vec{x})\otimes_\bR i\overline{g}(\vec{y})\right]\ d\vec{x} d\vec{y} \\
&+ \int_{\vec{x} \in \bR^4}\int_{\vec{y} \in \bR^4} {\rm Im}\ \tilde{W}(\vec{x}, \vec{y})\left[\ \underline{f}(\vec{x})\otimes_\bR i\overline{g}(\vec{y})+  i\overline{f}(\vec{x}) \otimes_\bR\underline{g}(\vec{y})\right]\ d\vec{x} d\vec{y}.
\end{align*}
\end{rem}

\begin{lem}\label{l.w.4}
Fix a space-like plane $S$ and let $\{\hat{f}_a\}_{a=0}^3$ be a basis as defined in Definition \ref{d.a.1}. Suppose $\hat{H}(\rho_n)^2 - \hat{P}(\rho_n)^2 > 0$.

We have
\begin{align*}
{\rm Re}\ \tilde{W}(\vec{x}, \vec{y}) - {\rm Re}\ \tilde{W}(\vec{y}, \vec{x}) &= 0, \\
{\rm Im}\ \tilde{W}(\vec{x}, \vec{y}) + {\rm Im}\ \tilde{W}(\vec{y}, \vec{x}) &= 0,
\end{align*}
provided $\vec{0} \neq \vec{x} - \vec{y}$ can be written as $c_1 (\hat{P}(\rho_n)\hat{f}_0 + \hat{H}(\rho_n)\hat{f}_1) + \sum_{i=2}^3 c_i \hat{f}_i$ for some constants $c_i$'s.
\end{lem}

\begin{proof}
Proof is similar to Lemma \ref{l.w.3}, hence omitted.
\end{proof}

Finally consider the following bilinear map that sends
\begin{align*}
(f,g)&\in \mathscr{P} \times \mathscr{P}
\longmapsto \left\langle  \phi^{\alpha,n}(f)^\ast \phi^{\beta, n}(g) 1, (\tilde{S}, \tilde{h}_\gamma \otimes \rho_n(E^\gamma) )\right\rangle .
\end{align*}
From Proposition \ref{p.td.1}, we have a tempered distribution $\check{W}(\vec{x}, \vec{y})$ such that
\begin{align*}
\int_{\vec{x} \in \bR^4}&\int_{\vec{y} \in \bR^4} \check{W}(\vec{x}, \vec{y}) f(\vec{x})\otimes_\bR g(\vec{y})\ d\vec{x} d\vec{y} \\
=& \left\langle  \phi^{\alpha,n}(f)^\ast \phi^{\beta, n}(g) 1, (\tilde{S}, \tilde{h}_\gamma \otimes \rho_n(E^\gamma) )\right\rangle .
\end{align*}
Similar to Equation (\ref{e.f.2}), we can obtain a continuous complex-valued function defined on $S^\flat \times S^\flat$, from the integral.

\begin{rem}
Suppose we write $f = \underline{f} + i\overline{f}$, $g = \underline{g} + i\overline{g}$, whereby $\underline{f} = {\rm Re}\ f$, $\overline{f} = {\rm Im}\ f$, $\underline{g} = {\rm Re}\ g$, $\overline{g} = {\rm Im}\ g$. Define tempered distributions ${\rm Re}\ \check{W}$ and ${\rm Im}\ \check{W}$, similar to how we defined ${\rm Re}\ W$ and ${\rm Im}\ W$ in Notation \ref{n.co.3}. We understand
\begin{align*}
\int_{\vec{x} \in \bR^4}&\int_{\vec{y} \in \bR^4} \check{W}(\vec{x}, \vec{y})\left[ f(\vec{x})\otimes_\bR g(\vec{y})\right]\ d\vec{x} d\vec{y} \\
=& \int_{\vec{x} \in \bR^4}\int_{\vec{y} \in \bR^4} {\rm Re}\ \check{W}(\vec{x}, \vec{y})\left[\ \underline{f}(\vec{x})\otimes_\bR \underline{g}(\vec{y}) + i\overline{f}(\vec{x}) \otimes_\bR i\overline{g}(\vec{y})\right]\ d\vec{x} d\vec{y} \\
&+ \int_{\vec{x} \in \bR^4}\int_{\vec{y} \in \bR^4} {\rm Im}\ \check{W}(\vec{x}, \vec{y})\left[\ \underline{f}(\vec{x})\otimes_\bR i\overline{g}(\vec{y})+ i\overline{f}(\vec{x}) \otimes_\bR \underline{g}(\vec{y})\right]\ d\vec{x} d\vec{y}.
\end{align*}
\end{rem}

\begin{lem}\label{l.w.5}
Fix a space-like plane $S$ and let $\{\hat{f}_a\}_{a=0}^3$ be a basis as defined in Definition \ref{d.a.1}.
Suppose $\hat{H}(\rho_n)^2 - \hat{P}(\rho_n)^2 > 0$.

We have
\begin{align*}
{\rm Re}\ \check{W}(\vec{x}, \vec{y}) - {\rm Re}\ \check{W}(\vec{y}, \vec{x}) &= 0, \\
{\rm Im}\ \check{W}(\vec{x}, \vec{y}) + {\rm Im}\ \check{W}(\vec{y}, \vec{x}) &= 0,
\end{align*}
provided $\vec{0} \neq \vec{x} - \vec{y}$ can be written as $c_1 (\hat{P}(\rho_n)\hat{f}_0 + \hat{H}(\rho_n)\hat{f}_1) + \sum_{i=2}^3 c_i \hat{f}_i$ for some constants $c_i$'s.
\end{lem}

\begin{proof}
Proof is similar to Lemma \ref{l.w.3}, hence omitted.
\end{proof}

\begin{rem}
Note that $\phi^{\alpha,n}(f)^\ast 1 = 0$ by definition, thus we always have \beq \phi^{\alpha,n}(f) \phi^{\beta, n}(g)^\ast 1 \pm \phi^{\alpha,n}(g) \phi^{\beta, n}(f)^\ast 1 = 0. \nonumber \eeq

The Lemmas \ref{l.w.1}, \ref{l.w.3}, \ref{l.w.4} and \ref{l.w.5} are collectively known as the local commutation and anti-commutation relations.

Indeed, we see that the relations hold, because $\hat{H}(\rho_n)\hat{f}_0 + \hat{P}(\rho_n)\hat{f}_1$ is time-like. If it is space-like or null vector, then we see local commutativity only holds for space-like directions in the span of $\{\hat{f}_2, \hat{f}_3\}$.

Now, $\{\hat{P}(\rho_n) \hat{f}_0 + \hat{H}(\rho_n)\hat{f}_1, \hat{f}_2, \hat{f}_3\}$ is a set of space-like directional vectors in $\vec{x} + S^\flat$. In fact, for any vector given by a linear combination of these three vectors, there exists a sequence of translations and Lorentz transformations that rotates it by $\pi$ radians in $\bR^4$, which is time and space inversion. Refer to Lemma \ref{l.b.3} and Remark \ref{r.p.3}.
\end{rem}

\begin{cor}(CPT Theorem)\\
Fix a space-like plane $S$, with a time-like plane $S^\flat$ spanned by $\{\hat{f}_0, \hat{f}_1\}$, whereby $\{\hat{f}_a\}_{a=0}^3$ is a basis defined in Definition \ref{d.a.1}. By abuse of notation, write
\beq W = {\rm Re}\ W + \sqrt{-1}\ {\rm Im}\ W, \ \ \tilde{W} = {\rm Re}\ \tilde{W} + \sqrt{-1}\ {\rm Im}\ \tilde{W}, \ \ \check{W} = {\rm Re}\ \check{W} + \sqrt{-1}\ {\rm Im}\ \check{W}. \nonumber \eeq

Suppose $\hat{H}(\rho_n)^2 - \hat{P}(\rho_n)^2 > 0$.
We have that
\begin{align*}
W(\vec{y}, \vec{x}) &= \overline{W}(\vec{x}, \vec{y}), \quad \tilde{W}(\vec{y}, \vec{x}) = \overline{\tilde{W}}(\vec{x}, \vec{y}), \quad
\check{W}(\vec{y}, \vec{x}) = \overline{\check{W}}(\vec{x}, \vec{y}),
\end{align*}
provided $\vec{0} \neq \vec{x} - \vec{y} = c_1 (\hat{P}(\rho_n)\hat{f}_0 + \hat{H}(\rho_n)\hat{f}_1) + \sum_{i=2}^3 c_i \hat{f}_i$, for constants $c_i$'s.
\end{cor}

\begin{proof}
Immediate from the statements in Lemmas \ref{l.w.3}, \ref{l.w.4} and \ref{l.w.5}.
\end{proof}

\begin{rem}\label{r.cpt.2}
In Lemma \ref{l.ex.1}, we will see that we can write $W(\vec{y}, \vec{x}) = \mathscr{W}(\vec{\xi})$, for some distribution $\mathscr{W}$, and $\vec{\xi} = \vec{y} - \vec{x}$. The same remark applies to $\tilde{W}(\vec{y}, \vec{x})$ and $\check{W}(\vec{y}, \vec{x})$.

Thus, the above corollary says that by taking time inversion and space inversion (parity) of a space-like vector, i.e. $\vec{\xi} \mapsto -\vec{\xi}$, is equivalent to taking the complex conjugation. This is the content of the CPT Theorem. Refer also to Remark \ref{r.cpt.1}. But, this only applies if $\vec{\xi}$ lies in the hyperplane spanned by $\{\hat{P}(\rho_n) \hat{f}_0 + \hat{H}(\rho_n)\hat{f}_1, \hat{f}_2, \hat{f}_3\}$.
\end{rem}

In \cite{streater}, the local commutativity is actually stated as
\begin{align*}
\phi^{\alpha,n}(f)\phi^{\beta,n}(g) \pm \phi^{\beta,n}(g)\phi^{\alpha,n}(f) &= 0, \\
\phi^{\alpha,n}(f)^\ast\phi^{\beta,n}(g) \pm \phi^{\beta,n}(g)\phi^{\alpha,n}(f)^\ast &= 0,
\end{align*}
if their respective supports are space-like separated. But this is only true iff \\
${\rm supp}\ f^{\{\hat{f}_0,\hat{f}_1\}} \cap {\rm supp}\ g^{\{\hat{f}_0,\hat{f}_1\}} = \emptyset$.

It would be ideal, that Lemmas \ref{l.w.3}, \ref{l.w.4} and \ref{l.w.5} hold, provided $0 \neq \vec{x} - \vec{y}$ is any space-like vector. But, this cannot be true in general. What we have shown instead, is that the local commutation and anti-commutation relations hold in a three dimensional subspace, containing space-like vectors.

%
%

\section{Yang-Mills path integrals}

We have completed the description of Wightman's axioms. We have seen that to satisfy local commutativity, the vector $\hat{H}(\rho_n)\hat{f}_0 + \hat{P}(\rho_n)\hat{f}_1$ must be time-like, thus $\hat{H}(\rho_n)^2 - \hat{P}(\rho_n)^2 = m_n^2 > 0$, for each $n \geq 1$. Local commutativity already implies the existence of a positive mass gap in each component Hilbert space $\mathscr{H}(\rho_n)$. Indeed, we will see that $\hat{H}(\rho_n)\hat{f}_0 + \hat{P}(\rho_n)\hat{f}_1$ defines a time-like\footnote{In subsection \ref{ss.cd}, we will see that the mass gap $m_n$ is the generator for translation in the $\tilde{f}_0^n$ direction.} vector $m_n \tilde{f}_0^n$, $\tilde{f}_0^n \cdot \tilde{f}_0^n = -1$, in Definition \ref{d.c.1}, which is crucial to prove the Clustering Theorem \ref{t.cl}.

But how do we choose $\{(\hat{H}(\rho_n), \hat{P}(\rho_n)): n \in \mathbb{N}\}$? In the next section, we will explain how we are going to compute the eigenvalues for the Hamiltonian and momentum operator. To prove the existence of a positive mass gap, we need to further show that $\inf_{n \in \mathbb{N}}m_n > 0$.

To do the quantization, we need to turn to Yang-Mills path integrals, which we will now summarize the construction done in \cite{YM-Lim01} and \cite{YM-Lim02}.

\subsection{Hermite polynomials}\label{ss.hp}

Consider the inner product space $\mathcal{S}_\kappa(\mathbb{R}^4)$, consisting of functions of the form $f  \sqrt{\phi_\kappa}$, whereby $\phi_{\kappa}(\vec{x}) = \kappa^4 e^{-\kappa^2|\vec{x}|^2/2}/(2\pi)^{2}$ is a Gaussian function and $f$ is a polynomial in $\vec{x} = (x^0, x^1, x^2, x^3) \in \bR^4$. Its inner product is given by \beq \left\langle f \sqrt{\phi_\kappa}, g \sqrt{\phi_\kappa} \right\rangle = \int_{\bR^4} f g \cdot \phi_\kappa\ d\lambda, \nonumber \eeq $\lambda$ is Lebesgue measure on $\mathbb{R}^4$.

Suppose $h_i/\sqrt{i}$ is a normalized Hermite polynomial of degree $i$ on $\bR$. Let $\overline{\mathcal{S}}_\kappa(\bR^4)$ be the smallest Hilbert space containing $\mathcal{S}_\kappa(\bR^4)$ and hence \beq \left\{ \frac{h_i(\kappa x^0)h_j(\kappa x^1)h_k(\kappa x^2)h_l(\kappa x^3)}{\sqrt{i!j!k!l!}} \sqrt{\phi_\kappa(\vec{x})}\ \Big|\ \vec{x} = (x^0, x^1, x^2, x^3) \in \bR^4,\ i,j,k,l \geq 0\right\} \nonumber \eeq forms an orthonormal basis. Note its dependence on $\kappa>0$.

Recall we chose the standard metric on $T\bR^4$, thus the volume form on $\bR^4$ is given by $d\omega = dx^0 \wedge dx^1 \wedge dx^2 \wedge dx^3$. Using the Hodge star operator and the above volume form, we will define an inner product on $\mathcal{\mathcal{S}}_\kappa(\bR^4)\otimes \Lambda^2(\bR^4)$ from Equation (\ref{e.x.7a}). Explicitly, it is given by \beq \left\langle \sum_{0\leq a< b \leq3} f_{ab} \otimes dx^a\wedge dx^b, \sum_{0\leq a< b \leq3} \hat{f}_{ab} \otimes dx^a\wedge dx^b \right\rangle = \sum_{0\leq a< b \leq3} \left\langle f_{ab}, \hat{f}_{ab} \right\rangle. \label{e.inn.5} \eeq

Write $\partial_a = \partial/\partial x^a$. Given $f = \sum_{i=1}^3 f_i \otimes dx^i \in \mathcal{S}_\kappa(\bR^4)\otimes \Lambda^1(\bR^3)$, the differential $df$ is given by \beq df = \sum_{i=1}^3\partial_0 f_i \otimes dx^0 \wedge dx^i + \sum_{1\leq i < j \leq 3}(\partial_i f_j - \partial_j f_i) dx^i \wedge dx^j. \label{e.diff.1} \eeq

\begin{defn}\label{d.g.1}
Recall $\{E^\alpha \in \mathfrak{g}:\ 1\leq \alpha \leq N\}$ is an orthonormal basis in $\mathfrak{g}$.
Define \beq c_{\gamma}^{\alpha \beta} = -\Tr\left[E^\gamma [E^\alpha, E^\beta] \right],\ E^\alpha, E^\beta, E^\gamma \in \mathfrak{g}. \nonumber \eeq The term $c_{\gamma}^{\alpha\beta}$ is referred to as the structure constant.
\end{defn}

\begin{prop}\label{p.p.1}
Suppose $A = \sum_{\alpha=1}^N \sum_{i=1}^3 a_{i, \alpha} \otimes dx^i \otimes E^\alpha \in \mathcal{S}_\kappa(\bR^4)\otimes \Lambda^1(\bR^3) \otimes \mathfrak{g}$. Write $a_{0:j, \gamma} = \partial a_{j,\gamma}/\partial x^0$ and $a_{i;j, \alpha} = \partial a_{i,\alpha}/\partial x^j- \partial a_{j,\alpha}/\partial x^i$. Then,
\begin{align}
dA + A \wedge A
=& \sum_{\gamma=1}^N\Bigg[\sum_{j=1}^3  a_{0:j, \gamma}\otimes  dx^0 \wedge dx^j +  \sum_{1\leq i<j \leq 3}a_{i;j, \gamma}\otimes dx^i \wedge dx^j \nonumber \\
&\hspace{1cm} +\sum_{1\leq i<j \leq 3}\sum_{1\leq \alpha,\beta \leq N}a_{i,\alpha}a_{j,\beta} c_\gamma^{\alpha \beta} \otimes dx^i \wedge  dx^j  \Bigg]\otimes E^\gamma. \label{e.g.3}
\end{align}
\end{prop}

\begin{proof}
By direct computation, using Equation (\ref{e.diff.1}).
\end{proof}

\subsection{Yang-Mills measure}\label{ss.ymm}

We will now summarize the following results proved in \cite{YM-Lim02}. After applying axial gauge fixing using the time-axis, the only $\mathfrak{g}$-valued gauge fields $A$ over $\bR^4$, which we need to consider are in $L^2(\bR^4)\otimes \Lambda^1(\bR^3) \otimes \mathfrak{g}$. Instead of making sense of a path integral over in the Hilbert space $L^2(\bR^4)\otimes \Lambda^1(\bR^3) \otimes \mathfrak{g}$, we will make sense of a Yang-Mills measure over $\mathcal{S}_\kappa(\bR^4)\otimes \Lambda^1(\bR^3) \otimes \mathfrak{g}$, of the form \beq \frac{1}{Z}e^{-\frac{1}{2}\int_{{\mathbb{R}^4}}|dA + A \wedge A|^2\ d\omega}  D[dA], \label{e.ym.1} \eeq whereby \beq Z =  \int_{\{dA:\ A \in \mathcal{S}_\kappa(\bR^4)\otimes \Lambda^1(\bR^3) \otimes \mathfrak{g}\}} e^{-\frac{1}{2}\int_{{\mathbb{R}^4}}|dA + A \wedge A|^2\ d\omega}  D[dA], \nonumber \eeq and $D[dA]$ is `Lebesgue measure' on $\mathcal{S}_\kappa(\bR^4)\otimes \Lambda^2(\bR^4) \otimes \mathfrak{g}$, which does not exist.

\begin{rem}
\begin{enumerate}
  \item This is analogous to how Balaban in his series of papers from 1984 to 1989, defined a finite dimensional integral over a finite lattice gauge of spacing $\epsilon$, synonymous with $1/\kappa$. A complete citation of Balaban's work can be found in \cite{YM-Lim01}.
  \item The norm $|\cdot|$ on $\mathcal{S}_\kappa(\bR^4)\otimes \Lambda^2(\bR^4) \otimes \mathfrak{g}$ is from the tensor inner product, taken from Equations (\ref{e.inn.5}) and (\ref{e.i.2}).
\end{enumerate}

\end{rem}

Using Equation (\ref{e.inn.5}) and Proposition \ref{p.p.1},
\begin{align}
\int_{\bR^4}|dA + A \wedge A|^2\ d\omega =& \sum_{1\leq i<j\leq 3}\int_{\bR^4}\Bigg[ \sum_{\alpha=1}^Na_{i;j, \alpha}^2 + \sum_{\gamma=1}^N \sum_{\myfrac{\alpha, \beta } {\hat{\alpha}, \hat{\beta}}}a_{i,\alpha}a_{j,\beta}a_{i,\hat{\alpha}}a_{j,\hat{\beta}}c_{\gamma}^{\alpha\beta}c_\gamma^{\hat{\alpha}\hat{\beta}} \nonumber \\
&+2\sum_{\gamma=1}^N\sum_{\alpha, \beta}a_{i;j,\gamma}a_{i,\alpha}a_{j,\beta}c_\gamma^{\alpha\beta} \Bigg]\ d\omega+ \sum_{j=1}^3\int_{\bR^4}\sum_{\alpha=1}^N a_{0:j,\alpha}^2\ d\omega.
\label{e.ym.2}
\end{align}

\begin{rem}
Those terms that contain the structure constants in Equation (\ref{e.ym.2}) will henceforth  be referred to as the interaction terms in the Yang-Mills Lagrangian.
\end{rem}

It is conventional wisdom to interpret \beq \exp\left[ -\frac{1}{2}\sum_{\alpha=1}^N\int_{\bR^4}d\omega\ \sum_{1\leq i<j \leq 3}a_{i;j, \alpha}^2 +  \sum_{j=1}^3a_{0:j,\alpha}^2  \right]D[dA] \nonumber \eeq as a Gaussian measure. To define this measure, we will work with holomorphic sections of 2-forms over $\bC^4$, instead of working on $\mathcal{S}_\kappa(\bR^4)\otimes \Lambda^2(\bR^4)$. We will go over the construction, done in \cite{YM-Lim01}.

Consider the real vector space spanned by $\{z^n:\ z \in \bC\}_{n=0}^\infty$, integrable with respect to the Gaussian measure, equipped with a real inner product, given by
\begin{align}
\langle z^{r},& z^{r'}\rangle
= \frac{1}{\pi}\int_{\mathbb{C}}   z^{r} \cdot \overline{z^{r'}}e^{-|z^2|}
\ dx\ dp,\ z = x + \sqrt{-1}p.  \label{e.l.2}
\end{align}
Note that $\overline{z}$ means complex conjugate. Denote this (real) inner product space by $H^2(\bC)$, which consists of polynomials in $z$. An orthonormal basis is hence given by \beq \left\{ \frac{z^n}{\sqrt{n!}}:\ n \geq 0 \right\}. \nonumber \eeq Consider the tensor product $H^2(\bC)^{\otimes^4}$, which we will also denote it as $H^2(\bC^4)$, to consist of all polynomials in $z = (z_0, z_1, z_2, z_3) \in \bC^4$, equipped with the tensor inner product $\langle \cdot, \cdot \rangle$ from Equation (\ref{e.l.2}). Denote the closure of $H^2(\bC^4)$ using this tensor inner product by $\mathcal{H}^2(\bC^4)$.

The Segal Bargmann Transform maps the inner product space $\mathcal{S}_\kappa(\bR^4)$ to $H^2(\bC^4)$, i.e. \beq \Psi_\kappa: \frac{h_i(\kappa \cdot)}{\sqrt{i!}}\frac{h_j(\kappa \cdot)}{\sqrt{j!}}\frac{h_k(\kappa \cdot)}{\sqrt{k!}}\frac{h_l(\kappa \cdot)}{\sqrt{l!}} \sqrt{\phi_\kappa} \longmapsto \frac{z_0^i}{\sqrt{i!}}\frac{z_1^j}{\sqrt{j!}}\frac{z_2^k}{\sqrt{k!}}\frac{z_3^l}{\sqrt{l!}}. \nonumber \eeq

Given any $f_{i,\alpha} \otimes dx^i \otimes E^\alpha \in \mathcal{S}_\kappa(\bR^4)\otimes \Lambda^1(\bR^3) \otimes \mathfrak{g}$, $f_{i,\alpha} \in \mathcal{S}_\kappa(\bR^4)$, we map it inside $H^2(\bC^4) \otimes \Lambda^1(\bR^3) \otimes \mathfrak{g}$ by \beq \Psi_\kappa: f_{i,\alpha} \otimes dx^i \otimes E^\alpha \longmapsto \Psi_\kappa(f_{i,\alpha}) \otimes dx^i \otimes E^\alpha. \nonumber \eeq We will term $\Psi_\kappa$ as a renormalization flow, mapping a continuous sequence of vector spaces, into a fixed vector space $H^2(\bC^4)\otimes \Lambda^1(\bR^3) \otimes \mathfrak{g}$.

We will further define an inner product on $\mathcal{H}^2(\bC^4)\otimes \Lambda^2(\bR^4)$ by \beq \left\langle \sum_{0\leq a< b \leq3} f_{ab} \otimes dx^a\wedge dx^b, \sum_{0\leq a< b \leq3} \hat{f}_{ab} \otimes dx^a\wedge dx^b \right\rangle = \sum_{0\leq a< b \leq3} \left\langle f_{ab}, \hat{f}_{ab} \right\rangle. \label{e.inn.8} \eeq Together with Equation (\ref{e.i.2}), we will continue to use $\langle\cdot, \cdot \rangle $ to denote the tensor inner product on $\mathcal{H}^2(\bC^4)\otimes \Lambda^2(\bR^4) \otimes \mathfrak{g}$.

\begin{defn}
For each $a = 0, 1, 2, 3$, define a linear operator $\fd_a$ acting on $H^2(\bC^4)$ as \beq \mathfrak{d}_a\left[z_a^p \prod_{\myfrac{b \neq a}{b=0,\cdots , 3}} z_b^{q_b} \right] = \left[\frac{p}{2}z_a^{p-1} - \frac{1}{2}z_a^{p+1} \right]\cdot \prod_{\myfrac{b \neq a}{b=0,\cdots , 3}} z_b^{q_b}. \nonumber \eeq
\end{defn}

The Segal Bargmann Transform maps $\partial_a$ to the linear operator $\kappa \fd_a$. As a consequence, we will have the following renormalization rule on $H^2(\bC^4)$. For each $\kappa > 0$, we will add an extra factor $\kappa$ to $\fd_a$, i.e. $\fd_a \mapsto \kappa \fd_a$.

In \cite{YM-Lim01}, we defined an operator $\fd: H^2(\bC^4)\otimes \Lambda^1(\bR^3) \rightarrow H^2(\bC^4)\otimes \Lambda^2(\bR^4)$, \beq \ \fd \sum_{i=1}^3 f_i \otimes dx^i = \sum_{i=1}^3[\fd_0 f_i]\otimes dx^0 \wedge dx^i + \sum_{1 \leq i < j \leq 3}[\fd_i f_j - \fd_j f_i] \otimes dx^i \wedge dx^j. \nonumber \eeq

Extend the Segal Bargmann Transform $\Psi_\kappa: \mathcal{S}_\kappa(\bR^4) \otimes \Lambda^2(\bR^4)\otimes \mathfrak{g} \rightarrow H^2(\bC^4) \otimes \Lambda^2(\bR^4) \otimes \mathfrak{g}$, i.e. \beq \Psi_\kappa:\ \sum_{1\leq i < j \leq 3}f_{i,j,\alpha} \otimes dx^i\wedge dx^j \otimes E^\alpha  \longmapsto \sum_{1\leq i < j \leq 3}\Psi_\kappa[f_{i,j,\alpha}] \otimes dx^i\wedge dx^j \otimes E^\alpha . \nonumber \eeq

Thus, $\Psi_\kappa[dA] = \kappa \fd \Psi_\kappa[A]$, due to the renormalization rule. As such, the Segal Bargmann Transform $\Psi_\kappa$ is an isometry (up to a constant $\kappa^2$) between their respective spaces, i.e. \beq \langle dA, dA \rangle = \kappa^2 \langle \fd \Psi_{\kappa}[A], \fd \Psi_\kappa[A] \rangle, \nonumber \eeq $A \in \mathcal{S}_\kappa(\bR^4) \otimes \Lambda^1(\bR^3 ) \otimes \mathfrak{g}$, using their respective tensor inner products.

Balaban used a renormalization flow to map a finite gauge lattice of spacing $\epsilon$, to a finite gauge lattice of unit spacing. The analogous procedure here would be to use the Segal Bargmann Transform, to define Expression \ref{e.ym.1} over in $H^2(\bC^4) \otimes \Lambda^2(\bR^4) \otimes \mathfrak{g}$, as \beq \frac{1}{Z}e^{-\frac{1}{2}\int_{\bC^4}|\kappa\fd A + A \wedge A|^2 d\lambda_4}  D[\fd A], \label{e.ym.3} \eeq whereby \beq Z =  \int_{\{\fd A:\ A \in H^2(\bC^4) \otimes \Lambda^1(\bR^3) \otimes \mathfrak{g}\}} e^{-\frac{1}{2}\int_{\bC^4}|\kappa\fd A + A \wedge A|^2 d\lambda_4}  D[\fd A], \nonumber \eeq and $D[\fd A]$ is `Lebesgue measure' on $H^2(\bC^4)\otimes \Lambda^2(\bR^3) \otimes \mathfrak{g}$, which does not exist. Note that $d\lambda_4$ is a 4-dimensional Gaussian measure on $\bC^4$, by using the tensor inner product from Equation (\ref{e.l.2}).

The factor $\kappa$ is due to the above renormalization rule. As a result of this factor, we complete the inner product space into a Banach space equipped with a Wiener measure of variance $1/\kappa^2$, as was constructed in \cite{YM-Lim01} and \cite{YM-Lim02}.

\begin{thm}\label{t.ym.1}
Let the span of $\{dx^0 \wedge dx^i, 1\leq i \leq 3\}$ and the span of $\{dx^i \wedge dx^j, 1\leq i < j\leq 3\}$ be denoted by $\ast \Lambda^2(\bR^3)$ and $\Lambda^2(\bR^3)$ respectively. Define
\begin{align*}
\mathbb{H} &:= \left\{ [\fd_0 H^2(\bC^4)] \otimes \left[\ast \Lambda^2(\bR^3) \right]\right\} \oplus \left\{ H^2(\bC^4) \otimes \Lambda^2(\bR^3)\right\} \\
&\subset H^2(\bC^4) \otimes \Lambda^2(\bR^4),
\end{align*}
$\fd_0 H^2(\bC^4)$ denotes the range of $\fd_0$.

In \cite{YM-Lim01}, we completed it into a Banach space $\dB$ using a supremum norm, equipped with a $\sigma$-algebra defined on it. On this Banach space, we constructed a product Wiener measure denoted as $\tilde{\mu}_{\kappa^2}$, which is an infinite dimensional Gaussian measure of variance $1/\kappa^2$, therefore making $(\dB, \tilde{\mu}_{\kappa^2})$ a probability space.

When we consider the tensor inner product space $\mathbb{H} \otimes \mathfrak{g}$, then the completion of it into a Banach space will be $\dB\otimes \mathfrak{g}$. The product Wiener measure on it, hereby denoted as $\tilde{\mu}_{\kappa^2}^{\times^{N}}$, will make $\dB\otimes \mathfrak{g}$ into a probability space. We will denote the expectation on this probability space using $\bE$.

Because of the renormalization flow $\{\Psi_\kappa: \kappa > 0\}$, we showed in \cite{YM-Lim02} that there exists a sequence of positive functions defined on the Wiener space $\dB \otimes \mathfrak{g}$, denoted by $\{\mathcal{Y}^\kappa: \kappa > 0\}$, such that we can define a measure as $\mathcal{Y}^\kappa d\tilde{\mu}_{\kappa^2}^{\times^{N}}$. Hence, we will interpret the Yang-Mills measure in Expression \ref{e.ym.3} as \beq \frac{1}{Z}e^{-\frac{1}{2}\int_{\bC^4}|\kappa \fd A + A \wedge A|^2 d\lambda_4}  D[\fd A] := \frac{\mathcal{Y}^\kappa d\tilde{\mu}_{\kappa^2}^{\times^{N}}}{\int_{\dB \otimes \mathfrak{g}} \mathcal{Y}^\kappa d\tilde{\mu}_{\kappa^2}^{\times^{N}}} = \frac{\mathcal{Y}^\kappa d\tilde{\mu}_{\kappa^2}^{\times^{N}}}{\bE[ \mathcal{Y}^\kappa ]}, \nonumber \eeq which is also a probability measure on the Banach space $\dB \otimes \mathfrak{g}$.
\end{thm}

\begin{notation}
For each $\kappa > 0$, the Yang-Mills measure is a probability measure, so it would be more convenient to use expectation to denote this integral. For any measurable and bounded function $F$ on $\dB \otimes \mathfrak{g}$, we will write \beq \mathbb{E}_{{\rm YM}}^\kappa[F] :=
\frac{1}{\int_{\dB \otimes \mathfrak{g}} \mathcal{Y}^\kappa\ d\tilde{\mu}_{\kappa^2}^{\times^{N}}}\int_{\dB \otimes \mathfrak{g}} F\mathcal{Y}^\kappa\ d\tilde{\mu}_{\kappa^2}^{\times^{N}}. \nonumber \eeq
\end{notation}

\begin{rem}
This sequence of Yang-Mills measure is analogous to how Balaban defined a sequence of finite-dimensional integrals on a finite lattice gauge, with spacing $\epsilon$, which were renormalized into finite lattice gauge with unit spacing. One should think of $\epsilon$ as synonymous to $1/\kappa$. Refer to \cite{YM-Lim01} and \cite{YM-Lim02} for details.
\end{rem}

We will term the vectors in $\dB \otimes \mathfrak{g}$ as Yang-Mills gauge fields. When we complete the space into a Banach space using a supremum norm, we showed in \cite{YM-Lim01} that this Banach space consists of holomorphic $\Lambda^2(\bR^4)\otimes \mathfrak{g}_\bC$-valued functions, over the complex space $\bC^4$. Or one can refer this Banach space as containing holomorphic sections of the complexified trivial bundle $\bC^4 \times [\Lambda^2(\bR^4)\otimes \mathfrak{g}_\bC] \rightarrow \bC^4$.

\subsection{Asymptotic freedom}

In physics, asymptotic freedom refers to the phenomenon whereby the coupling constant $c$ goes to zero as the momentum scale (or energy scale) increases. See page 425 in \cite{peskin1995}. This coupling constant, which depends on momentum, satisfies a certain differential equation. See page 459 in \cite{glimm1981}. Thus, it must vary continuously with momentum.

In physics literature, the interaction terms in Equation (\ref{e.ym.2}) is actually due to the non-linear term in Equation (\ref{e.g.3}). Physicists would introduce a coupling constant $c$ to the non-linear term, so that they can make the interaction terms small. See \cite{glimm1981}. When the coupling constant $c$ is zero, we obtain the free theory.

In a compact semi-simple Yang-Mills gauge theory, the structure constant $c^{\alpha\beta}_\gamma$ of a semi-simple Lie Algebra is non-zero, only if $\alpha$, $\beta$ and $\gamma$ are all distinct, and this introduces interaction terms into the Yang-Mills Lagrangian, making the Yang-Mills path integrals impossible to compute analytically. The coupling constant $c$ is introduced, so that one can apply perturbation theory to compute the path integrals. When $c$ is small, that means one can use Feynmann diagrams to compute the path integrals. This happens when the energy scale is large.

Unfortunately, perturbation methods are no longer valid if the energy scale is small. It is believed that non-abelian gauge theories exhibit asymptotic freedom. It was showed in \cite{PhysRevLett31851} that for renormalizable quantum field theories in 4-dimensional space-time, only non-abelian gauge theories are asymptotically free. On page 541 in \cite{peskin1995}, one sees that for a non-abelian gauge group, the coupling constant $c(k)^2$ decreases at a rate of $\dfrac{1}{\ln k^2}$, whereby $k$ is the momentum scale.

In \cite{YM-Lim02}, we set the coupling constant $c = 1/\kappa$, and we approximate the Dirac delta function with a Gaussian function $\kappa^4e^{-\kappa^2|\vec{x}|^2/2}/(2\pi)^2$, so the variance is given by $1/\kappa^2$. To resolve points separated by short distances, we need a small variance, which means that $\kappa$ is large. Therefore, at short distances, the coupling constant is small.

\subsection{Callan-Symanzik beta function}\label{ss.csb}

To determine how $\kappa$ varies with the energy scale, we need to impose a Callan-Symanzik Equation, by first introducing a beta function $\beta(c)$, which depends on the coupling constant. The purpose is to formulate renormalization conditions in the energy scale, instead of the renormalization scale $\kappa$.

In the next section, we will define the quadratic Casimir operator, for each irreducible representation $\rho_n: \mathfrak{g} \rightarrow {\rm End}(\bC^{\tilde{N}_n})$. Because the Lie Algebra is simple, the Casimir operator is a constant $C_2(\rho_n)$ times the identity. We will interpret the increasing sequence $\{\frac{1}{4}C_2(\rho_n):\ n \in \mathbb{N}\}$ as the set of quantized energy levels squared in the theory.

Instead of using $C_2(\rho_n)$, it is actually more appropriate to use $\tilde{N}_n$, the dimension of the representation $\rho_n$.
Hence, we will define the beta function as,
\beq \beta(c) = \frac{\partial c}{\partial [\ln \tilde{N}_n]}, \label{e.b.3} \eeq to determine how the coupling constant varies. Compare this with Equation 12.90 in \cite{peskin1995}. We will compute the $\beta$ function in the next section, using a path integral. Because asymptotic freedom holds, beta must be chosen to be negative.

\begin{rem}
The dimension $\tilde{N}_n$ is defined using the highest weights of an irreducible representation, from the Weyl dimension formula. As such, the set $\{\tilde{N}_n:\ n \in \mathbb{N}\}$ is unbounded. In the definition of the beta function, we treat $\tilde{N}_n$ as a continuous variable.
\end{rem}

\section{Hamiltonian and Momentum operator}\label{s.iop}

\subsection{Renormalization}\label{ss.pt}

\begin{notation}(Casimir operator)\label{n.co.1}\\
Let $\mathfrak{g}$ be a semi-simple Lie Algebra. For an irreducible representation $\rho: \mathfrak{g} \rightarrow {\rm End}(\bC^{\tilde{N}})$ such that $\rho(\mathfrak{g})$ consists of skew-Hermitian matrices, we define $C(\rho) \in \bR$ such that \beq \Tr[\rho(E^\alpha)\rho(E^\beta)] = C(\rho)\Tr[E^\alpha E^\beta]. \label{e.ck.2} \eeq

Also define \beq \mathscr{E}(\rho) := -\sum_{\alpha=1}^N \rho(E^\alpha)\rho(E^\alpha) \nonumber \eeq to be its (quadratic) Casimir operator. When $\mathfrak{g}$ is simple, the Casimir operator is a constant multiple of the identity. We write $C_2(\rho)$ to denote this constant.
\end{notation}

\begin{rem}\label{r.mg.2}
When $\mathfrak{g}$ is simple, note that $C(\rho)$ satisfies $C_2(\rho)\tilde{N} = NC(\rho)$.
\end{rem}

\begin{defn}\label{d.p.2}
Suppose $e^0$ is a directional vector in the time direction, with length $|e^0| = T$, and $a \in \bR^3$ be any directional spatial vector, with length $|a|$. Let $\sigma: I^2 \equiv [0,1]^2 \rightarrow \bR^4$ be some parametrization of a compact time-like rectangular surface $R[a,T]$ contained in a plane, spanned by $e^0$ and $a$, of dimensions height T and length $|a|$.

Explicitly, for some $\vec{c} \in \bR^4$, we can choose $\sigma(s,t) = \vec{c} + s e^0 + t a \in \bR^4$, $0\leq s, t \leq 1$. For any $\delta \geq 0$, write $I_\delta = [-\delta, 1 + \delta]$ and $I \equiv I_0$. We can extend the parametrization $\sigma$ to be defined on $I_\delta^2 \equiv I_\delta \times I_\delta$. We will write $R_\delta[a,T]$ to be the image of $I_\delta^2$ under $\sigma$, which is a compact rectangular surface containing $R[a,T] \equiv R_0[a,T]$. Note that the dimensions of $R_\delta[a,T]$ is $|a|(1 + 2\delta)$ by $T(1 + 2\delta)$. When $T = 1$, we will write $R_\delta[a,1] = R_\delta[a]$.

Define a $\rho(\mathfrak{g})$- valued random variable $\left(\cdot, \nu_{R_\delta[a, T]}^{\kappa,\rho} \right)_\sharp$, which sends \\ $\sum_{\alpha=1}^N \sum_{j=1}^3 \fd_0A_{j, \alpha} \otimes dx^0 \wedge dx^j \otimes E^\alpha \in \dB \otimes \mathfrak{g}$ to
\beq  \frac{1}{\kappa}\sum_{\alpha=1}^N\frac{\kappa^2}{4}\int_{\hat{s} \in [-\delta,1+\delta]^2}d\hat{s}\ \sum_{j=1}^3 |J_{0j}^{\sigma}|(\hat{s}) \kappa[\psi \cdot \fd_0A_{j, \alpha}](\kappa \sigma(\hat{s})/2)\otimes \rho(E^\alpha), \label{e.f.1} \eeq
whereby $\psi(w) := \frac{1}{\sqrt{2\pi}}e^{-|w|^2/2}$, and $|J^\sigma_{0i}|$ is defined in Definition \ref{d.r.1}.
\end{defn}

\begin{rem}\label{r.ren.1}
This $\rho(\mathfrak{g})$-valued random variable on $\dB \otimes \mathfrak{g}$ first appeared in \cite{YM-Lim02}. The factor $\psi(w)$ is known as a renormalization factor. Its importance was  explained in \cite{YM-Lim01}. The factor $1/\sqrt{2\pi}$ can be replaced with any number $0<\tilde{c} < 1/\sqrt2$.

The factors of $\kappa$ which appeared in Expression \ref{e.f.1} are all due to renormalization:
\begin{enumerate}
  \item The factor $\kappa$ in front of $\psi\cdot \fd_0A_{j, \alpha}$ is due to the renormalization rule.
  \item We embed $\bR^4$ inside $\bC^4$, by $\vec{x} \in \bR^4 \hookrightarrow \kappa \vec{x}/2$, thus the surface $R_\delta[a,T] \subset \bR^4$ is scaled to be $\kappa R_\delta[a,T]/2 \subset \bC^4$. Hence, the factor $\kappa/2$ in parenthesis is due to this renormalization transformation, which will also give us an extra factor $\kappa^2/4$ for the surface integral.
\end{enumerate}

The factor $1/\kappa$ outside of the integral, is due to asymptotic freedom, whereby we set the coupling constant $c = 1/\kappa$.
\end{rem}

For a gauge field $A \in H^2(\bC^4) \otimes \Lambda^1(\bR^3)\otimes \mathfrak{g} $, we can interpret \beq c\sum_{\alpha=1}^N\int_{\hat{s} \in I_\delta^2}d\hat{s}\ \sum_{j=1}^3 |J_{0j}^{\sigma}|(\hat{s}) \left[\fd_0A_{j, \alpha} \right]( \sigma(\hat{s}))\otimes \rho(E^\alpha), \label{e.c.1} \eeq as measuring the field strength of $\fd A + A \wedge A$, over the time-like rectangular surface $R_\delta[a,T] \subset \bR^4 \hookrightarrow \bC^4$. It has dimension of energy and it is similar to Expression \ref{e.f.1}, but without the factors of $\kappa$ and renormalization factor $\psi(w)$.

\begin{rem}\label{r.ren.2}
Note that we introduced the coupling constant $c$ to the quantity in Expression \ref{e.c.1}, instead to the non-linear term in Equation (\ref{e.g.3}). So, by applying the renormalization techniques and asymptotic freedom explained in Remark \ref{r.ren.1}, to $c\int_{R_\delta[a]}\fd A + A \wedge A$, we will obtain Expression \ref{e.f.1}.
\end{rem}

Write \beq \left\langle \nu_{R[a,T]}^{\kappa,\rho} \right\rangle^2 = -\bE \left[
 \left( \cdot, \nu_{R[a,T]}^{\kappa,\rho} \right)_\sharp^2  \mathcal{Y}^\kappa
\right] , \nonumber \eeq whereby $\mathcal{Y}^\kappa$ was defined in Theorem \ref{t.ym.1}. We interpret $\sqrt{\left\langle \nu_{R[a,T]}^{\kappa,\rho}\right\rangle^2}$ as measuring the average flux passing through the time-like rectangular surface $R[a,T]$, over a time interval $T$, using the non-abelian Yang-Mills measure, and hence it has dimension of energy. See also Equation 19.11 in \cite{Nair}.

\begin{rem}
Our term $\left(\fd A, \nu_{R[a,T]}^{\kappa, \rho} \right)_\sharp$ is actually a skew-Hermitian matrix. When we square it, it will be a non-positive definite matrix. We need to put a negative sign in front, to make it non-negative definite.
\end{rem}

However, we will instead consider for some $\delta > 0$, \beq  -\bE \left[
\left( \cdot, \nu_{R[a,T]}^{\kappa,\rho} \right)_\sharp \left( \cdot, \nu_{R_\delta[a,T]}^{\kappa,\rho} \right)_\sharp  \mathcal{Y}^\kappa\right], \nonumber \eeq which we showed in \cite{YM-Lim02}, is a good approximation to $\left\langle \nu_{R[a,T]}^{\kappa,\rho} \right\rangle^2$, when $\delta$ is small. But more importantly, we will understand it as a 2-point correlation Green's function, which we will use it in the Callan-Symanzik Equation.

For an irreducible representation $\rho$ for $\mathfrak{g}$ and $A = \sum_{\alpha=1}^N\sum_{i=1}^3 a_{i,\alpha}\otimes dx^i \otimes E^\alpha$, $a_{i,\alpha} \in \mathcal{S}_\kappa(\bR^4)$, we will write \beq A^\rho := \sum_{\alpha=1}^N\sum_{i=1}^3 a_{i,\alpha}\otimes dx^i \otimes \rho(E^\alpha) \in \mathcal{S}_\kappa(\bR^4) \otimes \Lambda^1(\bR^3)\otimes \rho(\mathfrak{g}). \nonumber \eeq

In \cite{{YM-Lim02}}, we made sense of a non-abelian Yang-Mills path integral using the renormalization flow $\{\Psi_\kappa: \kappa > 0\}$,
\begin{align*}
\frac{1}{Z}\int_{\{dA \in \mathcal{S}_\kappa(\bR^4) \otimes \Lambda^2(\bR^4)\otimes \mathfrak{g}\} } \exp \left[  c\int_{R[a]} d[A^\rho] \right] e^{-\frac{1}{2}S_{{\rm YM}}(A)}&\ D[dA] \\
&:= \bE_{{\rm YM}}^\kappa \Bigg[\exp\left[  \left( \cdot, \nu_{R[a]}^{\kappa,\rho} \right)_\sharp\right] \Bigg],
\end{align*}
whereby $D[dA]$ is some Lebesgue type of measure and \beq Z = \int_{\{dA \in \mathcal{S}_\kappa(\bR^4) \otimes \Lambda^2(\bR^4)\otimes \mathfrak{g}\} }e^{-\frac{1}{2}S_{{\rm YM}}(A)}\ DA. \nonumber \eeq Taking the limit as $\kappa \rightarrow \infty$, it will give us the Wilson Area Law formula.

And in the same article, we also showed that
\begin{align}
-\bE \left[
 \left( \cdot, \nu_{R[a]}^{\kappa,\rho_n} \right)_\sharp \left( \cdot, \nu_{R_\delta[a]}^{\kappa,\rho_n} \right)_\sharp
\mathcal{Y}^\kappa\right] = \frac{|a|}{4} \otimes \mathscr{E}(\rho_n) - \epsilon(n, \kappa),\label{e.mg.1}
\end{align}
whereby if $\kappa \geq \kappa_0$ for some $\kappa_0 \in \mathbb{N}$ dependent on $\delta$ but independent of $\rho$, the error term $\epsilon(n, \kappa)$ is a matrix, for which its trace is given by \beq \frac{\underline{c}}{\kappa^4}C(\rho_n) \leq \Tr\ \epsilon(n,\kappa) \leq \frac{\overline{c}}{\kappa^4}C(\rho_n), \label{e.mg.4} \eeq for positive constants $\underline{c}, \overline{c}$, independent of $\kappa$ and $\rho_n$.

\begin{rem}\label{r.a.4}
In the proof of Equation (\ref{e.mg.1}) in \cite{YM-Lim02}, we did a perturbation series expansion for $\mathcal{Y}^\kappa$, in terms of $1/\kappa$. This accounts for the $1/\kappa^4$ term in the trace estimate. We also proved in the said article, that the path integral in Equation (\ref{e.mg.1}) is continuously differentiable in $c=1/\kappa$, with its derivative bounded for $c > 0$.

\end{rem}

Due to the compactness of the gauge group, we can assume that $\rho_n(\mathfrak{g})$ consists of skew-Hermitian matrices, so the trace of the LHS of Equation (\ref{e.mg.1}) is positive. 
Besides using renormalization techniques and applying asymptotic freedom, the structure constants of a simple Lie Algebra and the quartic term in the Yang-Mills action, are all instrumental in proving this equality. Refer to the proof in \cite{{YM-Lim02}}.

\begin{rem}
In Equation (\ref{e.mg.1}), we see that $|a|(1+2\delta)^2$ is the area of the rectangular time-like surface $R_\delta[a]$. In general, we will obtain the area $|a|T(1+2\delta)^2$ for a time-like surface $R_\delta[a, T]$.

As such, the Yang-Mills path integral will give us the area density $d\rho$ given in Definition \ref{d.r.1}, which in turn allows us to construct a unitary representation of ${\rm SL}(2, \bC)$, which acts on space-like surfaces in $\bR^4$.

Recall we used the quantized values $\hat{H}(\rho_n)$, $\hat{P}(\rho_n)$ in Equation (\ref{e.p.3}). They will be defined via a Yang-Mills path integral, given by Equation (\ref{e.mg.1}). For a given space-like surface $S$, associated with it is a set $\{\hat{f}_0, \hat{f}_1\}$ contained in a Minkowski frame, spanning a time-like plane $S^\flat$. From Definition \ref{d.a.1}, this set can be transformed from $\{e_0, e_1\}$, a set spanning the $x^0-x^1$ plane. Without any loss of generality, we will quantize energy and momentum using time-like rectangular surfaces $R[a]$ and $R_\delta[a]$, contained in a plane parallel to the time-axis.
\end{rem}

\subsection{Callan-Symanzik Equation}\label{ss.cse}

To define our sequence of masses $\{m_n:\ n \in \mathbb{N}\}$, we need to determine the beta function, so that we can see how $\kappa$ correlates with $\tilde{N}_n$. The beta function is found typically by solving a Callan-Symanzik Equation.

The set $\{\rho_n(E^\alpha)/\sqrt{C(\rho_n)}\}_{\alpha=1}^N$ forms an orthonormal basis. If we replace each $\rho_n(E^\alpha)$ with $\rho_n(E^\alpha)/\sqrt{C(\rho_n)}$ in Equation (\ref{e.f.1}), we will obtain from Equation (\ref{e.mg.1}),
\begin{align}
-\frac{1}{C(\rho_n)}\bE \left[
 \left( \cdot, \nu_{R[a]}^{\kappa,\rho_n} \right)_\sharp \left( \cdot, \nu_{R_\delta[a]}^{\kappa,\rho_n} \right)_\sharp
\mathcal{Y}^\kappa\right] = \frac{|a|}{4} \otimes \frac{C_2(\rho_n)}{C(\rho_n)}\mathbb{I}_n - \frac{1}{C(\rho_n)}\epsilon(n, \kappa), \label{e.f.4}
\end{align}
whereby $\mathbb{I}_n$ is a $\tilde{N}_n \times \tilde{N}_n$ identity matrix for an irreducible representation $\rho_n: \mathfrak{g} \rightarrow {\rm End}(\bC^{\tilde{N}_n})$.

Note that $\dfrac{C_2(\rho)}{C(\rho)} = \dfrac{N}{\tilde{N}}$, which varies according to the representation $\rho$. Set $|a|=4$. Since $\{\rho_n(E^\alpha)/\sqrt{C(\rho_n)} \}_{\alpha=1}^N$ is an orthonormal basis for every $n \in \mathbb{N}$, from Equation (\ref{e.f.1}) and that $\mathcal{Y}^\kappa$ is independent of any representation used, it is not difficult to see that $\Tr\ \epsilon(n, \kappa) = C(\rho_n)\Tr\ \epsilon(1, \kappa)/C(\rho_1)$ for all $n \in \mathbb{N}$.

When we take the trace on the RHS of Equation (\ref{e.f.4}), we will obtain \beq \tilde{N}_n\frac{C_2(\rho_n)}{C(\rho_n)} - \frac{1}{C(\rho_n)}\Tr\ \epsilon(n, \kappa) \equiv N - \frac{1}{C(\rho_1)}\Tr\ \epsilon(1, \kappa). \nonumber \eeq

Write $\epsilon(\kappa) := \epsilon(1, \kappa)/C(\rho_1)$. Recall, $c = 1/\kappa$. For each $n \in \mathbb{N}$, we will define
\begin{align}
G_{n}^{(2)}(c, e) :=& \frac{\tilde{N}_n}{e} - \Tr\  \epsilon(1/c) \nonumber \\
=&\frac{\tilde{N}_n}{e} - c^4 \bar{\lambda} + f(c^5), \label{e.m.1}
\end{align}
whereby $\bar{\lambda}$ is some positive number independent of $\rho_n$, which can be computed from the triple and quartic term in the Yang-Mills action, and the remainder term is denoted by $f(c^5)$, which has a bounded derivative in $c$, independent of $n$. See Remark \ref{r.a.4}. Furthermore, there exists a constant $\tilde{C}_1$, independent of $n$, such that $|\bar{\lambda}| \leq \tilde{C}_1$, $|f(c^5)| \leq \tilde{C}_1c^5$ and its derivative $|d f(c^5)/dc| \leq \tilde{C}_1 c^4$. That we can write the path integral given by Equation (\ref{e.m.1}), with the stated properties, has been proved in \cite{YM-Lim02}.

\begin{rem}
The integral $G_{n}^{(2)}(c, \tilde{N}_n/N)$ is actually a 2-point correlation Green's function, which follows from taking the trace of Equation (\ref{e.f.4}).
\end{rem}

Recall the definition of the beta function, $e\dfrac{\partial c}{\partial e} = \beta(c)$, with $e \equiv \tilde{N}$. We will impose the following Callan-Symanzik Equation
\beq \left[e\frac{\partial }{\partial e} + \beta(c)\frac{\partial }{\partial c}+ 2\gamma(c)\right]G_{n}^{(2)}(c, e)= 0. \label{e.cs.1} \eeq
See Equation (12.41) in \cite{peskin1995}. This equation asserts that there exist two scalar functions $\beta(c)$ and $\gamma(c)$, related to the shifts in the coupling constant and the field strength, that compensates for the shift in the `new' renormalization scale $e$.

\begin{prop}
For $c$ small, there exists a scalar-valued function $\beta(c) = -c/4 + \lambda(c)$ and $\gamma = 1/2$ that solves Equation (\ref{e.cs.1}). Note that $|\lambda(c)| \leq \tilde{C}_4 c^2$, for some constant $\tilde{C}_4$ independent of $n$.
\end{prop}

\begin{proof}
Now \beq G_{n}^{(2)}(c, e) = \frac{\tilde{N}_n}{e} - c^4 \bar{\lambda} + f(c^5). \nonumber \eeq
A direct computation shows that
\beq \frac{\partial }{\partial e}G_{n}^{(2)}(c, e) = -\frac{\tilde{N}_n}{e^2} , \quad \frac{\partial }{\partial c}G_{n}^{(2)}(c, e) = -4c^3\bar{\lambda} + \tilde{f}(c^4), \nonumber \eeq $\tilde{f}(c^4) \equiv d f(c^5)/d c$ is scalar valued and $|\tilde{f}(c^4)| \leq \tilde{C}_2 c^4$, $\tilde{C}_2$ is a constant independent of $n$. 

Plug into Equation (\ref{e.cs.1}), we have that
\beq -\frac{\tilde{N}_n}{e} - 4\beta(c)c^3 \bar{\lambda} + 2\gamma(c)G_{n}^{(2)}(c, e) + \beta(c)\tilde{f}(c^4) = 0. \nonumber \eeq
To satisfy the above Callan-Symanzik Equation, we must have that $\beta(c) = -c/4+ \lambda(c)$, $\gamma(c) = 1/2 $, $\lambda(c)$ to be determined.

Hence, we have
\beq -\frac{c}{4}\tilde{f}(c^4) + f(c^5) -4c^3\lambda(c)\bar{\lambda} + \lambda(c)\tilde{f}(c^4) = 0. \nonumber\eeq

Recall we stated that $\bar{\lambda}$, which is independent of $n$, is positive. Furthermore, we have a constant $\tilde{C}_3$ such that \beq \frac{1}{c^4} |\tilde{f}(c^4)| + \frac{1}{c^5} |f(c^5)|  \leq \tilde{C}_3, \nonumber \eeq whereby this constant $\tilde{C}_3$ is independent of $n$.

Therefore, if $c$ is small enough, then $-4c^3\bar{\lambda} + \tilde{f}(c^4)$ is indeed non-zero. Hence, \beq \lambda(c) = \frac{1}{-4c^3\bar{\lambda} + \tilde{f}(c^4)}\left[\frac{c}{4}\tilde{f}(c^4) - f(c^5)\right], \nonumber \eeq is scalar valued, and is such that $|\lambda(c)| \leq \tilde{C}_4c^2$, $0< \tilde{C}_4$ is independent of $n$.
\end{proof}

\begin{defn}\label{d.b.1}
The Callan-Symanzik beta function is given as $\beta(c) = -\dfrac{c}{4} + \lambda(c)$, whereby $|\lambda(c)| \leq \tilde{C}_4 c^2$.
\end{defn}

This will allow us to correlate the coupling constant $c$ with the dimension $\tilde{N}_n$. Its solution will be found in the next subsection.

\subsection{Existence of positive mass gap}\label{ss.mg}

In Wightman's zeroth axiom, we note that $\hat{H}^2 - \hat{P}^2 = m^2 \geq 0$. Rewriting this equation, the Hamiltonian and momentum eigenvalue equation is equivalent to \beq \frac{\hat{P}^2}{\hat{H}^2} - 1 = -\frac{m^2}{\hat{H}^2}. \label{e.mg.5} \eeq In this equation, it should be understood that $\hat{P}$ and $\hat{H}$ are eigenvalues of their respective operators on $\bigoplus_{n=1}^\infty \mathscr{H}(\rho_n)$.

We now need to define momentum and Hamiltonian eigenvalues, that satisfy the above equation, in such a way that the operators are unbounded. We will see later that $m^2 \rightarrow \infty$ and $0 > \hat{P}^2/\hat{H}^2 -1 \rightarrow 0$, all implying their respective eigenvalues tend towards infinity.

The error term in Equation (\ref{e.mg.1}) comes from the interaction terms in the Yang-mills action in Equation (\ref{e.ym.2}). Equation (\ref{e.mg.4}) gives us its trace bound, and because of the quartic term, it also gives us a positive mass gap. Indeed, the trace of Equation (\ref{e.mg.1}) is actually a `continuous' version of the above eigenvalue equation.

Let us review our setup. We have a compact gauge group $G$ with a (real) simple Lie Algebra $\mathfrak{g}$, henceforth considered as a sub-Lie Algebra in $\mathfrak{u}(\bar{N})$, for which we can define an inner product on $\mathfrak{g}$. Hence, we will assume that $\rho(E)$ is skew-Hermitian, i.e. $-\rho(E) = \rho(E)^\ast$. We also let $\{\alpha_1, \cdots, \alpha_l\}$ be a simple system of roots for $\mathfrak{g}_\bC$. Furthermore, let $\{H_1, H_2, \cdots, H_l\}$ be a basis for a Cartan subalgebra $\mathfrak{h} \subset \mathfrak{g}$. It may not be orthonormal, so we can define an orthonormal set $\{E^\alpha\}_{\alpha=1}^l$ as \beq E^\alpha = \sum_{\beta=1}^la_{\alpha, \beta}H_\beta,\ \ 1 \leq \alpha \leq l. \nonumber \eeq  Define an invertible $l \times l$ matrix $B = \{a_{\alpha, \beta}\}_{\alpha,\beta=1}^l$. For any vector $u \in \bR^l$, there exists a constant $c>0$ such that $| Bu |_2^2 \geq c| u |_2^2$, $| \cdot |_2 $ is the standard Euclidean norm. Extend $\{E^\alpha:\ 1 \leq \alpha \leq l\}$ to be an orthonormal basis $\{E^\alpha\}_{\alpha=1}^N$ in $\mathfrak{g}$.

Each inequivalent irreducible representation is indexed uniquely by the highest weight \beq \lambda_\rho \equiv \left(\lambda_\rho(H_1), \cdots, \lambda_\rho(H_l) \right), \nonumber \eeq which is a $l$-tuple of non-negative half-integers or integers, i.e. $\lambda_\rho(H_i) \geq 0$. We will write $|\lambda_\rho|_2 := \left(\sum_{i=1}^l \lambda_\rho(H_i)^2  \right)^{1/2}$.

For each representation $\rho: \mathfrak{g} \rightarrow {\rm End}(\bC^{\tilde{N}})$, its corresponding Casimir operator is given by \beq \mathscr{E}(\rho) = -\sum_{\alpha=1}^N \rho(E^\alpha)\rho(E^\alpha) = C_2(\rho)I,\ \ C_2(\rho) \geq 0. \nonumber \eeq Since $\rho(E^\alpha)$ is skew-Hermitian, we have \beq \left\langle -\rho(E^\alpha)\rho(E^\alpha)v,v \right\rangle = \left\langle \rho(E^\alpha)v, \rho(E^\alpha)v \right\rangle \geq 0 \nonumber \eeq for any vector $v \in \bC^{\tilde{N}}$.

Because $\rho(H)$ is skew-Hermitian, its eigenvalues are purely imaginary. Let $v$ be a unit weight vector corresponding to the highest weight, i.e. $\langle \rho(H)v, \rho(H)v \rangle = \lambda_{\rho}(H)^2$, $H \in \mathfrak{h}$.

In terms of the highest weight, we have
\begin{align*}
\left\langle C_2(\rho)v,v \right\rangle &= \sum_{\alpha=1}^N \left\langle \rho(E^\alpha)v, \rho(E^\alpha)v \right\rangle
\geq \sum_{\alpha=1}^l \left| \sum_{\beta=1}^l a_{\alpha, \beta}\lambda_\rho(H_\beta)v \right|_2^2 \\
&= \sum_{\beta=1}^l\sum_{\gamma=1}^l\sum_{\alpha=1}^l \lambda_\rho(H_\beta) a_{\alpha, \beta}a_{\alpha, \gamma}\lambda_\rho(H_\gamma) \geq  c|\lambda_{\rho}|_2^2.
\end{align*}
Therefore, we see that $\{C_2(\rho_n): n \in \mathbb{N}\} $ is unbounded, because the highest weight is a $l$-tuple consisting of non-negative half-integers or integers.

We can thus list the inequivalent irreducible representations as a sequence \beq \left\{\rho_n: \mathfrak{g} \rightarrow {\rm End}(\bC^{\tilde{N}_n}) \right\}_{n=1}^\infty,\quad {\rm such\ that} \quad  0< C_2(\rho_n) \leq C_2(\rho_{n+1}),\nonumber \eeq for all $n \geq 1$.

Recall for each inequivalent irreducible non-trivial representation $\rho: \mathfrak{g} \rightarrow {\rm End}(\bC^{\tilde{N}})$, we defined a Hilbert space $\mathscr{H}(\rho)$. Using the above sequence, construct a Hilbert space $\{1\} \oplus \bigoplus_{n=1}^\infty \mathscr{H}(\rho_n)$, for which the Wightman's axioms are satisfied.

Define the Hamiltonian eigenvalue $\hat{H}(\rho_n)$ to be \beq \hat{H}(\rho_n)^2 := \frac{\tilde{N}_n}{4}C_2(\rho_n) = \frac{N}{4}C(\rho_n) > 0,\nonumber \eeq for each irreducible non-trivial representation $\rho_n: \mathfrak{g} \rightarrow {\rm End}(\bC^{\tilde{N}_n})$. See Remark \ref{r.mg.2}.

When $n \in \mathbb{N}$ is large, we see that the energy level is correspondingly large. By asymptotic freedom, the coupling constant $c$ should weaken with large values of $n$. Since $c = 1/\kappa$, this means that $\kappa$ must be large when $n$ is large. To define the momentum eigenvalues, we need to choose an unbounded sequence $\{\kappa_n: n \in \mathbb{N}\}$ for Equation (\ref{e.mg.1}).

Recall we solved for the beta function in subsection \ref{ss.cse}. By Definition \ref{d.b.1}, we have that \beq \frac{\partial c}{\partial[\ln \tilde{N}]} = -\frac{c}{4} + \lambda(c),\ \ |\lambda(c)| \leq \tilde{C}_4c^2 . \nonumber \eeq

Now,
\begin{align*}
&\frac{dc}{d[ \ln \tilde{N}]} = -\frac{c}{4} + \lambda(c) \quad
\Longrightarrow \frac{dc}{c- 4\lambda(c)} = -\frac{d [\ln \tilde{N}]}{4}.
\end{align*}
Write $\mu(c) = -4\lambda(c)/c$. Thus,
\begin{align*}
&\frac{1}{c}\frac{dc}{1+ \mu(c)} = -\frac{d [\ln \tilde{N}]}{4} \quad
\Longrightarrow \left[\frac{1}{c}\sum_{k=0}^\infty (-1)^k\mu(c)^k \right]dc= -\frac{d [\ln \tilde{N}]}{4},
\end{align*}
whereby $\mu(c)$ is such that $|\mu(c)| \leq \tilde{C}_5 c$. Integrate, we obtain \beq \ln c + \tilde{\mu}(c) = -\frac{1}{4}\ln \tilde{N} + C, \nonumber \eeq $C$ is some constant. Note that $\tilde{\mu}(c)$ is a scalar function such that $|\tilde{\mu}(c)| \leq \tilde{C}_6 c$. Exponentiate, we have \beq c e^{\tilde{\mu}(c)} = \frac{1}{\hat{C}} \tilde{N}^{-1/4}\ \Longrightarrow\ c(1 + \bar{\mu}(c)) = \frac{1}{\hat{C}} \tilde{N}^{-1/4}, \nonumber \eeq whereby $0<\hat{C}$ is some positive constant and $\bar{\mu}(c) = e^{\tilde{\mu}(c)}-1$ is such that $|\bar{\mu}(c)| \leq \tilde{C}_7 c$. Thus, we have that $c\tilde{N}^{1/4} = \dfrac{1}{\hat{C}}\dfrac{1}{1 + \bar{\mu}(c)}$ and
\begin{align*}
&\left| \frac{1}{\hat{C}}\frac{1}{1 + \bar{\mu}(c)} \right| \leq \frac{1}{\hat{C}}\left[ 1 + \tilde{C}_8|\bar{\mu}(c)| \right] \leq \frac{1}{\hat{C}}\left[ 1 + \tilde{C}_8 \tilde{C}_7 c \right] \\
&\leq \frac{1}{\hat{C}}\left[ 1 + \tilde{C}_7\tilde{C}_8 \frac{1}{\hat{C}}\tilde{N}^{-1/4}\left( 1 + \tilde{C}_8 \tilde{C}_7  \right) \right] = \frac{1}{\hat{C}} + \frac{\tilde{N}^{-1/4}}{\hat{C}^2}\tilde{C}_7\tilde{C}_8\left(1 + \tilde{C}_7\tilde{C}_8\right).
\end{align*}
Thus \beq \tilde{N}^{1/4}|\bar{\mu}(c)| \leq c\tilde{N}^{1/4}\tilde{C}_7 \leq \frac{\tilde{C}_7}{\hat{C}} + \frac{\tilde{C}_7}{\hat{C}^2}\tilde{C}_7\tilde{C}_8\left(1 + \tilde{C}_7\tilde{C}_8\right). \nonumber \eeq

Recall $c = 1/\kappa$. Hence,
\begin{align*}
\kappa =& \frac{1}{c} = \hat{C}\tilde{N}^{1/4}(1 + \bar{\mu}(c)) \\
=& \hat{C}\tilde{N}^{1/4} + \tilde{R}(c),
\end{align*}
whereby the remainder term $\tilde{R}(c)$ is such that \beq |\tilde{R}(c)| \leq \tilde{C}_7 + \frac{\tilde{C}_7}{\hat{C}}\tilde{C}_7\tilde{C}_8\left(1 + \tilde{C}_7\tilde{C}_8\right) = \tilde{C}_7 + \frac{1}{\hat{C}}\tilde{C}_9. \nonumber \eeq In the above calculations, $\tilde{C}_k$ are positive constants, $5 \leq k \leq 9$, all independent of $\tilde{N}$. And we see that if $\hat{C} > 1$, then we have that $|\tilde{R}(c)| \leq \tilde{C}_7 + \tilde{C}_9$, independent of $\hat{C}$ and $\tilde{N}$.

We will henceforth define for $n \in \mathbb{N}$, $\kappa_n := \hat{C}\tilde{N}_n^{1/4}+ \hat{C}_n $, for some fixed positive constant $\hat{C}>1$ and $|\hat{C}_n| \leq \tilde{C}_7 + \tilde{C}_9$ for all $n \in \mathbb{N}$.

Set $|a|=1$. Start with the representation $\rho_1: \mathfrak{g} \rightarrow {\rm End}(\bC^{\tilde{N}_1})$. We choose a constant $\hat{C}$ large enough, such that when we plug in $\hat{C}\tilde{N}_n^{1/4} + \hat{C}_1 = \kappa_1 \geq \kappa_0$ into Equation (\ref{e.mg.1}), we have \beq 0< m_1^2  := \Tr\ \epsilon(1,\kappa_1)  = \frac{\tilde{N}_1}{4} C_2(\rho_1) + \Tr\ \bE \left[
 \left( \cdot, \nu_{R[a]}^{\kappa_1,\rho_1} \right)_\sharp \left( \cdot, \nu_{R_\delta[a]}^{\kappa_1,\rho_1} \right)_\sharp
\mathcal{Y}^\kappa\right] , \nonumber \eeq $\Tr$ is a matrix trace. Since the mass gap $m_1 > 0$, we define the quantized momentum eigenvalue $\hat{P}(\rho_1)^2 := \hat{H}(\rho_1)^2 - m_1^2$.

Take the trace on Equation (\ref{e.mg.1}), \beq -\Tr\ \bE \left[
 \left( \cdot, \nu_{R[a]}^{\kappa,\rho_n} \right)_\sharp \left( \cdot, \nu_{R_\delta[a]}^{\kappa,\rho_n} \right)_\sharp
\mathcal{Y}^\kappa\right]  = \frac{\tilde{N}_n}{4}C_2(\rho_n) - \Tr\ \epsilon(n,\kappa) . \label{e.mg.6} \eeq

Rewrite this, we will obtain
\beq \frac{4}{NC(\rho_n)}\Tr\ \bE \left[
- \left( \cdot, \nu_{R[a]}^{\kappa,\rho_n} \right)_\sharp \left( \cdot, \nu_{R_\delta[a]}^{\kappa,\rho_n} \right)_\sharp
\mathcal{Y}^\kappa\right]  - 1 =   -\frac{4\Tr\ \epsilon(n,\kappa) }{NC(\rho_n)} . \nonumber \eeq Compare this with Equation (\ref{e.mg.5}). Note that we made use of Remark \ref{r.mg.2}. It remains to plug in $\kappa = \kappa_n$ in the RHS of the equation, to determine the momentum eigenvalues.

Consider $n \geq 2$. For any irreducible representation $\rho_n: \mathfrak{g} \rightarrow {\rm End}(\bC^{\tilde{N}_n})$, we define $m_n$, such that \beq m_n^2 := \Tr[\epsilon(n,\kappa_n)] . \label{e.mg.8} \eeq

From Equation (\ref{e.mg.4}), we have that for all $n \in \mathbb{N}$, \beq  0< \frac{\underline{c}}{\kappa_n^4}<\frac{1}{C(\rho_n)}\Tr[\epsilon(n,\kappa_n)] \leq \frac{\bar{c}}{\kappa_n^4}, \label{e.mg.7} \eeq for positive constants $\underline{c}, \bar{c}$ independent of $\rho_n$.

Hence, the trace $\Tr[\epsilon(n,\kappa_n)]$ is of the order $C(\rho_n)/\kappa_n^4$. In terms of the dimension scale, we see that it is of the order $C(\rho_n)/\hat{C}^4\tilde{N}_n \equiv C_2(\rho_n)/N\hat{C}^4$. From Equation (\ref{e.mg.7}), we must have that $m_n^2$ is of the order $\dfrac{C_2(\rho_n)}{ \hat{C}^4} > 0$, which tends to infinity, as $n \rightarrow \infty$, because $\{C_2(\rho_n): n \in \mathbb{N}\}$ is unbounded,

Define the quantized momentum eigenvalue, in the direction of a unit vector $a$ as
\begin{align*}
\hat{P}(\rho_n)^2 :=& -\Tr\ \bE \left[
 \left( \cdot, \nu_{R[a]}^{\kappa_n,\rho_n} \right)_\sharp \left( \cdot, \nu_{R_\delta[a]}^{\kappa_n,\rho_n} \right)_\sharp
\mathcal{Y}^\kappa\right] \\
=& \frac{\tilde{N}_n}{4}C_2(\rho_n) - \Tr\ \epsilon(n,\kappa_n)
= \frac{\tilde{N}_n}{4} C_2(\rho_n) - m_n^2 > 0,
\end{align*}
from Equations (\ref{e.mg.6}) and (\ref{e.mg.8}). Thus, the average of the flux through the time-like rectangular surface $R[a]$ using the Yang-Mills measure, quantize the momentum in the direction $a$.

Note that $m_n^2/\hat{H}(\rho_n)^2$, is of the order $\dfrac{1}{\tilde{N}_n\hat{C}^4} > 0$. Hence, we see that
the momentum operator eigenvalues will go to infinity, from \beq \frac{\hat{P}(\rho_n)^2}{\hat{H}(\rho_n)^2} = 1   -\frac{m_n^2}{\hat{H}(\rho_n)^2} \longrightarrow 1, \nonumber \eeq as $n \rightarrow \infty$.

\begin{rem}\label{r.mg.1}
Note that we interpret $\tilde{N}C_2(\rho)$ as energy squared. Because of the beta function, we see that $\kappa$ increases with the dimension of the irreducible representation. In the case of ${\rm SU}(2)$ or ${\rm SU}(3)$, we see that the Casimir operator is large, when the dimension of the representation is large. More generally, the Weyl dimension formula says that when the dimension of the representation $\rho$ is large, then $|\lambda_\rho|$ will be large, which also implies that the Casimir constant $C_2(\rho)$ will be large. Since the coupling constant $c = 1/\kappa$, the coupling constant weakens at high energies, which the physicists will term as asymptotic freedom.
\end{rem}

\begin{defn}(Hamiltonian and momentum operator)\label{d.ma.1}\\
For each $n \in \mathbb{N}$, let $\rho_n$ be an irreducible non-trivial representation for a simple Lie Algebra $\mathfrak{g}$. Recall we defined $\hat{H}(\rho_n)^2= \frac{\tilde{N}_n}{4}C_2(\rho_n)>0$ and $\hat{P}(\rho_n)$ such that
\beq \hat{H}(\rho_n)^2 -  \hat{P}(\rho_n)^2 = m_n^2 > 0, \nonumber \eeq for some positive mass gap $m_n > 0$ defined by Equation (\ref{e.mg.8}).

From Definition \ref{d.a.1}, we have a basis for $\bR^4$. Define a Hamiltonian $\hat{H}(\vec{a}, \rho)$, in the direction $\vec{a}  \in \bR^4$, as $\hat{H}(\vec{a}, \rho) := (\vec{a}\cdot \hat{f}_0 ) \hat{H}(\rho)\hat{f}_0$. Explicitly, if $\vec{a} = \sum_{b=0}^3 a^b \hat{f}_b$, we will write $\hat{H}(\vec{a}, \rho) := -a^0\hat{H}(\rho)\hat{f}_0$.

We define a momentum operator, $\hat{P}(\vec{a}, \rho)$, in the direction $\vec{a} \in \bR^4$, as $\hat{P}(\vec{a}, \rho) := (\vec{a}\cdot \hat{f}_1) \hat{P}(\rho)\hat{f}_1$. Explicitly, if $a = \sum_{b=0}^3 a^b \hat{f}_b$, we will write $\hat{P}(\vec{a}, \rho) := a^1\hat{P}(\rho)\hat{f}_1$. Thus
\begin{align*}
U(\vec{a},1)&\left(S, f_\alpha \otimes \rho(E^\alpha), \{\hat{f}_a\}_{a=0}^3\right) \\
:=& e^{-i[\hat{f}_0\cdot\hat{H}(\vec{a}, \rho) + \hat{f}_1\cdot \hat{P}(\vec{a},\rho)]}\left(S+\vec{a}, f_\alpha(\cdot - \vec{a}) \otimes \rho(E^\alpha), \{\hat{f}_a\}_{a=0}^3\right) \\
=& e^{i[a^0 \hat{H}(\rho) -  a^1\hat{P}(\rho)]}\left(S+\vec{a}, f_\alpha(\cdot - \vec{a}) \otimes \rho(E^\alpha), \{\hat{f}_a\}_{a=0}^3\right).
\end{align*}
\end{defn}

We also define $\hat{H}(\rho_0) = \hat{P}(\rho_0) = 0$, thus for the vacuum state, $\hat{P}1 = \hat{H}1 = 0$. The momentum operator $\hat{P}$ and Hamiltonian $\hat{H}$ act on the Hilbert space $\mathbb{H}_{{\rm YM}}(\mathfrak{g})$, as
\begin{align*}
\hat{P}\sum_{n=0}^\infty v_n := \sum_{n=1}^\infty \hat{P}(\rho_n)v_n,\quad
\hat{H}\sum_{n=0}^\infty v_n := \sum_{n=1}^\infty \hat{H}(\rho_n)v_n ,
\end{align*}
provided the sum converges. Because $\{C_2(\rho_n)\}_{n \geq 1}$ is unbounded and the eigenvalues are real, we see that $\hat{P}$ and $\hat{H}$ are unbounded self-adjoint operators.

\begin{thm}(Positive mass Gap Theorem)\\
Consider the Hilbert space defined in Equation (\ref{e.h.2}). The quantized momentum operator $\hat{P}$ and Hamiltonian $\hat{H}$ are non-negative, unbounded self-adjoint operators on this Hilbert space. They annihilate the vacuum state 1 and their other eigenspaces are $\mathscr{H}(\rho_n)$, $n \geq 1$, with corresponding eigenvalues $\hat{P}(\rho_n)$ and $\hat{H}(\rho_n)$ respectively, such that, $\hat{H}(\rho_n)^2 - \hat{P}(\rho_n)^2 = m_n^2$ for some positive mass gap $m_n>0$.

These eigenvalues were chosen because of the Callan-Symanzik Equation, and we have that $\lim_{n \rightarrow \infty}m_n^2 = \infty$. Thus $m_0 := \inf_{n \in \mathbb{N}}m_n > 0$, showing the existence of a positive mass gap in a 4-dimensional quantum compact and simple Yang-Mills gauge theory.
\end{thm}

\begin{rem}
In \cite{ALBEVERIO197439}, they explicitly state that the mass gap refers to the gap in the spectrum of the Hamiltonian operator. But this cannot be correct, as the mass gap refers to the difference between the Hamiltonian and the momentum operator. A strictly positive mass gap is required to show the short range nature of the weak force.
\end{rem}

Our discussion on the construction of a 4-dimensional quantum Yang-Mills simple and compact gauge group is now complete. Does our construction apply to an abelian gauge group? Answer is no.

When we have an abelian group ${\rm U}(1)$, we can define a non-trivial representation $F_a: u \in  \mathfrak{u}(1) \mapsto 0 \neq a \in \bR$, for any $a$. It is clear that $F_a$ is a representation of $\mathfrak{u}(1)$. Essentially, there is only one irreducible representation of $\mathfrak{u}(1)$, up to a constant. As a consequence, using the construction as outlined, the momentum operator and Hamiltonian, will not be unbounded operators.

An abelian gauge group describes the quantum electromagnetic force. As photons are massless, we see that $(\hat{H}, \hat{P}, 0, 0)$ must be a null-vector. As a result, the commutation and anti-commutation relations discussed in Section \ref{s.c} will only hold on a 2-dimensional subspace.

A positive mass gap does not exist for an abelian gauge group ${\rm U}(1)$. The existence of a mass gap is due to the quartic term of the interaction term in the Yang-Mills action. And this gives us Equation (\ref{e.mg.1}), proved in \cite{YM-Lim02} using asymptotic freedom and the structure constants of the Lie Algebra $\mathfrak{g}$, which do not apply to an abelian gauge group. Incidently, it will also give us the area law formula. Because of the absence of asymptotic freedom, the area law formula will not hold in the case of an abelian gauge group, as was shown in \cite{YM-Lim01}.

\section{Clustering}\label{s.cluster}

One important feature of the weak and strong force is that both are short range. This is formulated mathematically as the Clustering Theorem. To further validate that our construction of a 4-dimensional quantum field theory over in $\bR^4$, satisfying Wightman's axioms with a mass gap is correct, we will devote this section to the proof of the Clustering Theorem.

The first version was proved by Ruelle in \cite{Ruelle}. The proof that it decays exponentially is taken from \cite{Araki1962}. Our main reference for the Cluster Decomposition Property will be from \cite{streater}, and we will follow closely the notations used in there.

\begin{rem}
We will use the standard orthonormal basis $\{e_a\}_{a=0}^3$, whereby $\{e_2, e_3\}$ span $S_0$, and $\{e_0, e_1\}$ will span a time-like plane $S_0^\flat$. The coordinates used in this section will be pertaining to the standard orthonormal basis. Thus, $\vec{a} = (a^0, a^1, a^2, a^3)$ means the vector $\sum_{b=0}^3 a^b e_b \in \bR^4$.
\end{rem}

\begin{defn}(Fourier Transform)\label{d.f.1}\\
We define the Fourier Transform of a $L^2$ function $f: \bR \rightarrow \bR$ by \beq \hat{f}(p) = \mathcal{F}[f](p) := \frac{1}{\sqrt{2\pi}}\int_{\bR} e^{-ip x}f(x)\ dx. \nonumber \eeq

On $\bR^4$, our Fourier Transform is given by \beq \hat{f}(\vec{p}) = \mathcal{F}[f](\vec{p}) := \frac{1}{(2\pi)^2}\int_{\bR^4} e^{-i\vec{p}\cdot \vec{x}}f(\vec{x})\ d\vec{x}. \nonumber \eeq See Equation (\ref{e.f.5a}).
\end{defn}

Recall we fixed a surface $S_0$, which is the $x^2-x^3$ plane, and we have the quantum state $\Big(S_0, f_\alpha \otimes \rho(E^\alpha), \{e_a\}_{a=0}^3\Big)$. But how does one understand it? From the definition of the field operators in Section \ref{s.qfo}, we see that $f_\alpha$ is synonymous with $F_\alpha\left(\hat{H}(\rho), \hat{P}(\rho), x^2, x^3 \right)$, whereby $F_\alpha\left(\hat{H}(\rho), \hat{P}(\rho), \cdot, \cdot \right): S_0 \rightarrow \bC$, is a Schwartz function on $S_0$.

Let $\hat{F}_\alpha$ be the Fourier Transform of $F_\alpha$. By taking the Fourier Transform of $F_\alpha$ over in $S_0$,
\begin{align*}
\left(\frac{1}{\sqrt{2\pi}}\right)^2&\int_{S_0} e^{-i[q^2 x^2 + q^3 x^3]}F_\alpha\left(\hat{H}(\rho), \hat{P}(\rho), x^2, x^3 \right)\ dx^2dx^3 \otimes \rho(E^\alpha) \\
=& \hat{F}_\alpha\left(\hat{H}(\rho), \hat{P}(\rho), q^2, q^3 \right) \otimes \rho(E^\alpha).
\end{align*}
Hence one can view the quantum state as an operator-valued tempered distribution, over in momentum-energy space.

\begin{rem}\label{r.m.1}
The momentum coordinates $\left(\hat{H}(\rho), \hat{P}(\rho), q^2, q^3 \right)$ are with respect to the basis $\{e_a\}_{a=0}^3$. In general, if $S$ is a space-like plane equipped with a Minkowski frame $\{\hat{f}_a\}_{a=0}^3$, then the Fourier Transform will yield $\hat{F}_\alpha\left(\hat{H}(\rho), \hat{P}(\rho), q^2, q^3 \right) \otimes \rho(E^\alpha)$, but with the coordinates pertaining to the basis $\{\hat{f}_a\}_{a=0}^3$.

Thus, it is clear that $\hat{F}_\alpha$ depends on the basis $\{\hat{f}_a\}_{a=0}^3$. We wish to point out that it is also dependent on $\vec{a} = a^0 \hat{f}_0 + a^1 \hat{f}_1$ in space-time, which is a position vector in the space-like plane $S$. This is because its Fourier Transform is
\begin{align}
\hat{F}&\left(\hat{H}(\rho)\hat{f}_0 + \hat{P}(\rho)\hat{f}_1 + q^2 \hat{f}_2 + q^3 \hat{f}_3\right) \nonumber\\
&:= \frac{e^{-i[a^0\hat{H}(\rho) - a^1\hat{P}(\rho)]}}{2\pi}\int_{\hat{s} \in \bR^2} e^{-i(sq^2 + \bar{s}q^3)}f^{\{\hat{f}_0, \hat{f}_1\}}(\hat{H}(\rho), \hat{P}(\rho))(s \hat{f}_2 + \bar{s} \hat{f}_3)\ d\hat{s} \label{e.f.5a}\\
&\equiv e^{-i[a^0\hat{H}(\rho) - a^1\hat{P}(\rho)]}\hat{f}\left(\hat{H}(\rho)\hat{f}_0 + \hat{P}(\rho)\hat{f}_1 + q^2 \hat{f}_2 + q^3 \hat{f}_3\right), \nonumber
\end{align}
from Equation (\ref{e.p.3}). Also refer to Item \ref{r.f.3c} in Remark \ref{r.f.3}.

Hence, $\hat{F}_\alpha\left(\hat{H}(\rho)\hat{f}_0 + \hat{P}(\rho)\hat{f}_1 + q^2 \hat{f}_2 + q^3 \hat{f}_3\right) \otimes \rho(E^\alpha)$ allows us to reconstruct \\
$\left( S, f_\alpha \otimes \rho(E^\alpha), \{\hat{f}_a\}_{a=0}^3 \right)$.

For a Schwartz function $f \in \mathscr{P}$, ${\rm supp}\ \hat{f}$ may or may not lie in the positive light cone in energy-momentum space. Thus, $\hat{H}(\rho)^2 - \hat{P}(\rho)^2 - q^{2,2} - q^{2,3}$ may not be greater than or equal to 0. But this does not matter, as we are still able to prove the Clustering Decomposition Property in Theorem \ref{t.cd.1}. Even though it is not required or necessary to impose the condition that $\hat{H}(\rho)^2 - \hat{P}(\rho)^2 - q^{2,2} - q^{2,3} \geq 0$, this train of thought is not quite correct.

In the proof of Clustering Theorem \ref{t.cl}, we will see that by a suitable choice of a Minkowski frame, we will instead use $\left(m, q^1, q^2, q^3\right)$ as the coordinates in a 4-dimensional mass-momentum space after taking Fourier Transform. The translation operator in a time-like direction will be generated by mass, not energy. The momentum coordinates $q^i$'s will be free, and the existence of a positive mass gap $m_0$ says that the total energy is given by $m^2 + \sum_{i=1}^3 |q^i|^2 \geq m_0^2$.
\end{rem}

Refer to Definition \ref{d.ma.1}. When we take the Fourier Transform of \\
$U(\vec{a}, 1)\left(S_0, f_\alpha \otimes \rho(E^\alpha), \{e_a\}_{a=0}^3 \right)$, a similar set of calculations will give us \beq  e^{-i \vec{a} \cdot \vec{\alpha}}e^{-i[q^2 a^2 + q^3 a^3]}\hat{F}_\alpha\left(\hat{H}(\rho),\hat{P}(\rho), q^2, q^3 \right)  \otimes \rho(E^\alpha), \nonumber \eeq $\vec{\alpha} = (\hat{H}(\rho), \hat{P}(\rho), 0, 0)$.

Recall, $\{0\} \cup \{\hat{H}(\rho_n), \hat{P}(\rho_n)\}_{n\geq 1}$ are the discrete eigenvalues of the Hamiltonian and momentum operator. By the SNAG theorem, we see that when $\vec{a} \in S_0^\flat$, i.e. $\vec{a} = a^0e_0 + a^1 e_1$, the spectrum of the translation operator $U(\vec{a}, 1)$ on $\mathbb{H}_{{\rm YM}}(\mathfrak{g})$ is $\{1\} \cup \{ e^{i [a^0 \hat{H}(\rho_n) - a^1\hat{P}(\rho_n)]}\}_{n \geq 1}$, which is discrete.

Now, recall $S_0^\flat$ is the $x^0-x^1$ plane and we have the field operator $\phi^{\alpha,n}(f)$, $f \in \mathscr{P}$, defined in Definitions \ref{d.co.1} and \ref{d.co.2}. We can define a distribution, denoted as $\left\langle \phi^{\alpha,n}(\vec{x})1, \phi^{\beta,n}(\vec{y} )1 \right\rangle$, via sending $(f, g)\in \mathscr{P} \times \mathscr{P}$ to \beq \left\langle \phi^{\alpha,n}(f)1, \phi^{\beta,n}(g)1 \right\rangle := C(\rho_n)\Tr[-F^\alpha F^\beta]\int_{S_0} \left[f^{\{e_0, e_1\}}\overline{g^{\{e_0, e_1\}}}\right](\hat{H}(\rho_n), \hat{P}(\rho_n))(\hat{s})\ d\hat{s}. \nonumber \eeq 

From this, we can define an $\rho_n(\mathfrak{g})$-valued distribution $\phi^{\alpha,n}(\vec{x})$, and write \beq \phi^{\alpha,n}(f)1 = \int_{\vec{x} \in \bR^4} d\vec{x}\ f(\vec{x})\phi^{\alpha,n}(\vec{x})1, \nonumber \eeq when it acts on vacuum state. See \cite{streater, Ruelle}.

Using the transformation law for the field operator, we see that \beq U(\vec{a}, 1)\phi^{\alpha, n}(\vec{x})U(\vec{a}, 1)^{-1} = \phi^{\alpha,n}(\vec{x} + \vec{a}). \nonumber \eeq

\begin{rem}
We present an alternative way to understand quantum fields. One will be tempted to view $\phi^{\alpha,n}(\vec{x})1$ heuristically as \beq
\frac{1}{2\pi}e^{i[x^0\hat{H}(\rho_n) - x^1\hat{P}(\rho_n)]}\delta(\cdot - (x^2, x^3))\otimes \rho_n(F^\alpha), \nonumber \eeq whereby $\delta$ is the 2-dimensional Dirac delta function. But this is not entirely correct.

Suppose we replace the Dirac delta function with a Gaussian function $p_{\kappa}^{\vec{x}}(\cdot) = (\kappa^3/2\sqrt{2\pi})\exp[-\kappa^2|\vec{x}-\cdot|^2/4]/(2\pi)$. When we take the inner product,
\beq \int_{\vec{z} \in \bR^4} p_\kappa^{\vec{x}}(\vec{z})p_\kappa^{\vec{y}}(\vec{z})\ d\vec{z} = \frac{\kappa^2}{4(2\pi)}\exp[-\kappa^2|\vec{x}-\vec{y}|^2/8]. \nonumber \eeq

Write $x^+ = (0,0,x^2, x^3)$, $\hat{H} \equiv \hat{H}(\rho_n)$ and $\hat{P} \equiv \hat{P}(\rho_n)$. Thus, we can approximate $\langle \phi^{\alpha,n}(\vec{x})1, \phi^{\beta,n}(\vec{y})1 \rangle$ with \beq \left\langle \frac{1}{2\pi}e^{i[x^0\hat{H} - x^1\hat{P}]}p_{\kappa}^{x^+}\otimes \rho_n(F^\alpha), \frac{1}{2\pi}e^{i[y^0\hat{H} - y^1\hat{P}]}p_{\kappa}^{y^+}\otimes \rho_n(F^\beta) \right\rangle, \nonumber \eeq as
\begin{align*}
\frac{1}{(2\pi)^2}&e^{i[\hat{H}(x^0 - y^0) - \hat{P}(x^1 - y^1)]}\int_{\vec{z} \in \bR^4} p_\kappa^{x^+}(\vec{z})p_\kappa^{y^+}(\vec{z})\ d\vec{z} \cdot \left\langle \rho_n(F^\alpha), \rho_n(F^\beta) \right\rangle \\
=& \frac{\kappa^2}{4(2\pi)^3}e^{i[\hat{H}(x^0 - y^0) - \hat{P}(x^1 - y^1)]}\exp[-\kappa^2|x^+-y^+|^2/8]\cdot \left\langle \rho_n(F^\alpha), \rho_n(F^\beta) \right\rangle.
\end{align*}
Taking the limit, we see that we can understand the distribution $\left\langle \phi^{\alpha,n}(\vec{x})1, \phi^{\beta,n}(\vec{y} )1 \right\rangle$ as \beq
\frac{1}{(2\pi)^2}e^{i[\hat{H}(x^0 - y^0) - \hat{P}(x^1 - y^1)]}\cdot\delta(x^+ - y^+)\cdot C(\rho_n)\Tr[-F^\alpha F^\beta],\ x^+, y^+ \in S_0.\nonumber \eeq

Therefore, we can approximate $\langle \phi^{\alpha,n}(f)1, \phi^{\beta,n}(g)1 \rangle$ as
\begin{align*}
\frac{1}{(2\pi)^2}&\int_{\vec{x},\vec{y}\in\bR^4} f(\vec{x})\overline{g(\vec{y})}\left\langle e^{i(x^0\hat{H} - x^1\hat{P})}p_{\kappa}^{x^+}\otimes \rho_n(F^\alpha), e^{i(y^0\hat{H} - y^1\hat{P})}p_{\kappa}^{y^+}\otimes \rho_n(F^\beta) \right\rangle d\vec{x}d\vec{y} \\
=& \frac{\kappa^2}{4}\int_{\hat{s}, \hat{t} \in S_0} f^{\{e_0, e_1\}}(\hat{s})\overline{g^{\{e_0, e_1\}}}(\hat{t})\frac{1}{2\pi}\exp[-\kappa^2|\hat{s}-\hat{t}|^2/8]d\hat{s}d\hat{t}\cdot \left\langle \rho_n(F^\alpha), \rho_n(F^\beta) \right\rangle ,
\end{align*}
which approaches to \beq \int_{\hat{s}\in S_0} [f^{\{e_0, e_1\}}\overline{g^{\{e_0, e_1\}}}](\hat{s})d\hat{s}\cdot \left\langle \rho_n(F^\alpha), \rho_n(F^\beta) \right\rangle, \nonumber \eeq when $\kappa \rightarrow \infty$. This justifies the inner product given in Definition \ref{d.inn.2}. These calculations can also be found in \cite{YM-Lim01}.

Observe that \beq \lim_{\kappa \rightarrow \infty}\int_{\vec{z} \in \bR^4}p_{\kappa}^{\vec{x}}(\vec{z}) f(\vec{z})\ d\vec{z} = 0, \nonumber\eeq for any $f \in \mathscr{P}$. Thus, $\phi^{\alpha,n}(\vec{x})1$ itself has no meaning. But, we can give meaning to $\int_{\vec{x} \in \bR^4} d\vec{x}\ f(\vec{x})\phi^{\alpha,n}(\vec{x})1$.
\end{rem}

\subsection{Vacuum Expectation}\label{ss.ve}

\begin{notation}\label{n.lc.1}
Let $\{f_1, \cdots, f_r\}$ and $\{g_1, \cdots, g_s\}$ be two sets of compactly supported functions in $\mathscr{P}$. Let $\psi^{\alpha_\tau,n}(f_\tau)$ be either the creation operator $\phi^{\alpha_\tau,n}(f_\tau)$ or annihilation operator $\phi^{\alpha_\tau,n}(f_\tau)^\ast$. Similar notation for $\psi^{\beta_\theta,n}(g_\theta)$. Without loss of generality, we assume that $r \geq s$.

Fix a $n \in \mathbb{N}$. Throughout this subsection \ref{ss.ve}, we will write $\hat{H}(\rho_n) = \hat{H}$, $\hat{P}(\rho_n) = \hat{P}$. Define
\begin{align*}
A_r^n =& \psi^{\alpha_1,n}(f_1) \cdots \psi^{\alpha_r,n}(f_r),\quad  B_s^n = \psi^{\beta_1,n}(g_1) \cdots \psi^{\beta_s,n}(g_s).
\end{align*}
\end{notation}

Consider the vacuum expectation \beq \left\langle  A_r^nP_0B_s^n 1, 1\right\rangle \equiv \left\langle A_r^nB_s^n 1, 1 \right\rangle - \left\langle A_r^n1,1 \right\rangle \left\langle B_s^n1, 1 \right\rangle. \nonumber \eeq Here, $P_0$ is the orthogonal projection onto $\bigoplus_{n \geq 1}\mathscr{H}(\rho_n)$.

To write down an explicit formula, we need the following notation.

\begin{notation}\label{n.q.1}
Define $\left\{h_\theta \in \mathscr{P} \right\}_{\theta=1}^{r+s}$ as follows. When $1\leq \theta \leq r$, then
\beq h_\theta =
\left\{
  \begin{array}{ll}
    f_\theta, & \hbox{{\rm if} $\psi^{\alpha_\theta, n}(f_\theta) = \phi^{\alpha_\theta, n}(f_\theta) $;} \\
    -\overline{f}_\theta, & \hbox{{\rm if} $\psi^{\alpha_\theta, n}(f_\theta) = \phi^{\alpha_\theta, n}(f_\theta)^\ast$.}
  \end{array}
\right. \nonumber \eeq

When $r+1 \leq \theta \leq r+s$, then
\beq h_\theta =
\left\{
  \begin{array}{ll}
    g_{\theta-r}, & \hbox{{\rm if} $\psi^{\beta_{\theta-r}, n}(g_{\theta-r}) = \phi^{\beta_{\theta-r}, n}(g_{\theta-r}) $;} \\
    -\overline{g}_{\theta-r}, & \hbox{{\rm if} $\psi^{\beta_{\theta-r}, n}(g_{\theta-r}) = \phi^{\beta_{\theta-r}, n}(g_{\theta-r})^\ast$.}
  \end{array}
\right. \nonumber \eeq

Similarly, define $\left\{\tilde{h}_\theta \in \mathscr{P}\right\}_{\theta=1}^{r+s}$ as follows. When $1 \leq \theta \leq s$, then
\beq \tilde{h}_\theta =
\left\{
  \begin{array}{ll}
    g_{\theta}, & \hbox{{\rm if} $\psi^{\alpha_\theta, n}(g_{\theta}) = \phi^{\alpha_\theta, n}(g_{\theta}) $;} \\
    -\overline{g}_{\theta}, & \hbox{{\rm if} $\psi^{\alpha_\theta, n}(g_{\theta}) = \phi^{\alpha_\theta, n}(g_{\theta})^\ast$.}
  \end{array}
\right. \nonumber \eeq

When $s+1 \leq \theta \leq r$, then
\beq \tilde{h}_\theta =
\left\{
  \begin{array}{ll}
    f_{\theta-s}, & \hbox{{\rm if} $\psi^{\alpha_\theta, n}(f_{\theta-s}) = \phi^{\alpha_\theta, n}(f_{\theta-s}) $;} \\
    -\overline{f}_{\theta-s}, & \hbox{{\rm if} $\psi^{\alpha_\theta, n}(f_{\theta-s}) = \phi^{\alpha_\theta, n}(f_{\theta-s})^\ast$.}
  \end{array}
\right. \nonumber \eeq

When $r+1 \leq \theta \leq r+s$, then
\beq \tilde{h}_\theta =
\left\{
  \begin{array}{ll}
    f_{\theta-s}, & \hbox{{\rm if} $\psi^{\beta_{\theta-r}, n}(f_{\theta-s}) = \phi^{\beta_{\theta-r}, n}(f_{\theta-s}) $;} \\
    -\overline{f}_{\theta-s}, & \hbox{{\rm if} $\psi^{\beta_{\theta-r}, n}(f_{\theta-s}) = \phi^{\beta_{\theta-r}, n}(f_{\theta-s})^\ast$.}
  \end{array}
\right. \nonumber \eeq

For those $\theta \in \{1, 2, \cdots, r+s\}$, such that
\begin{align*}
\psi^{\alpha_\theta, n}(h) &= \phi^{\alpha_\theta, n}(h)^\ast,\ 1\leq \theta \leq r, \\
\psi^{\beta_{\theta - r}, n}(h) &= \phi^{\beta_{\theta-r}, n}(h)^\ast,\ r+1 \leq \theta \leq r+s,
\end{align*}
we will say that $\theta$ is adjoint. Then, we will write \beq \chi(\theta) =
\left\{
  \begin{array}{ll}
    -1, & \hbox{$\theta$ is adjoint;} \\
    1, & \hbox{otherwise.}
  \end{array}
\right.\nonumber \eeq

For $\vec{x} = (x^0, x^1, x^2, x^3)$, we will write $d\vec{x} = dx^0dx^1 dx^2dx^3$. And we will write $\vec{x} = (x^-, x^+)$, $x^- = (x^0, x^1)$, $x^+ = (x^2, x^3)$, $dx^- = dx^0dx^1$, $dx^+ = dx^2 dx^3$.

A partition of $R = \{1,2, \cdots, r + s\}$ is given by $Q = \{A_1, \ldots, A_{n(Q)}\}$, whereby
\begin{itemize}
  \item $R = \bigcup_{l=1}^{n(Q)}A_l$;
  \item $A_l \cap A_{\hat{l}} = \emptyset$, if $l \neq \hat{l}$;
  \item $A_l = \{z_l, z_l+1, \cdots, z_l + k_l\}$, for some $z_l \in R$ and some $0\leq k_l \leq r+s$.
\end{itemize}

Let $\Gamma$ be the set of all such possible partitions of $R$. For a partition $Q= \{A_1, \ldots, A_{n(Q)}\} \in \Gamma$, we will write
\begin{align*}
\int_Q \{h_{\theta}\}_{\theta=1}^{r+s} &:= \prod_{l=1}^{n(Q)}\left\{ \int_{S_0} \left[\prod_{\theta \in A_l} \int_{y_\theta^- \in \bR^2}\frac{e^{i\chi(\theta)[y_\theta^0 \hat{H} - y_\theta^1\hat{P}]}}{2\pi}h_\theta(y_\theta^-, y^+) dy_{\theta}^- \right]dy^+ \right\}\\
&=\prod_{l=1}^{n(Q)}\left\{\int_{S_0} \left[\prod_{\theta \in A_l} \int_{y_\theta^0, y_\theta^1 \in \bR}\frac{e^{i\chi(\theta)[y_\theta^0 \hat{H} - y_\theta^1\hat{P}]}}{2\pi}h_\theta(y_\theta^0, y_\theta^1, y^2, y^3) dy_{\theta}^0 dy_\theta^1 \right]dy^2 dy^3 \right\}.
\end{align*}
Note that it is a product of $n(Q)$ integrals and $A_l = \{z, z+1, \cdots, z+k_l\}$ for some $k_l$.

By abuse of notation, we will write $R = \{1, 2, \cdots, r + s\} \in \Gamma$, which is itself a partition of $R$.
\end{notation}

From Definitions \ref{d.co.1}, \ref{d.co.2}, \ref{d.ado}, we see that we can write
\begin{align}
\left\langle  A_r^nP_0B_s^n 1, 1\right\rangle
&=\sum_{Q \in \Gamma}c_Q\int_Q \{h_{\theta}\}_{\theta=1}^{r+s}, \label{e.q.1}
\end{align}
for some set of complex coefficients $\{c_Q \in \bC \}_{Q \in \Gamma}$.

\begin{rem}\label{r.s.1}
It is not difficult to see from Definitions \ref{d.co.2} and \ref{d.ado}, that if $Q$ contains a subset $\{z\}$, $z \in R$, then we must have $c_Q = 0$.
\end{rem}

\begin{defn}\label{d.co.3}
Refer to Notation \ref{n.lc.1}. Recall we assume $r \geq s$. Define
\begin{align*}
C_r^n &= \psi^{\alpha_1,n}(g_1) \cdots \psi^{\alpha_s,n}(g_s)\psi^{\alpha_{s+1},n}(f_1) \cdots \psi^{\alpha_{r},n}(f_{r-s}), \\
D_s^n &= \psi^{\beta_{1},n}(f_{r-s+1}) \cdots \psi^{\beta_{s},n}(f_r).
\end{align*}

Refer to Notation \ref{n.q.1}.
Define the following tempered distribution
\begin{align*}
\mathscr{W}^n &: f_1 \otimes_\bR \cdots \otimes_\bR f_r \otimes_\bR g_1 \otimes_\bR \cdots \otimes_\bR g_s \longmapsto \left\langle  A_r^nP_0B_s^n 1, 1\right\rangle.
\end{align*}

Here,
\begin{align*}
\mathscr{W}^n\left(\{\vec{x}_\tau\}_{\tau=1}^r, \{\vec{x}_\theta\}_{\theta=r+1}^{r+s}\right)
&\equiv \mathscr{W}^n\left(\{\vec{x}_\tau\}_{\tau=1}^s, \{\vec{x}_\theta\}_{\theta=s+1}^{s+r}\right) \\
&:= \mathscr{W}^n\left(\vec{x}_1, \cdots, \vec{x}_{r}, \vec{x}_{r+1}, \cdots, \vec{x}_{r+s} \right),
\end{align*}
is a tempered distribution such that ($\vec{x} = (x^0, x) = (x^0, x^1, x^2, x^3) \in \bR^4$)
\begin{align}
& \left\langle  A_r^nB_s^n 1, 1\right\rangle - \left\langle A_r^n1,1 \right\rangle \left\langle B_s^n1, 1 \right\rangle \equiv \left\langle A_r^nP_0 B_s^n1, 1\right\rangle \nonumber \\
& = \int_{\bR^4 \times \cdots \times \bR^4} \mathscr{W}^n\left(\{\vec{x}_\tau\}_{\tau=1}^r, \{\vec{x}_\theta\}_{\theta=r+1}^{r+s} \right)\bigotimes_{\tau=1}^{r+s} p_\tau(\vec{x}_\tau)\cdot \prod_{\tau=1 }^{r+s}d\vec{x}_\tau \nonumber \\
&= c_R\int_{S_0} \left[\prod_{\tau=1}^r \int_{x_\tau^- \in \bR^2}E(x_\tau^-)h_\tau(x_\tau^-, x^+ )dx_\tau^- \cdot \prod_{\theta=r+1}^{r+s} \int_{x_\theta^- \in \bR^2}E(x_\theta^-)h_\theta(x_\theta^-, x^+)\ dx_\theta^- \right]dx^+ \nonumber \\
& + \sum_{\myfrac{Q \neq R}{Q \in \Gamma}}c_Q\int_Q \{h_{\theta}\}_{\theta=1}^{r+s}, \label{e.c.3}
\end{align}
whereby $E(x_\tau^-) = E(x_\tau^0, x_\tau^1) = \dfrac{e^{i\chi(\tau)[x_\tau^0 \hat{H} - x_\tau^1\hat{P}]}}{2\pi}$.
Note that $\chi(\tau)$ was defined in Notation \ref{n.q.1} and $p_\tau=f_\tau$ if $1 \leq \tau \leq r$; $p_\tau = g_{\tau-r}$ if $r+1\leq \tau \leq r+s$.

If we switch the sets such that \beq \mathscr{W}^n : g_1 \otimes_\bR \cdots \otimes_\bR g_s \otimes_\bR f_1 \otimes_\bR \cdots \otimes_\bR f_r \longmapsto \left\langle  C_r^nP_0D_s^n 1, 1\right\rangle , \nonumber
\eeq
then
\begin{align}
&\left\langle  C_r^nD_s^n 1, 1\right\rangle - \left\langle C_r^n1,1 \right\rangle \left\langle D_s^n1, 1 \right\rangle \equiv \left\langle C_r^nP_0 D_s^n1, 1\right\rangle \nonumber \\
&= \int_{\bR^4 \times \cdots \times \bR^4}  \mathscr{W}^n\left(\{\vec{x}_\tau\}_{\tau=1}^s, \{\vec{x}_\theta\}_{\theta=s+1}^{s+r} \right)\bigotimes_{\tau=1}^{s+r} \tilde{p}_\tau(\vec{x}_\tau)\cdot \prod_{\tau=1 }^{s+r}d\vec{x}_\tau \nonumber \\
&= c_R\int_{S_0} \left[\prod_{\tau=1}^s \int_{x_\tau^- \in \bR^2}E(x_\tau^-)\tilde{h}_\tau(x_\tau^-, x^+)dx_\tau^- \cdot \prod_{\theta=s+1}^{s+r} \int_{x_\theta^- \in \bR^2}E(x_\theta^-)\tilde{h}_\theta(x_\theta^-, x^+)\ dx_\theta^- \right]dx^+  \nonumber \\
&+ \sum_{\myfrac{Q \neq R}{Q \in \Gamma}}c_Q\int_Q \{\tilde{h}_{\theta}\}_{\theta=1}^{s+r}. \label{e.c.2}
\end{align}
Note that $\tilde{p}_\tau=g_\tau$ if $1 \leq \tau \leq s$; $\tilde{p}_\tau = f_{\tau-s}$ if $s+1\leq \tau \leq s+r$.
\end{defn}

\begin{rem}
Refer to Remark \ref{r.c.1}.
\end{rem}

Write $B_s^{n,\vec{a}}1$ as
\begin{align*}
\psi^{\beta_1,n}(g_1(\cdot - \vec{a}))\cdots \psi^{\beta_{s-1},n}(g_{s-1}(\cdot - \vec{a}))\Big(S_0 + \vec{a}, g_s^{\{e_0, e_1\}}(\cdot - \vec{a}) \otimes \rho_n(F^{\beta_s}) , \{e_a\}_{a=0}^3\Big),
\end{align*}
and $\psi^{\beta_\tau,n}(g_\tau)_{U(\vec{a})} := U(\vec{a},1)\psi^{\beta_\tau,n}(g_\tau)U(\vec{a},1)^{-1}$, for $1 \leq \tau \leq s$.

By definition of the field operator and its adjoint, we have from Proposition \ref{p.w.1},
\begin{align}
P_0U(\vec{a},1)B_s^n 1
=& P_0\psi^{\beta_1,n}(g_1)_{U(\vec{a})} \psi^{\beta_2,n}(g_2)_{U(\vec{a})} \cdots \psi^{\beta_{s-1},n}(g_{s-1})_{U(\vec{a})}U(\vec{a},1)\psi^{\beta_s,n}(g_s)1 \nonumber \\
=& e^{i[a^0 \hat{H} - a^1\hat{P}]}P_0B_s^{n,\vec{a}}1
= e^{i[a^0 \hat{H} - a^1\hat{P}]}\left[B_s^{n,\vec{a}}1 - \left\langle B_s^n 1, 1 \right\rangle 1 \right], \nonumber
\end{align}
for $\vec{a} = \sum_{b=0}^3 a^b e_b$.

Hence, we have
\beq \left\langle  A_r^nP_0U(\vec{a},1)B_s^n 1, 1\right\rangle = e^{i[a^0 \hat{H} - a^1\hat{P}]}\left[\left\langle A_r^nB_s^{n,\vec{a}}1, 1 \right\rangle - \left\langle A_r^n1,1 \right\rangle \left\langle B_s^n1, 1 \right\rangle \right]. \label{e.c.8} \eeq

When $\vec{a} \notin S_0$, by our construction, we see that $\left\langle  A_r^nP_0 U(\vec{a}, 1)B_s^n 1, 1\right\rangle$ is zero, because $S_0 \cap (S_0 +\vec{a}) = \emptyset$.

By definition, we have
\begin{align}
e^{i[a^0 \hat{H} - a^1\hat{P}]}&\int_{\bR^2} \frac{e^{i[s\hat{H} - t\hat{P}]}}{2\pi}f\left(s, t, x^2, x^3   \right)\ dsdt \nonumber \\
=& \int_{\bR^2} \frac{e^{i[(s+a^0)\hat{H} - (t+a^1)\hat{P}]}}{2\pi}f\left(s, t, x^2, x^3 \right)\ dsdt \nonumber \\
=& \int_{\bR^2} \frac{e^{i[s\hat{H} - t\hat{P}]}}{2\pi}f\left(s-a^0, t-a^1, x^2, x^3 \right)\ dsdt \nonumber \\
=& f\left(\cdot -(a^0, a^1, 0, 0)\right)^{\{e_0, e_1\}}(\hat{H}, \hat{P})(0,0, x^2, x^3). \label{e.a.6}
\end{align}
Hence, after a translation by $\vec{a} = (a^0, a^1, 0, 0)$, we can consider \newline
$e^{i[a^0 \hat{H} - a^1\hat{P}]}\left(S_0 + \vec{a}, f^{\{e_0, e_1\}}(\cdot - \vec{a}) \otimes \rho(F^\alpha), \{e_0, e_1\}\right)$, as \newline $\left(S_0, f(\cdot - \vec{a})^{\{e_0, e_1\}}\otimes \rho(F^\alpha), \{e_a\}_{a=0}^3\right)$. Thus, we will define the following.

\begin{notation}\label{n.co.2}
Let $\vec{a} = (a^0, a) \in \bR^4$. Write \beq \mathscr{W}^n\left(\{\vec{x}_\tau\}_{\tau=1}^r, \{\vec{x}_\theta + \vec{a}\}_{\theta=r+1}^{r+s} \right) := \mathscr{W}^n(\vec{x}_1, \cdots, \vec{x}_{r}, \vec{x}_{r+1}+\vec{a}, \cdots, \vec{x}_{r+s}+ \vec{a}). \nonumber \eeq

And write $h_\theta^{\vec{a}}(\vec{y}) = h_\theta(y^0, y)$ if $1 \leq \theta \leq r$; $h_\theta^{\vec{a}}(\vec{y}) = h_\theta(y^0 - a^0, y - a)$ if $r+1 \leq \theta \leq r+s$.
\end{notation}

\begin{defn}\label{d.ex.1}
Refer to Definition \ref{d.ma.1}, Notations \ref{n.q.1} and \ref{n.co.2}.
Write $\vec{a} = (a^0, a) = (a^0, a^1, a^2, a^3) \in \bR^4$.

Using Equations (\ref{e.q.1}), (\ref{e.c.8}) and (\ref{e.a.6}), we will define
\begin{align}
&H^n(\vec{a}) \nonumber \\
&:=\int_{\bR^4 \times \cdots \times \bR^4}  \mathscr{W}^n\left(\{\vec{x}_\tau\}_{\tau=1}^r, \{\vec{x}_\theta\}_{\theta=r+1}^{r+s} \right)\bigotimes_{\tau=1}^r p_\tau(\vec{x}_\tau)\cdot \bigotimes_{\theta=r+1}^{r+s} p_\theta(\vec{x}_{r+\theta} - \vec{a}) \cdot \prod_{\tau=1}^{r+s}d\vec{x}_\tau  \nonumber \\
&\equiv \int_{\bR^4 \times \cdots \times \bR^4}  \mathscr{W}^n\left(\{\vec{x}_\tau\}_{\tau=1}^r, \{\vec{x}_\theta + \vec{a}\}_{\theta=r+1}^{r+s}  \right)\bigotimes_{\tau=1}^{r+s} p_\tau(\vec{x}_\tau)\cdot \prod_{\tau=1}^{r+s}d\vec{x}_\tau  \nonumber \\
&= c_R\int_{S_0}\left[ \prod_{\tau=1}^r \int_{x_\tau^- \in \bR^2}E(x_\tau^-) h_\tau(x_\tau^-, x^+) dx_\tau^- \cdot \prod_{\theta=r^+}^{r+s} \int_{x_\theta^- \in \bR^2}E(x_\theta^-)
h_\theta^{\vec{a}}(x_\theta^-,x^+)dx_\theta^- \right]\ dx^+ \nonumber \\
& + \sum_{\myfrac{Q \neq R}{Q \in \Gamma}}c_Q\int_Q \{h_{\theta}^{\vec{a}}\}_{\theta=1}^{r+s}, \label{e.c.9}
\end{align}
where $E(x_\tau^-) = E(x_\tau^0, x_\tau^1) = \dfrac{e^{i\chi(\tau)[x_\tau^0\hat{H} - x_\tau^1\hat{P}]}}{2\pi}$ and $p_\tau=f_\tau$ if $1 \leq \tau \leq r$; $p_\tau = g_{\tau-r}$ if $r+1\leq \tau \leq r+s$.
\end{defn}

\begin{rem}\label{r.i.1}
We will use $H^n(\vec{a})$, instead of $\left\langle  A_r^nP_0 U(\vec{a}, 1)B_s^n 1, 1\right\rangle$, to prove both the Cluster Decomposition Property and Clustering Theorem. When $\vec{a} \in S_0$, then they are equal to each other.
\end{rem}

\begin{lem}\label{l.ex.1}
Refer to Notation \ref{n.lc.1}, whereby we have the two sets $\{\psi^{\alpha_\tau,n}\}_{\tau = 1}^r$ and $\{\psi^{\beta_\theta,n}\}_{\theta = 1}^s$. Suppose
the sum $\sum_{\tau=1}^r \chi(\tau) + \sum_{\theta=r+1}^{r+s} \chi(\theta) = 0$, $\chi$ was defined in Notation \ref{n.q.1}.

Define for $1 \leq i \leq r+s-1$, $\vec{\xi}_i = \vec{x}_{i} - \vec{x}_{i+1}$, which are relative coordinates. Let $\vec{a} = (a^0, a^1, a^2, a^3)$, pertaining to the standard orthonormal basis.

\begin{enumerate}
  \item For any $1\leq \alpha \leq r+s$, $\mathscr{W}^n$ is independent of the variables $\{ x_\alpha^2, x_\alpha^3\}$.
  \item\label{i.ex.1b} There exists a tempered distribution $W^n$, such that \beq \mathscr{W}^n\left(\{\vec{x}_\tau\}_{\tau=1}^r, \{\vec{x}_\theta+ \vec{a}\}_{\theta=r+1}^{r+s}  \right) = W^n\left(\vec{\xi}_1, \cdots, \vec{\xi}_{r-1},  \vec{\xi}_{r}- \vec{a}, \vec{\xi}_{r+1}, \cdots, \vec{\xi}_{r+s-1}\right). \nonumber \eeq
  \item Their Fourier Transforms are related by
\begin{align}
\widehat{\mathscr{W}}^n&\left(\vec{p}_1, \cdots, \vec{p}_{r+s}\right) \nonumber \\
&= (2\pi)^4 \delta\left(\sum_{\tau=1}^{r+s} \vec{p}_{\tau} \right) \widehat{W}^n\left(\vec{p}_1, \vec{p}_1 + \vec{p}_2, \cdots, \vec{p}_1 + \vec{p}_2 + \cdots + \vec{p}_{r+s-1}\right). \label{e.c0.4}
\end{align}
\end{enumerate}
\end{lem}

\begin{proof}
Recall we have a set of Schwartz functions $\{f_1, \cdots, f_r, g_1, \cdots g_{s}\} \subset \mathscr{P}$. Since $\sum_{\tau=1}^r \chi(\tau) + \sum_{\theta=r+1}^{r+s} \chi(\theta) = 0$, the set of adjoint integers $\{\delta_1, \cdots, \delta_c\}$ is in one to one correspondence with the set of non-adjoint integers $\{\epsilon_1, \cdots, \epsilon_c\}$. We will pair each $\delta_i$ with $\epsilon_i$.  Note that for each $i=1, \cdots, c$, $c = \frac{r+s}{2}$, we can write \beq \vec{x}_{\delta_i} - \vec{x}_{\epsilon_i} = \vec{\xi}_{\delta_i} + \vec{\xi}_{\delta_i+1} + \cdots + \vec{\xi}_{\epsilon_i - 1}, \label{e.l.3} \eeq when $\delta_i < \epsilon_i$. When $\delta_i > \epsilon_i$, we have \beq \vec{x}_{\delta_i} - \vec{x}_{\epsilon_i} = -\left( \vec{\xi}_{\epsilon_i} + \vec{\xi}_{\epsilon_i+1} + \cdots + \vec{\xi}_{\delta_i - 1}\right). \label{e.l.4} \eeq

Hence, we can write \beq \sum_{\tau=1}^r \chi(\tau)[x_\tau^0 \hat{H} - x_\tau^1\hat{P}] +
\sum_{\theta=r+1}^{r+s} \chi(\theta)[x_\theta^0 \hat{H} - x_\theta^1\hat{P}] = -\sum_{i=1}^{r+s-1}c_i[\xi_i^0 \hat{H} - \xi_i^1\hat{P}],\label{e.l.1} \eeq for some integers $c_i$'s. Therefore,
\beq \prod_{\tau=1}^r E(x_\tau^-)\prod_{\theta=r+1}^{r+s} E(x_\theta^-) = \frac{1}{(2\pi)^{r+s}}\exp\left[ -i \sum_{j=1}^{r+s-1}c_j[\xi_j^0 \hat{H} - \xi_j^1\hat{P}]\right]. \label{e.d.1} \eeq

Thus, Remark \ref{r.s.1} and Equation (\ref{e.c.9}) imply that
\begin{enumerate}
  \item $\mathscr{W}^n$ is independent of any chosen variables $\{x_\alpha^2, x_\alpha^3\}$, by letting $x^+ = (x_\alpha^2, x_\alpha^3)$ in the equation;
  \item there is a tempered distribution $W^n$ such that
\beq W^n\left(\vec{\xi}_1, \cdots, \vec{\xi}_{r+s-1} \right) := \mathscr{W}^n \left(\{\vec{x}_\tau\}_{\tau=1}^r, \{\vec{x}_\theta\}_{\theta=r+1}^{r+s} \right). \nonumber \eeq
\end{enumerate}

We also have \beq \mathscr{W}^n\left(\{\vec{x}_\tau\}_{\tau=1}^r, \{\vec{x}_\theta+ \vec{a}\}_{\theta=r+1}^{r+s}  \right) = W^n\left(\vec{\xi}_1, \cdots, \vec{\xi}_{r-1},  \vec{\xi}_{r}- \vec{a}, \vec{\xi}_{r+1}, \cdots, \vec{\xi}_{r+s-1}\right), \nonumber \eeq by definition of $\vec{\xi}_j$'s.

The proof for Equation (\ref{e.c0.4}) can be found in \cite{streater}, hence omitted.
\end{proof}

\begin{lem}\label{l.b.1}
Refer to Definition \ref{d.ma.1}. Assume the sum $\sum_{\tau=1}^r \chi(\tau) + \sum_{\theta=r+1}^{r+s} \chi(\theta) = 0$. Recall in the proof of Lemma \ref{l.ex.1}, we defined the integers $c_i$'s in Equation (\ref{e.l.1}). For $\vec{a} \in \bR^4$, write $a^- = (a^0, a^1,0,0)$, $a^+ = (0,0, a^2, a^3)$, $da^- = da^0 da^1$, and $da^+ = da^2 da^3$. For any $1 \leq t \leq r + s-1$, we have \beq \int_{\bR^2}e^{-ia^- \cdot q^- }W^n\left(\vec{\xi}_1, \cdots, \vec{\xi}_{t}-  a^-, \vec{\xi}_{t+1}, \cdots, \vec{\xi}_{r+s-1}\right)\ da^- \nonumber \eeq vanishes if $q^- \neq -c_t(\hat{H}(\rho_n), \hat{P}(\rho_n), 0, 0 )$.
\end{lem}

\begin{proof}
Let $\vec{\alpha} = (\hat{H}(\rho_n), \hat{P}(\rho_n),0,0)$, and $t^+ \equiv t+1$. Using Equations (\ref{e.c.9}) and (\ref{e.d.1}), we see that \beq W^n\left(\vec{\xi}_1, \cdots, \vec{\xi}_{t}-  a^-, \vec{\xi}_{t+1}, \cdots, \vec{\xi}_{r+s-1}\right) = e^{-ic_ta^-\cdot \vec{\alpha}}W^n\left(\vec{\xi}_1, \cdots, \vec{\xi}_{t}, \vec{\xi}_{t+1}, \cdots, \vec{\xi}_{r+s-1}\right). \label{e.q.2} \eeq

If we take the Fourier Transform,
\begin{align*}
\int_{\bR^2}&e^{-i q^-\cdot a^- }W^n\left(\vec{\xi}_1, \cdots, \vec{\xi}_{t}-  a^-, \vec{\xi}_{t^+}, \cdots, \vec{\xi}_{r+s-1}\right)\ da^- \\
=&  \int_{\bR^2} e^{-i a^-\cdot q^- }e^{-i c_ta^-\cdot \vec{\alpha} }W^n\left(\vec{\xi}_1, \cdots, \vec{\xi}_{t}, \vec{\xi}_{t^+}, \cdots, \vec{\xi}_{r+s-1}\right)\ da^- \\
=& 2\pi W^n\left(\vec{\xi}_1, \cdots, \vec{\xi}_{t}, \cdots, \vec{\xi}_{r+s-1}\right) \cdot \delta\left(q^- + c_t(\hat{H}(\rho_n), \hat{P}(\rho_n), 0, 0 ) \right).
\end{align*}
This completes the proof.
\end{proof}

\begin{rem}\label{r.ex.1}
If one examines the proof, we note that the Fourier Transform of \beq a^- \longmapsto F^-(a^-) \equiv W^n\left(\vec{\xi}_1, \cdots, \vec{\xi}_{r}-  a^-, \vec{\xi}_{r+1}, \cdots, \vec{\xi}_{r+s-1}\right), \nonumber \eeq is simply a multiple of the Dirac delta function $\delta\left(\lambda_r^- +  c_r(\hat{H}(\rho_n), \hat{P}(\rho_n))\right)$. From Equation (\ref{e.c0.4}), $\lambda_r^- = \sum_{i=1}^r p_i^-$, $p_i^- \equiv (p_i^0, p_i^1)$, is the momenta variable conjugate to position variable $x_i^- \equiv (x_i^0, x_i^1)$.

A similar calculation will show that the Fourier Transform of
\begin{align*}
a^- \longmapsto F^+(a^-) &= \mathscr{W}^n\left(\{\vec{x}_\theta+ a^-\}_{\theta=1}^{s} , \{\vec{x}_\tau\}_{\tau=s+1}^{s+r}  \right) \\
&\equiv W^n\left(\vec{\xi}_1, \cdots, \vec{\xi}_{s}+ a^-, \vec{\xi}_{s+1}, \cdots, \vec{\xi}_{r+s-1}\right), \end{align*}
is simply a multiple of the Dirac delta function $\delta\left(\lambda_s^- -  c_s(\hat{H}(\rho_n), \hat{P}(\rho_n))\right)$. In this case, the support is on the set $\left\{\lambda_s^- =  c_s[\hat{H}(\rho_n)e_0+ \hat{P}(\rho_n)e_1]\right\}$, $c_s \in \mathbb{Z}$.
\end{rem}

\begin{defn}\label{d.reg.1}
Refer to Notation \ref{n.lc.1}. We say that the set $\{\psi^{\alpha_\tau,n}\}_{\tau = 1}^r$ is regular if there exists a unique $0 \leq k\leq r$ such that for any $h \in \mathscr{P}$,
\begin{itemize}
  \item $\psi^{\alpha_\tau,n}(h) = \phi^{\alpha_\tau,n}(h)^\ast$, for every $1 \leq \tau \leq k$, and
  \item $\psi^{\alpha_\tau,n}(h) = \phi^{\alpha_\tau,n}(h)$, for every $k+1 \leq \tau \leq r$.
\end{itemize}
\end{defn}

\begin{prop}\label{p.p.3}
Refer to Notation \ref{n.lc.1}, whereby we have the two sets $\{\psi^{\alpha_\tau,n}\}_{\tau = 1}^r$ and $\{\psi^{\beta_\theta,n}\}_{\theta = 1}^s$, $r \geq s$. Suppose
\begin{itemize}
  \item both sets are regular as defined in Definition \ref{d.reg.1};
  \item the former has $0 \neq \bar{k}$ annihilation operators, and $\bar{l} := r - \bar{k}$ creation operators;
  \item the latter has $\underline{k}$ annihilation operators, and $0 \neq \underline{l} := s - \underline{k}$ creation operators;
  \item the sum $\sum_{\tau=1}^r \chi(\tau) + \sum_{\theta=r+1}^{r+s} \chi(\theta) = 0$, $\chi$ was defined in Notation \ref{n.q.1}. That is, we have equal number of annihilation and creation operators, i.e. \beq \bar{k} + \underline{k} = \bar{l} + \underline{l}. \label{e.b.2} \eeq
\end{itemize}

Suppose ${\rm min}\ \{\underline{k}, \bar{l}\} \neq 0$. If $\bar{k} > \bar{l}$, then the coefficients $c_i$'s in Equation (\ref{e.l.1}) can all be chosen to be positive. In particular, we have that $c_r > 0$.
If $\bar{k} < \bar{l}$, then the coefficients $c_i$'s, $2\bar{k}+1 \leq i \leq r$, are all negative. In particular, we have that $c_r < 0$.

When ${\rm min}\ \{\underline{k}, \bar{l}\} = 0$, the coefficients $c_i$'s in Equation (\ref{e.l.1}) can all be chosen to be positive. In particular, we have that $c_r > 0$.
\end{prop}

\begin{proof}
First consider ${\rm min}\ \{\underline{k}, \bar{l}\} \neq 0$. We have the sets $A_1 :=\{1, \cdots, \bar{k}\}$ and $A_2:= \{r+1, \cdots, r+\underline{k}\}$, both containing integers which are adjoint. And we have the sets $B_1 := \{ \bar{k}+1, \cdots, \bar{k} + \bar{l}\}$ and $B_2:= \{ r + \underline{k}+1, \cdots, r+\underline{k} + \underline{l} \}$, both containing integers which are not adjoint.

We will pair the numbers in $A_1 \cup A_2$, with numbers in $B_1 \cup B_2$, since the cardinality of both sets are the same.
Assume $\bar{k} \geq \bar{l}$. Define a bijective map $G: A_1 \cup A_2 \longrightarrow B_1 \cup B_2$ by
\begin{align*}
G(a) =
\left\{
  \begin{array}{ll}
    a+\bar{k}, & \hbox{$1\leq a \leq \bar{l}$;} \\
    a+\bar{k}+\underline{k} , & \hbox{$\bar{l}+1 \leq a \leq \bar{k}$;} \\
    a+\underline{k} + \bar{k} - \bar{l} \equiv a + \underline{l}, & \hbox{$r+1 \leq a \leq r+\underline{k}$.}
  \end{array}
\right.
\end{align*}
Using this map, we see that we can pair in such a way that each $\delta_i \in A_1 \cup A_2$ is paired with an unique $\epsilon_i \in B_1 \cup B_2$, with each $\delta_i < \epsilon_i$. Equation (\ref{e.l.3}) holds for all the $1 \leq \delta_i \leq \bar{k}$ and $r+1\leq \delta_i \leq r+ \underline{k}$. The map $G$ shows that all the $c_i$'s in Equation (\ref{e.l.1}), except at $i=r$, are all positive integers. But if we assume $\bar{k} > \bar{l}$, we will have $G: \bar{l}+1 \mapsto r+\underline{k} + 1$, thus we see that $c_r > 0$.

Assume $\bar{k} < \bar{l}$. Define a bijective map $\tilde{G}: A_1 \cup A_2 \longrightarrow B_1 \cup B_2$ by
\begin{align*}
\tilde{G}(a) =
\left\{
  \begin{array}{ll}
    a+\bar{k}, & \hbox{$1\leq a \leq \bar{k}$;} \\
    a+\bar{k}-\bar{l} , & \hbox{$r+1 \leq a \leq r+\bar{l} - \bar{k}$;} \\
    a+ \bar{k} - \bar{l}+\underline{k}  \equiv a + \underline{l}, & \hbox{$r+\bar{l} - \bar{k}+1 \leq a \leq r+\underline{k}$.}
  \end{array}
\right.
\end{align*}
In this case, we see that for $r+1 \leq \delta_i \leq r+\bar{l} - \underline{k}$, we have that $\delta_i > \tilde{G}(\delta_i) = \epsilon_i$. Equation (\ref{e.l.4}) holds instead, for such $\delta_i$ and corresponding $\epsilon_i$. Thus, the coefficients $c_i$'s, $2\bar{k}+1 \leq i \leq r$ in Equation (\ref{e.l.1}) are negative. In particular, we have that $c_r< 0$.

Now we consider ${\rm min}\ \{\underline{k}, \bar{l}\} = 0$. When $\bar{l} = 0$, this implies that $r = s$ and thus $r = \bar{k} = \underline{l} = s$. Define a bijective map $\hat{G}: \{1, \cdots, r\} \longrightarrow \{r+1, \cdots, 2r\}$ as $G(a) = a+r$. Then clearly, $a < \hat{G}(a)$ and since $\hat{G}: r \mapsto r+s = 2r$, all the coefficients in Equation (\ref{e.l.1}) are strictly positive. In particular, $c_r > 0$.

When $\underline{k} = 0$, then we must have $\bar{k} = \bar{l} + \underline{l}$. Define a bijective map $\dot{G}$ that sends $a$ to $a + \bar{k}$, for $1 \leq a \leq \bar{k}$. Again, $a < \dot{G}(a)$ and because $\dot{G}: \bar{k} \mapsto r+s$, all the coefficients in Equation (\ref{e.l.1}) are strictly positive. In particular, $c_r > 0$.
\end{proof}

\begin{rem}
When $\bar{k} = \bar{l}$, we see that $c_r = 0$ in the proof, which implies that the support in Lemma \ref{l.b.1} is at $\{\vec{q} = \vec{0}\}$. This will pose a problem later, when we try to prove the Clustering Theorem \ref{t.cl}. Hence, we do not consider this particular case.

When $\bar{k} = 0$ or $\underline{l} = 0$, then $\mathscr{W}^n \equiv 0$, which is trivial.
\end{rem}

\begin{rem}\label{r.p.1}
Assume the conditions in Proposition \ref{p.p.3} hold. We will further assume one of the following cases hold:
\begin{enumerate}
  \item ${\rm min}\ \{\underline{k}, \bar{l}\} = 0$; or
  \item ${\rm min}\ \{\underline{k}, \bar{l}\} \neq 0$ and $\bar{k} > \bar{l}$; or
  \item ${\rm min}\ \{\underline{k}, \bar{l}\} \neq 0$, $\bar{k} < \bar{l}$ and $2\bar{k}+1 \leq s$.
\end{enumerate}

When we assume one of the above 3 cases holds, Proposition \ref{p.p.3} says that $\sgn(c_s) = \sgn(c_r)$.
Then the supports described in Remark \ref{r.ex.1} are on respective positive and negative cones in energy-momentum space. The fact that the supports are separated is key to prove the Clustering Decomposition Property, as demonstrated in \cite{streater} and \cite{Araki1962}.
\end{rem}

In general, it is not true that \beq \mathscr{W}^n\left(\{\vec{x}_\tau\}_{\tau=1}^r, \{\vec{x}_\theta\}_{\theta=r+1}^{r+s} \right)  = \pm
\mathscr{W}^n\left( \{\vec{x}_\theta\}_{\theta=r+1}^{r+s} ,\{\vec{x}_\tau\}_{\tau=1}^r\right). \nonumber \eeq The commutation and anti-commutation relations in Lemmas \ref{l.w.3}, \ref{l.w.4} and \ref{l.w.5} hold only for real or purely imaginary functions respectively. Thus, we need to consider in a similar manner.

\begin{rem}
Proposition \ref{p.c.1} will imply that the above equality holds, provided $x_\theta^+ - x_\tau^+ \neq 0$ for any $1 \leq \tau \leq r$ and $r+1 \leq \theta \leq r+s$.
\end{rem}

\begin{notation}\label{n.q.2}
For a natural number $r$, let the set $X_r := \{1,2, \cdots, r\}$. Decompose $X_r$ into 2 sets $\kappa_r$, $\bar{\kappa}_r$, i.e. $\kappa_r \cup \bar{\kappa}_r = X_r$ and $\kappa_r \cap \bar{\kappa}_r = \emptyset$, one of it possibly empty. Let $\Omega_r$ denote the set containing all such decompositions of $X_r$.

Recall we have two sets $\{f_1, \cdots, f_r\}$ and $\{g_1, \cdots, g_s\}$. Assume that each $f_i$ ($g_j$) is either real or purely imaginary. This defines 2 disjoint sets $\kappa_r$ ($\pi_s$), and $\bar{\kappa}_r$ ($\bar{\pi}_s$), whereby $f_i$ ($g_j$) is real if $i \in \kappa_r$ ($j \in \pi_s$), otherwise $f_i$ ($g_j$) is purely imaginary if $i \in \bar{\kappa}_r$ ($j \in \bar{\pi}_s$).

Then, we will write for $(\kappa_r, \bar{\kappa}_r) \in \Omega_r$, $(\pi_s, \bar{\pi}_s) \in \Omega_s$, tempered distributions such that
\begin{align*}
\mathscr{W}_{\kappa_r, \pi_s}^n &: f_1 \otimes_\bR \cdots \otimes_\bR f_r \otimes_\bR g_1 \otimes_\bR \cdots \otimes_\bR g_s \longmapsto \left\langle  A_r^nP_0B_s^n 1, 1\right\rangle,\\
\mathscr{W}_{\pi_s, \kappa_r}^n &: g_1 \otimes_\bR \cdots \otimes_\bR g_s \otimes_\bR f_1 \otimes_\bR \cdots \otimes_\bR f_r  \longmapsto
\left\langle  C_r^nP_0D_s^n 1, 1\right\rangle ,
\end{align*}
for these particular sets $\{f_1, \cdots, f_r\}$ and $\{g_1, \cdots, g_s\}$. Refer back to Equations (\ref{e.c.3}) and (\ref{e.c.2}).
\end{notation}

\begin{rem}\label{r.c.1}
Recall from Notation \ref{n.lc.1}, we have the sets $\{f_1, \cdots, f_r\}$ and $\{g_1, \cdots, g_s\}$. Write for $\tau=1, \cdots, r$, $f_\tau = \underline{f}_\tau + i\overline{f}_\tau$, whereby $\underline{f}_\tau = {\rm Re}\ f_\tau$ and $\overline{f}_\tau = {\rm Im}\ f_\tau$. Similarly, write for $\theta=1, \cdots, s$, $g_\theta = \underline{g}_\theta + i\overline{g}_\theta$, whereby $\underline{g}_\theta = {\rm Re}\ g_\theta$ and $\overline{g}_\theta = {\rm Im}\ g_\theta$.

Refer to Notation \ref{n.q.2}. We can now understand
\begin{align}
&\int_{\bR^4 \times \cdots \times \bR^4} \mathscr{W}^n\left(\{\vec{x}_\tau\}_{\tau=1}^r, \{\vec{x}_\theta\}_{\theta=r+1}^{r+s} \right)\bigotimes_{\tau=1}^{r+s} p_\tau(\vec{x}_\tau)\cdot \prod_{\tau=1 }^{r+s}d\vec{x}_\tau \nonumber \\
&= \sum_{\myfrac{(\kappa_r, \bar{\kappa}_r) \in \Omega_r }{(\pi_s, \bar{\pi}_s) \in \Omega_s}}\int_{\bR^4 \times \cdots \times \bR^4} \mathscr{W}_{\kappa_r, \pi_s}^n\left(\{\vec{x}_\tau\}_{\tau=1}^r, \{\vec{x}_\theta\}_{\theta=r+1}^{r+s} \right)\prod_{\tau=1}^{r+s} q^{\{\kappa_r, \pi_s\}}_\tau(\vec{x}_\tau)\cdot \prod_{\tau=1 }^{r+s}d\vec{x}_\tau, \label{e.d.3}
\end{align}
whereby for each integral term in the above sum, indexed by $\{\kappa_r, \pi_s\}$,
\begin{itemize}
  \item for $1 \leq \tau \leq r$, $q^{\{\kappa_r, \pi_s\}}_\tau$ is $\underline{f}_\tau$ if $\tau \in \kappa_r$, otherwise $q^{\{\kappa_r, \pi_s\}}_\tau$ is $i\overline{f}_\tau \equiv \sqrt{-1}\overline{f}_\tau$;
  \item for $r+1 \leq \tau \leq r+s$, $q^{\{\kappa_r, \pi_s\}}_\tau$ is $\underline{g}_{\tau-r}$ if $\tau-r \in \pi_s$, otherwise $q^{\{\kappa_r, \pi_s\}}_\tau$ is $i\overline{g}_{\tau-r} \equiv \sqrt{-1}\overline{g}_{\tau-r}$.
\end{itemize}
\end{rem}

To apply the commutation or anti-commutation relations in Lemmas \ref{l.w.1}, \ref{l.w.3}, \ref{l.w.4} and \ref{l.w.5} to switch the variables, we need to assume regularity of the sets in Definition \ref{d.reg.1}. Making this assumption and using the above notations, we will have the following lemma. The proof involves applying successively, the commutation and anti-commutation relations in the lemmas, and the details will be left to the reader.

\begin{lem}\label{l.lc.1}
Assume the two sets $\{\psi^{\alpha_\tau,n}\}_{\tau = 1}^r$ and $\{\psi^{\beta_\theta,n}\}_{\theta = 1}^s$, $r \geq s$, both are regular. We have that \beq \mathscr{W}_{\kappa_r, \pi_s}^n\left(\{\vec{x}_\tau\}_{\tau=1}^r, \{\vec{x}_\theta\}_{\theta=r+1}^{r+s} \right)  = \pm
\mathscr{W}_{\pi_s, \kappa_r}^n\left( \{\vec{x}_\theta\}_{\theta=r+1}^{r+s} ,\{\vec{x}_\tau\}_{\tau=1}^r\right), \label{e.b.1b} \eeq provided for any $1 \leq \tau \leq r$, $r+1 \leq \theta \leq r+s$, we have $0 \neq \vec{x}_\tau - \vec{x}_\theta$ lies in the hyperplane spanned by $\{ \hat{P}(\rho_n)e_0 + \hat{H}(\rho_n)e_1, e_2, e_3\}$. The $\pm$ sign depends on how many anti-commutation relations in Lemmas \ref{l.w.3}, \ref{l.w.4} and \ref{l.w.5} are applied.
\end{lem}

\subsection{Proof of Cluster Decomposition Property}\label{ss.cd}

\begin{defn}\label{d.t.1}
Let $\varphi_1$, $\varphi_2$ be Schwartz functions on $\bR^{4r}$ and $\bR^{4s}$ respectively, $r \geq s$. Define
\begin{align*}
T_1^n
&:= \int_{\bR^4 \times \cdots \times \bR^4}\prod_{\tau=1}^{r}d\vec{x}_\tau \cdot \varphi_1\left( \{\vec{x}_\tau\}_{\tau=1}^{r}\right) \cdot \prod_{\tau=1}^{r}\psi^{\alpha_\tau,n}(\vec{x}_\tau), \\
T_2^n
&:= \int_{\bR^4 \times \cdots \times \bR^4}\prod_{\theta=r+1}^{r+s}d\vec{x}_\theta \cdot \varphi_2\left( \{\vec{x}_\theta\}_{\theta=r+1}^{r+s}\right) \cdot \prod_{\theta=r+1}^{r+s}\psi^{\beta_{\theta-r},n}(\vec{x}_\theta).
\end{align*}

And define
\begin{align*}
\mathcal{T}_1^n&\left(\{\vec{x}_\tau\}_{\tau=s+1}^{r}\right) \\
&:= \int_{\bR^4 \times \cdots \times \bR^4}\prod_{\theta=r+1}^{r+s}d\vec{x}_\theta \cdot \varphi_1\left( \{\vec{x}_\tau\}_{\tau=s+1}^{r}, \{\vec{x}_\theta\}_{\theta=r+1}^{r+s}\right) \cdot \prod_{\theta=r+1}^{r+s}\psi^{\beta_{\theta-r},n}(\vec{x}_\theta), \\
\mathcal{T}_2^n&\left(\{\vec{x}_\tau\}_{\tau=s+1}^{r} \right) \\
&:= \int_{\bR^4 \times \cdots \times \bR^4}\prod_{\tau=1}^{s}d\vec{x}_\tau \cdot \varphi_2\left( \{\vec{x}_\tau\}_{\tau=1}^{s}\right) \cdot \prod_{\tau=1}^{r}\psi^{\alpha_\tau,n}(\vec{x}_\tau),
\end{align*}
whereby $\psi^{\alpha_\tau,n}(\vec{x}_\tau)$ is either $\phi^{\alpha_\tau,n}(\vec{x}_\tau)$
or its adjoint operator. Likewise for $\psi^{\beta_{\theta-r},n}(\vec{x}_\theta)$.

We will need to assume the following:
\begin{itemize}
  \item both the sets $\{\psi^{\alpha_\tau, n}\}_{\tau = 1}^r$ and $\{\psi^{\beta_\theta, n}\}_{\theta = 1}^s$ are regular, as defined in Definition \ref{d.reg.1};
  \item $\sum_{\tau=1}^r \chi(\tau) + \sum_{\theta=r+1}^{r+s} \chi(\theta) = 0$, i.e. there are an equal number of creation and annihilation operators.
\end{itemize}

Suppose $\vec{a} \in \bR^4$. Define
\begin{align*}
T_1^n\left( \vec{a} \right) &:=
U(\vec{a},1) T_1^n U(\vec{a},1)^{-1}, \\
T_2^n\left(\vec{a}\right) &:=
U(\vec{a},1) T_2^n U(\vec{a},1)^{-1} .
\end{align*}

By default, we will write $T_i^n \equiv T_i^n(\vec{0})$, $i=1,2$.
Define similarly \beq \mathcal{T}_1^n\left(\{\vec{x}_\tau\}_{\tau=s+1}^{r},\vec{a}\right), \quad \mathcal{T}_2^n\left(\{\vec{x}_\tau\}_{\tau=s+1}^{r},\vec{a}\right), \nonumber \eeq replacing $T_i^n$ with $\mathcal{T}_i^n$ respectively in the above formulas.
\end{defn}

Refer to Definition \ref{d.t.1}. Define
\begin{align*}
&T_1 := \sum_{n=1}^\infty c_{1,n}T_1^n , \quad
T_2 := \sum_{n=1}^\infty c_{2,n}T_2^n, \\
&\mathcal{T}_1\left(\{\vec{x}_\tau\}_{\tau=s+1}^{r}\right) := \sum_{n=1}^\infty c_{1,n}\mathcal{T}_1^n\left(\{\vec{x}_\tau\}_{\tau=s+1}^{r}\right) , \\
&\mathcal{T}_2\left(\{\vec{x}_\tau\}_{\tau=s+1}^{r} \right) := \sum_{n=1}^\infty c_{2,n}\mathcal{T}_2^n\left(\{\vec{x}_\tau\}_{\tau=s+1}^{r} \right),
\end{align*}
each $\{c_{i,n}\}_{n=1}^\infty$ is a sequence of numbers in $l^2$ space, such that $\sum_{n=1}^\infty |c_{i,n}|^2C(\rho_n)^{2k} < \infty$, for any $k \in \mathbb{N}$.

\begin{rem}
Note that for $j=1,2$, $| T_j^n1 |^2 = O(C(\rho_n))$, $| \mathcal{T}_j^n1 |^2 = O(C(\rho_n))$. Because for any $k \in \mathbb{N}$, $\sum_{n=1}^\infty |c_{j,n}|^2C(\rho_n)^{2k} < \infty$, we see that \beq | \hat{H}^{k}T_j(\cdot, \vec{a})1 |^2 \leq \sum_{n=1}^\infty C(\rho_n)^{2k}|c_{j,n}|^2| T_j^n(\cdot)1|^2 < \infty. \label{e.a.1} \eeq
Similar bound for $\mathcal{T}_j$.
This is required to prove the Clustering Decomposition Property \ref{t.cd.1} and Clustering Theorem \ref{t.cl}.
\end{rem}

Indeed,
\begin{align*}
\Big\langle &T_1\left( \vec{y} \right) T_2\left( \vec{x} \right)1, 1 \Big\rangle \\
&= \sum_{n=1}^\infty c_{1,n}c_{2,n}\left\langle   U(\vec{y},1)T_1^n U(\vec{y},1)^{-1}U(\vec{x},1)
T_2^n
U(\vec{x},1)^{-1}1 , 1 \right\rangle \\
&= \sum_{n=1}^\infty c_{1,n}c_{2,n}\left\langle T_1^n U(\vec{x} - \vec{y},1) T_2^n1, 1 \right\rangle,
\end{align*}
which converges by Equation (\ref{e.a.1}).

Similarly,
\begin{align*}
\Big\langle & \mathcal{T}_2\left( \{\vec{x}_\tau\}_{\tau=s+1}^{r},\vec{x}\right)\mathcal{T}_1\left( \{\vec{x}_\tau\}_{\tau=s+1}^{r},\vec{y}\right)1, 1 \Big\rangle \\
&= \sum_{n=1}^\infty c_{1,n}c_{2,n}\left\langle \mathcal{T}_2^n\left( \{\vec{x}_\tau\}_{\tau=s+1}^{r},\vec{x}\right)
\mathcal{T}_1^n\left(\{\vec{x}_\tau\}_{\tau=s+1}^{r},\vec{y}\right) 1, 1 \right\rangle,
\end{align*}
which also converges.

\begin{defn}\label{d.h.2}
For $\vec{a}= \vec{x} -\vec{y} \in S_0$,
we will define
\begin{align*}
h_{12}^n\left( \vec{a}\right)
:=& \left\langle T_1^n\left(\vec{y}\right) P_0 T_2^n\left( \vec{x}\right)1, 1 \right\rangle\\
=&  \left\langle T_1^n U(\vec{y} ,1)^{-1} U(\vec{x} ,1)T_2^n1, 1 \right\rangle - \langle T_1^n1, 1 \rangle \langle T_2^n 1, 1 \rangle  \\
=& \left\langle T_1^n U(\vec{a} ,1)T_2^n1, 1 \right\rangle - \langle T_1^n1, 1 \rangle \langle T_2^n 1, 1 \rangle .
\end{align*}

Similarly, define ($s^+ = s+1$)
\begin{align*}
h_{21}^n&\left(\vec{a}\right) \\
:=& \int_{\bR^4 \times \cdots \times \bR^4 }\prod_{\tau=s^+}^{r}d\vec{x}_\tau\ \left\langle \mathcal{T}_2^n\left(\{\vec{x}_\tau\}_{\tau=s^+}^{r}, \vec{x}\right) P_0 \mathcal{T}_1^n\left(\{\vec{x}_\tau\}_{\tau=s^+}^{r}, \vec{y}\right)1, 1 \right\rangle\\
=& \left\langle \int_{\bR^{4(r-s)} }\prod_{\tau=s^+}^{r}d\vec{x}_\tau\ \mathcal{T}_2^n \left(\{\vec{x}_\tau\}_{\tau=s^+}^{r}\right) U(-\vec{a} ,1)\mathcal{T}_1^n\left( \{\vec{x}_\tau\}_{\tau=s^+}^{r}\right)1, 1 \right\rangle \\
&-  \int_{\bR^{4(r-s)} }\prod_{\tau=s^+}^{r}d\vec{x}_\tau\left\langle \mathcal{T}_1^n\left(\{\vec{x}_\tau\}_{\tau=s^+}^{r}\right)1, 1 \right\rangle \left\langle \mathcal{T}_2^n\left(\{\vec{x}_\tau\}_{\tau=s^+}^{r}\right) 1, 1 \right\rangle .
\end{align*}
\end{defn}

Using Definition \ref{d.ex.1}, we will extend the definitions of $h_{12}^n\left( \vec{a}\right)$ and $h_{21}^n\left( \vec{a}\right)$, for $\vec{a} \in S_0$, to all of $\vec{a} \in \bR^4$.

From Equation (\ref{e.c.9}), we will have ($r^+ = r+1$)
\begin{align}
h_{12}^n&\left(\vec{a}\right)   \nonumber \\
:=& \int_{\bR^4 \times \cdots \times \bR^4}\prod_{\tau=1}^{r+s}d\vec{x}_\tau \ \mathscr{W}^n\left(\{\vec{x}_\tau\}_{\tau=1}^r, \{\vec{x}_\theta + \vec{a}\}_{\theta=r^+}^{r+s}  \right)\varphi_1\left(\{\vec{x}_\tau\}_{\tau=1}^{r}  \right)\varphi_2\left(\{ \vec{x}_\theta\}_{\theta=r^+}^{r+s}  \right).
\label{e.e.1a}
\end{align}

Similarly, we will define ($s^+ = s+1$)
\begin{align}
h_{21}^n&\left(\vec{a}\right) \nonumber \\
:=& \int_{\bR^4 \times \cdots \times \bR^4}\prod_{\tau=1}^{r+s}d\vec{x}_\tau\
\mathscr{W}^n\left(\{\vec{x}_\theta+ \vec{a} \}_{\theta=r^+}^{r+s}, \{\vec{x}_\tau \}_{\tau=1}^{r} \right)\varphi_1\left(\{\vec{x}_\tau\}_{\tau=1}^{r}  \right)\varphi_2\left(\{ \vec{x}_\theta\}_{\theta=r^+}^{r+s}  \right).\label{e.e.1b}
\end{align}

When $\vec{a} \in S_0$, the LHS of Equations (\ref{e.e.1a}) and (\ref{e.e.1b}) are equal to
\begin{align*}
&\left\langle T_1^n U(\vec{a} ,1)T_2^n1, 1 \right\rangle - \langle T_1^n1, 1 \rangle \langle T_2^n 1, 1 \rangle \quad {\rm and} \\
&\left\langle \int_{\bR^{4(r-s)} }\prod_{\tau=s^+}^{r}d\vec{x}_\tau\ \mathcal{T}_2^n \left(\{\vec{x}_\tau\}_{\tau=s^+}^{r}\right) U(-\vec{a} ,1)\mathcal{T}_1^n\left( \{\vec{x}_\tau\}_{\tau=s^+}^{r}\right)1, 1 \right\rangle \\
&-  \int_{\bR^{4(r-s)} }\prod_{\tau=s^+}^{r}d\vec{x}_\tau\left\langle \mathcal{T}_1^n\left(\{\vec{x}_\tau\}_{\tau=s^+}^{r}\right)1, 1 \right\rangle \left\langle \mathcal{T}_2^n\left(\{\vec{x}_\tau\}_{\tau=s^+}^{r}\right) 1, 1 \right\rangle
\end{align*}
respectively.

\begin{defn}
Define for any $\vec{a} \in \bR^4$,
\beq h_{12}\left(\vec{a} \right)
:= \sum_{n=1}^\infty c_{1,n}c_{2,n} h_{12}^n\left(\vec{a}\right),\quad
h_{21}\left(\vec{a} \right) := \sum_{n=1}^\infty c_{1,n}c_{2,n} h_{21}^n\left(\vec{a}\right).
\nonumber \eeq
\end{defn}

\begin{defn}\label{d.c.1}
Recall from Definition \ref{d.ma.1}, we defined the set $\{\hat{H}(\rho_n), \hat{P}(\rho_n)\}_{n \in \mathbb{N}}$, such that $\hat{H}(\rho_n)^2 - \hat{P}(\rho_n)^2 = m_n^2$, with $m_n>0$ as the mass gap in $\mathscr{H}(\rho_n)$.

Let $\Lambda_n$ be a Lorentz transformation that maps the standard orthonormal basis $e_a \mapsto \tilde{f}_a^n := \Lambda_n e_a$, whereby $\tilde{f}_0^n$ and $\tilde{f}_1^n$ are \beq \frac{\hat{H}(\rho_n)}{m_n}e_0 + \frac{\hat{P}(\rho_n)}{m_n}e_1, \quad \frac{\hat{P}(\rho_n)}{m_n}e_0 + \frac{\hat{H}(\rho_n)}{m_n}e_1 \nonumber \eeq respectively, and $\tilde{f}_2^n = e_2$, $\tilde{f}_3^n = e_3$, which together spans the plane $S_0$.
\end{defn}

\begin{prop}\label{p.c.1}
Write $\vec{a} = (a^-, a^+)$, whereby $a^- = (a^0, a^1)$, $a^+ = (a^2, a^3)$. If $x_\tau^+ - x_\theta^+ \neq a^+$ for any $1 \leq \tau \leq r$ and $r+1 \leq \theta \leq r+s$, then we have that \beq \mathscr{W}^n\left(\{\vec{x}_\tau\}_{\tau=1}^r, \{\vec{x}_\theta + \vec{a}\}_{\theta=r+1}^{r+s}  \right) = 0. \nonumber \eeq
\end{prop}

\begin{proof}
It is not difficult to see from Equation (\ref{e.q.2}) that
\begin{align}
\mathscr{W}^n&\left(\{\vec{x}_\tau\}_{\tau=1}^r, \{\vec{x}_\theta + \vec{a}\}_{\theta=r+1}^{r+s}  \right) \nonumber\\
=&
e^{-i\vec{a} \cdot c_rm_n\tilde{f}_0^n}\prod_{\theta=1}^{r+s} E(x_\theta^-)\cdot \mathscr{W}_0^n\left(\{ x^+_\tau\}_{\tau=1}^r, \{ x_\theta^+ + a^+\}_{\theta=r+1}^{r+s}  \right), \label{e.s.1}
\end{align}
whereby $E(x_\theta^-) = E(x_\theta^0, x_\theta^1) = \dfrac{e^{i\chi(\theta)[x_\theta^0\hat{H} - x_\theta^1\hat{P}]}}{2\pi}$ and ($0^- \equiv (0,0)$) \beq \mathscr{W}_0^n\left(\{x^+_\tau\}_{\tau=1}^r, \{x_\theta^+ + a^+\}_{\theta=r+1}^{r+s}  \right) = \mathscr{W}^n\left(\{0^-, x^+_\tau\}_{\tau=1}^r, \{0^-, x_\theta^+ + a^+\}_{\theta=r+1}^{r+s}  \right). \label{e.w.1} \eeq

Because we subtract off $\left\langle A_r^n1,1 \right\rangle \left\langle B_s^n1, 1 \right\rangle$ in Equation (\ref{e.c.3}),  we see that $c_Q$ defined in Equation (\ref{e.q.1}) is non-zero, only if there exists a set $A_l \in Q$ which contains both $r$ and $r+1$. From Equation (\ref{e.c.9}), we have that $\mathscr{W}_0^n \equiv 0$, if $x_\tau^+ - x_\theta^+ \neq a^+$ for any $1 \leq \tau \leq r$ and $r+1 \leq \theta \leq r+s$, and this completes the proof.
\end{proof}

\begin{rem}
From Item \ref{i.ex.1b} in Lemma \ref{l.ex.1}, we see that $\mathscr{W}^n = 0$ iff for any partition $Q = \{A_1, \ldots, A_{n(Q)}\} \in \Gamma$ with $c_Q$ defined in Equation (\ref{e.q.1}) is non-zero, there exists some $\theta \neq r$ such that $\{\theta, \theta+1\} \subset A_l$ for some $1 \leq l \leq n(Q)$ and $\xi_\theta^+ \neq (0,0)$, or $\xi_r^+ \neq a^+$.
\end{rem}

\begin{prop}\label{p.c.2}
Refer to Definition \ref{d.c.1}. Recall the sets $\kappa_r \subseteq \{1, \cdots, r\}$, $\pi_s \subseteq \{1, \cdots, s\}$, as explained in Notation \ref{n.q.2}. Fix $a^2, a^3 \in \bR$.

Suppose one of the 3 cases hold in in Remark \ref{r.p.1} and \beq \mathscr{W}_{\kappa_r, \pi_s}^n\left(\{\vec{x}_\tau\}_{\tau=1}^r, \{\vec{x}_\theta + \vec{a}\}_{\theta=r+1}^{r+s}  \right) = \pm\mathscr{W}_{\pi_s, \kappa_r}^n\left(\{\vec{x}_\theta + \vec{a}\}_{\theta=r+1}^{r+s} , \{\vec{x}_\tau\}_{\tau=1}^r  \right) \nonumber \eeq holds for any $a^0 \in \bR$, such that $\vec{a} = a^0 \tilde{f}_0^n + a^2 \tilde{f}_2^n + a^3 \tilde{f}_3^n$ is a space-like vector. Then it is necessary and sufficient that $\mathscr{W}_{\kappa_r, \pi_s}^n\left(\{\vec{x}_\tau\}_{\tau=1}^r, \{\vec{x}_\theta + \vec{a}\}_{\theta=r+1}^{r+s}  \right)$ and $\mathscr{W}_{\pi_s, \kappa_r}^n\left(\{\vec{x}_\theta + \vec{a}\}_{\theta=r+1}^{r+s} , \{\vec{x}_\tau\}_{\tau=1}^r  \right)$ both must vanish.
\end{prop}

\begin{proof}
We only consider when the sign is a plus. The other case is similar. From Equation (\ref{e.q.2}), we have
\begin{align*}
\mathscr{W}_{\kappa_r, \pi_s}^n\left(\{\vec{x}_\tau\}_{\tau=1}^r, \{\vec{x}_\theta + \vec{a}\}_{\theta=r+1}^{r+s}  \right)
=&
e^{ia^0 c_rm_n}\mathscr{W}_{\kappa_r, \pi_s}^n\left(\{\vec{x}_\tau\}_{\tau=1}^r, \{\vec{x}_\theta + a^+ \}_{\theta=r+1}^{r+s}  \right),
\end{align*}
and
\begin{align*}
\mathscr{W}_{\pi_s, \kappa_r}^n\left(\{\vec{x}_\theta + \vec{a}\}_{\theta=r+1}^{r+s} , \{\vec{x}_\tau\}_{\tau=1}^r  \right)
=&
e^{-ia^0  c_sm_n}\mathscr{W}_{\pi_s, \kappa_r}^n\left(\{\vec{x}_\theta + a^+\}_{\theta=r+1}^{r+s} , \{\vec{x}_\tau\}_{\tau=1}^r  \right).
\end{align*}
Because we assume one of the 3 cases hold, $c_s$ and $c_r$ have the same sign and thus the preceding two equations are not equal for general values of $a^0$, unless both terms are equal to zero.
\end{proof}

Recall from Definition \ref{d.n.1}, we defined the norm $\parallel \cdot \parallel_{p,q}$ for Schwartz functions on $\bR^4$. We can generalize this definition, for a Schwartz function defined on the Euclidean space $\bR^{r}$, for any $r \in \mathbb{N}$.

By considering $h_{12} - h_{21}$, we can mimic the proof of Theorem 3-4 (Cluster Decomposition Property) in \cite{streater}, or the lemma in \cite{Ruelle}, to prove that $\left| h_{12}\left(\vec{a}\right)\right|$ decays at the rate of $\frac{1}{|a^+|^p}$ for $|a^+|$ large and for any $p \in \mathbb{N}$. However, using Lemma \ref{l.lc.1} is inadequate to prove this result. We instead need local commutativity as given by Proposition \ref{p.c.2} and assuming one of the 3 cases hold in Remark \ref{r.p.1}, to prove the result. We will however, prove the following slightly stronger result, using Proposition \ref{p.c.1} and without assuming any of the cases in Remark \ref{r.p.1} hold.

\begin{thm}\label{t.cd.1}(Cluster Decomposition Property)\\
Recall $m_0$ is the mass gap of the 4-dimensional quantum Yang-Mills theory. Suppose the assumptions in Definition \ref{d.t.1} hold, i.e. regularity of the two sets and equal number of creation and annihilation operators. Let $\vec{a} = \sum_{b=0}^3 a^b e_b$ be a space-like vector, with $a^- \equiv (a^0, a^1)$ and $a^+ \equiv (a^2, a^3) \neq (0,0)$. Then, we have for some positive constant $C$,
\begin{align}
\left| h_{12}\left(\vec{a}\right)\right| &= \left|\sum_{n=1}^\infty c_{1,n}c_{2,n}h_{12}^n\left(\vec{a}\right)\ \right| \nonumber \\
&\leq \frac{C}{m_0^l + |a^+|^k}\left(\parallel \varphi_1 \parallel_{p,q}\parallel \varphi_2 \parallel_{p,q}\right),\label{e.cd.1}
\end{align}
for any $k,l \in \mathbb{N}$, provided $p,q \in \mathbb{N}$ are large enough.
\end{thm}

\begin{proof}
Recall the sets $\kappa_r \subseteq \{1, \cdots, r\}$, $\pi_s \subseteq \{1, \cdots, s\}$, as explained in Notation \ref{n.q.2}. Consider the real vector spaces $\mathcal{Y}_{\kappa_r}$ and $\widetilde{\mathcal{Y}}_{\pi_s}$, each spanned by tensor products of functions of the form \beq f_1 \otimes_\bR f_2 \otimes_\bR \cdots \otimes_\bR f_r  \quad {\rm and} \quad g_1 \otimes_\bR \cdots \otimes_\bR g_s \nonumber \eeq respectively, whereby each $f_i$ ($g_j$) is real if $i \in \kappa_r$ ($j \in \pi_s$), purely imaginary if $i \notin \kappa_r$ ($j \notin \pi_s$).

Instead of proving the result for general $\varphi_i$'s, from Equation (\ref{e.d.3}), it suffices to prove it for $\varphi_1 \in \mathcal{Y}_{\kappa_r}$ and $\varphi_2 \in \widetilde{\mathcal{Y}}_{\pi_s}$, whereby we need to change the definition of $h_{12}^n$, by replacing $\mathscr{W}^n$ with $\mathscr{W}_{\kappa_r, \pi_s}^n$ in Equation (\ref{e.e.1a}).

Without any loss of generality, we assume that $C(\rho_n) > 1$, otherwise we can replace $C(\rho_n)$ in the following inequalities with $(1+ C(\rho_n))$ instead. Write \beq \varphi\left( \{\vec{x}_\theta\}_{\theta=1}^{r+s} \right) = \varphi_1\left(\{\vec{x}_\tau\}_{\tau=1}^{r}  \right)\otimes_\bR\varphi_2\left(\{\vec{x}_\theta\}_{\theta=r^+}^{r+s}  \right). \nonumber \eeq

Let $q^- = (q^0, q^1)$ be coordinates pertaining to the $\{e_0, e_1\}$ basis, with $x_\tau^-\cdot q_\tau^- = -x_\tau^0q_\tau^0 + x_\tau^1 q_\tau^1$, and define the following partial Fourier Transforms,
\begin{align}
\begin{split}
\tilde{\varphi}_1\left( \{q_\tau^-,x_\tau^+\}_{\tau=1}^{r} \right) &:= \int_{\bR^{2r}}
\prod_{\tau=1}^{r}\frac{e^{-i\chi(\tau)x_\tau^-\cdot q_\tau^-}}{2\pi}\cdot \varphi_1\left(\{x_\tau^-, x_\tau^+\}_{\tau=1}^{r} \right)\ \prod_{\tau=1}^{r+s}dx_\tau^-, \\
\tilde{\varphi}_2\left( \{q_\theta^-,x_\theta^+\}_{\theta=r+1}^{r+s} \right) &:= \int_{\bR^{2s}}
\prod_{\theta=r+1}^{r+s}\frac{e^{-i\chi(\theta)x_\theta^-\cdot q_\theta^-}}{2\pi}\cdot \varphi_2\left(\{x_\theta^-, x_\theta^+\}_{\theta=r+1}^{r+s} \right)\ \prod_{\theta=r+1}^{r+s}dx_\theta^-,
\end{split}\label{e.s.2}
\end{align}
and \beq \tilde{\varphi}\left( \{q_\tau^-,x_\tau^+\}_{\tau=1}^{r+s} \right) :=
\tilde{\varphi}_1\left( \{q_\tau^-,x_\tau^+\}_{\tau=1}^{r} \right)\tilde{\varphi}_2\left( \{q_\theta^-,x_\theta^+\}_{\theta=r+1}^{r+s} \right). \nonumber \eeq

Suppose we let $H_n^- = (\hat{H}(\rho_n), \hat{P}(\rho_n))$ and recall we defined $\mathscr{W}_0^n$ in Equation (\ref{e.w.1}). Then, we see that
\begin{align*}
&h_{12}^n(\vec{a}) \\
&= e^{-ic_rm_n\vec{a} \cdot \tilde{f}_0^n}\int_{\bR^{2(r+s)}}\prod_{\tau=1}^{r+s}dx_\tau^+\   \mathscr{W}_0^n\left(\{x_\tau^+\}_{\tau=1}^r, \{x_\theta^+ + a^+\}_{\theta=r^+}^{r+s}  \right)
\tilde{\varphi}\left( \{H_n^-,x_\tau^+\}_{\tau=1}^{r+s} \right),
\end{align*}
which shows that it is a constant in the $\tilde{f}_1^n$ direction.

Since the Fourier transform of a Schwartz function remains a Schwartz function, we see that for some $\bar{n} \in \mathbb{N}$, independent of $n$, \beq \left| \tilde{\varphi}\left( \{q_\theta^-,x_\theta^+\}_{\theta=1}^{r+s} \right) \right| \leq C(\rho_n)^{\bar{n}}\frac{\parallel \varphi \parallel_{\hat{k},\hat{l}}}{\sum_{\theta=1}^{r+s}(|q_\theta^{0}|^2 + |q_\theta^{1}|^2)^{l/2} + \left(|x_\theta^{2}|^2 + |x_\theta^{3}|^2\right)^{k/2}}, \nonumber \eeq for some $\hat{k}\geq k$, $\hat{l} \geq l$ large enough. In particular, $(\hat{H}(\rho_n)^2 + \hat{P}(\rho_n)^2)^{l/2} \geq m_n^l$. Thus, we have that  \beq \left| \tilde{\varphi}\left( \{H_n^-,x_\theta^+\}_{\theta=1}^{r+s} \right) \right| \leq C(\rho_n)^{\bar{n}}\frac{\parallel \varphi \parallel_{\hat{k},\hat{l}}}{\sum_{\theta=1}^{r+s}m_n^l + \left(x_\theta^{2,2} + x_\theta^{3,2}\right)^{k/2}}. \nonumber \eeq

There exists a set containing polynomially bounded continuous functions $\{G_{\vec{m}}^n:\ 0 \leq |\vec{m}| \leq \mathcal{N}\}$, such that
\begin{align*}
e^{-ic_rm_n\vec{a} \cdot \tilde{f}_0^n}\sum_{|\vec{m}| \leq \mathcal{N}}\int_{\bR^{2(r+s)}}&\prod_{\tau=1}^{r+s}dx_\tau^+ D^{\vec{m}} G_{\vec{m}}^n\left(\{x_\tau^+\}_{\tau=1}^{r+s}; a^+ \right)\tilde{\varphi}\left( \{H_n^-, x_\theta^+\}_{\theta=1}^{r+s} \right) \\
=& h_{12}^n \left(\vec{a}\right),
\end{align*}
and
\begin{align*}
\sum_{|\vec{m}| \leq \mathcal{N}}D^{\vec{m}} G_{\vec{m}}^n\left(\{ \{x_\tau^+\}_{\tau=1}^{r+s}; a^+ \right)
\equiv& \mathscr{W}_0^n\left(\{x_\tau^+\}_{\tau=1}^r, \{x_\theta^++ a^+\}_{\theta=r+1}^{r+s}  \right),
\end{align*}
with
\beq \sum_{|\vec{m}| \leq \mathcal{N}}|G_{\vec{m}}^n|\left( \{x_\tau^+\}_{\tau=1}^{r+s}; a^+ \right) \leq C(\rho_n)^{\tilde{k}} \left[|a^+|^\alpha +  \left(\sum_{\tau=1}^{r+s}|x_\tau^+|^2 \right)^{\gamma/2} \right], \nonumber\eeq for constants $\tilde{k}, \alpha, \gamma$, all independent of $n$.

Choose $R_0 = |a^+|/2 > 0$. For $0 < \epsilon < |a^+|/2$ small, the set \beq  \{R > R_0-\epsilon\} := \left\{ \{x_\tau^+\}_{\tau=1}^{r+s}:\ \sum_{\tau=1}^{r+s}|x_\tau^+|^2 > (R_0-\epsilon)^2 \right\} \subset \bR^{2(r+s)}. \nonumber \eeq

Hence, we can apply Proposition \ref{p.c.1} and $D^{\vec{m}} G_{\vec{m}}^n\left(\cdot; a^+  \right)$ vanishes on its complement. Thus we have
\begin{align*}
\sum_{|\vec{m}| \leq \mathcal{N}}&\int_{\bR^{2(r+s)}} \prod_{\tau=1}^{r+s}dx_\tau^+\ D^{\vec{m}} G_{\vec{m}}^n\left( \{x_\tau^+\}_{\tau=1}^{r+s}; a^+ \right)\tilde{\varphi}\left( \{ q_\theta^-, x_\theta^+\}_{\theta=1}^{r+s} \right)  \\
&= \sum_{|\vec{m}| \leq \mathcal{N}}\int_{\{R > R_0-\epsilon\}}\prod_{\tau=1}^{r+s}dx_\tau^+\ D^{\vec{m}}G_{\vec{m}}^n\left(\{x_\tau^+\}_{\tau=1}^{r+s}; a^+ \right)\tilde{\varphi}\left( \{q_\theta^-, x_\theta^+\}_{\theta=1}^{r+s} \right).
\end{align*}

Using spherical coordinates,
\begin{align*}
&\bigg|h_{12}^n  \left(\vec{a}\right)\bigg| \\
&\leq
\sum_{|\vec{m}| \leq \mathcal{N}}\int_{\{R > R_0-\epsilon\}}\prod_{\tau=1}^{r+s}dx_\tau^+\ \left|G_{\vec{m}}^n\left( \{x_\tau^+\}_{\tau=1}^{r+s}; a^+ \right)D^{\vec{m}}\tilde{\varphi}\left( \{H_n^-, x_\theta^+\}_{\theta=1}^{r+s} \right)\right| \\
&\leq C(\rho_n)^{\tilde{K}}
\int_{\{R > R_0-\epsilon\}}  \left[|a^+|^\alpha + R^\gamma \right]\frac{\parallel \varphi \parallel_{p,q}}{(r+s)|m_n|^{l} + R^{k+ \bar{l}}}\frac{ R^{2(r+s)}}{R}dR d\Omega,
\end{align*}
for some norm $\parallel \cdot \parallel_{p,q}$, and $k, l, \bar{l}$ can be any natural numbers less than $p$, $q$. And $\tilde{K} \geq \tilde{k} + \bar{n}$ is some natural number, independent of $n$.

By choosing $p\geq k+\bar{l}, q\geq l$ and $\tilde{K}$ large enough, we see that the above inequality is less than some positive constant times
\begin{align*}
C(\rho_n)^{\tilde{K}}\frac{\parallel \varphi_1 \parallel_{p,q}\parallel \varphi_2 \parallel_{p,q}}{m_n^l + |a^+|^k},
\end{align*}
for any $|a^+|$.

Since the above equation holds for any non-compact $\varphi_i$, we have for each $n$,
\begin{align*}
\mathscr{W}_{\kappa_r, \pi_s}^n&\left(\{\vec{x}_\tau\}_{\tau=1}^r, \{ \vec{x}_\theta + \vec{a}\}_{\theta=r^+}^{r+s}  \right) \longrightarrow 0,
\end{align*}
as $|a^+| \equiv |(a^2, a^3)| \rightarrow \infty$, decaying at the rate of $\dfrac{1}{m_n^l + |a^+|^k}$, for any power $k, l \in \mathbb{N}$.

For some fixed integer $\tilde{K}> 0$, $C = O\left(\sum_{n=1}^\infty |c_{1,n}c_{2,n}|C(\rho_n)^{\tilde{K}} \right)$ from Equation (\ref{e.a.1}). Since $m_0 \leq m_n$ for all $n \in \mathbb{N}$, we have
\begin{align*}
\bigg|h_{12} \left( \vec{a}\right)\bigg|
&\leq\sum_{n=1}^\infty |c_{1,n}c_{2,n}|\bigg|h_{12}^n \left( \vec{a}\right)\bigg|  \\
&\leq \frac{C}{m_0^l + |a^+|^k}\left(\parallel \varphi_1 \parallel_{p,q}\parallel \varphi_2 \parallel_{p,q}\right).
\end{align*}
\end{proof}

\begin{notation}\label{n.q.3}
Refer to Notation \ref{n.q.1}. Let $\Omega_r$ denote the set consisting of partitions in $\Gamma$, such that $Q=\{A_1, \ldots, A_{n(Q)}\} \in \Omega_r$, iff there exists a (unique) $1 \leq z \leq n(Q)$, such that both $r, r+1 \in A_z$.

Recall $H_n^- = (\hat{H}(\rho_n), \hat{P}(\rho_n))$. Write for $1 \leq i \leq n(Q)$, $n_0 := 0$, $n_i = \sum_{k=1}^i |A_k|$ and
\beq \{y_j^+\}^{|A_i|} = \Big(\underbrace{H_n^-,y_i^+, H_n^-,y_i^+, \cdots, H_n^-,y_i^+}_{|A_i|\ {\rm copies}\ {\rm of}\ (H_n^-, y_i^+)} \Big). \nonumber \eeq

Refer to Equation (\ref{e.s.2}) for the notations \beq
\tilde{\varphi}_1\left( \{H_n^-, x_\tau^+\}_{\tau=1}^{r} \right)\quad {\rm and} \quad \tilde{\varphi}_2\left( \{H_n^-, x_{\theta}^+\}_{\theta=r+1}^{r+s} \right). \nonumber \eeq

For a given partition $Q= \{A_1, \ldots, A_{n(Q)}\} \in \Omega_r$, we will write
\begin{align*}
&\tilde{\varphi}_{Q,1}^n\left( \{x_\theta^+\}_{\theta = n_{z-1}+1}^r\right) \\
&:= \int_{\bR^{2(z-1)}}\prod_{i=1}^{z-1}dy_i^+\tilde{\varphi}_{1}\left(\{y_j^+\}^{|A_1|},
\{y_j^+\}^{|A_2|}, \cdots, \{y_{j}^+\}^{|A_{z-1}|}, \{H_n^-, x_{\theta}^+\}_{\theta = n_{z-1}+1}^r\right).
\end{align*}

Similarly, we will write ($z^+ = z+1$)
\begin{align*}
&\tilde{\varphi}_{Q,2}^n\left( \{x_\theta^+\}_{\theta = r^+}^{n_z}\right) \\
&:= \int_{\bR^{2(n(Q)-z)}}\prod_{i=z^+}^{n(Q)}dy_i^+
\tilde{\varphi}_{2}\left(\{H_n^-, x_{\theta}^+\}_{\theta=r^+}^{n_z}, \{y_j^+\}^{|A_{z^+}|}, \{y_j^+\}^{|A_{z+2}|}, \cdots, \{y_{j}^+\}^{|A_{n(Q)}|} \right).
\end{align*}

Note that $\{n_{z-1}+1, \cdots, r, r+1, \cdots, n_z\} = A_z \in Q$. Finally, we will write \beq \tilde{\varphi}_{Q,1}^n(x_r^+) := \tilde{\varphi}_{Q,1}^n\left( \{x_r^+\}_{\theta = n_{z-1}+1}^r\right), \quad \tilde{\varphi}_{Q,2}^n(x_{r+1}^+) := \tilde{\varphi}_{Q,2}^n\left( \{x_{r+1}^+\}_{\theta = r+1}^{n_z}\right). \nonumber \eeq

From Equations (\ref{e.c.9}) and (\ref{e.e.1a}), we will leave to the reader to verify that for any $\vec{a} \equiv (a^-, a^+) \in \bR^4$,
\begin{align}
&h_{12}(\vec{a}) = e^{-ic_rm_n\vec{a} \cdot \tilde{f}_0^n}\sum_{Q \in \Omega_r} c_Q\int_{\bR^2}dx^+\ \tilde{\varphi}_{Q,1}^n(x^+)\tilde{\varphi}_{Q,2}^n(x^+ - a^+) \nonumber \\
&= e^{-ic_rm_n\vec{a} \cdot \tilde{f}_0^n}\sum_{Q \in \Omega_r} c_Q\int_{\bR^2}dx_r^+\int_{\bR^2}dx_{r^+}^+\ \tilde{\varphi}_{Q,1}^n(x_r^+)\delta(a^+ + x_{r^+}^+ - x_r^+)\tilde{\varphi}_{Q,2}^n(x_{r^+}^+ ),
\label{e.s.3}
\end{align}
whereby $\delta$ is Dirac delta function. Note that the coefficients $c_Q$ were defined in Equation (\ref{e.q.1}).
\end{notation}

\begin{thm}(Clustering Theorem)\label{t.cl}\\
Suppose the assumptions in Definition \ref{d.t.1} hold with $\varphi_i$ is compactly supported for $i=1,2$.
Assume that $0 \neq c_r$ in Equation (\ref{e.l.1}).

Recall $m_0$ is the mass gap of the 4-dimensional quantum Yang-Mills theory, and $\vec{a} = \sum_{b=0}^3 a^b e_b$, with $a^-=(a^0, a^1)$, $a^+ = (a^2, a^3)$. There exist positive constants $\epsilon$, $C$, and $p_i, q_i \in \mathbb{N}$ for $i=1,2$, such that
\begin{align}
\left| h_{12}\left(\vec{a}\right)\right| &= \left|\sum_{n=1}^\infty c_{1,n}c_{2,n}h_{12}^n\left(\vec{a}\right)\right| \nonumber \\
&\leq C\frac{e^{-m_0|a^+| } }{|a^+| - \epsilon}\parallel \varphi_1 \parallel_{p_1,q_1}\parallel \varphi_2 \parallel_{p_2,q_2}, \label{e.c.5}
\end{align}
provided $|a^+| > \epsilon $, with $\epsilon$ depending only on the support of $\varphi_1$ and $\varphi_2$. And, $C = O\left(\sum_{n=1}^\infty |c_{1,n}c_{2,n}|C(\rho_n)^{\tilde{k}}e^{m_0\epsilon} \right)$ for some fixed integer $\tilde{k}> 0$. Note that $a^+= (a^2, a^3)$ are coordinates with respect to $\{e_2, e_3\}$, spanning the $x^2-x^3$ plane $S_0$.
\end{thm}

\begin{proof}
Just like in the proof of Theorem \ref{t.cd.1}, it suffices to prove it for $\varphi_1 \in \mathcal{Y}_{\kappa_r}$ and $\varphi_2 \in \widetilde{\mathcal{Y}}_{\pi_s}$. Recall in Definition \ref{d.c.1}, we defined a new  basis $\{\tilde{f}_a^n\}_{a=0}^3$, which depends on the representation $\rho_n$ on $\mathfrak{g}$. Let $(a^0, a^1, a^2, a^3)$ be coordinates with respect to the basis $\{\tilde{f}_a^n\}_{a=0}^3$. Further recall that $h_{12}^n$ is independent of the $a^1$ coordinate, so we will henceforth write $h_{12}^n(\vec{a}) = h_{12}^n(a^0, a^+)$, $a^+ = (a^2, a^3)$.

Define \beq D(\varphi_1, \varphi_2) = \left\{ \vec{x}_\tau - \vec{y}_{\theta} \in \bR^4|\ \{\vec{x}_\tau\}_{\tau=1}^r \in {\rm supp}\ \varphi_1,\ \{\vec{y}_\theta\}_{\theta=r+1}^{r+s} \in {\rm supp}\ \varphi_2  \right\}, \nonumber \eeq and
consider orthogonally projecting it onto the $x^2-x^3$ plane spanned by $\{e_2, e_3\}$, this projected set denoted by $D^0$. It is bounded.

Denote a compact set $D^1$, which is the convex closure of $D^0$, which lies in the plane spanned by $\{e_2, e_3\}$. Then, we see that there exists an $\epsilon>0$ such that $D_1$ lies in an open disc of radius $\epsilon$, center at the origin.

Let $g$ be any compact Schwartz function on $\bR$. Consider the following elliptic equation\footnote{We used an elliptic equation, because the total energy squared is then given by $(c_rm_n)^2 + \sum_{i=1}^3 q^{i,2}$, which is mass squared plus total momentum squared. The generator for translation in the time-like direction $\tilde{f}_0^n$, is given by the mass eigenvalue $-c_rm_n$, and not energy eigenvalue.} \beq \left(\frac{\partial^2}{\partial a^{0,2}} - \hat{\mathcal{P}}\right)\Psi_n(a^0, a) := \left( \sum_{i=0}^3\frac{\partial^2}{\partial a^{i,2}}\right)\Psi_n(a^0, a) = h_{12}^n(a^0, a^+)g(a^1). \nonumber \eeq

Refer to Notation \ref{n.q.3}. We leave to the reader to verify from Equation (\ref{e.s.3}), ($r^+ \equiv r+1$)
\begin{align*}
\Psi_n&(a^0, a) = e^{ic_rm_n a^0}\tilde{G}(n; a), \ {\rm and}\\
\tilde{G}&(n; a) \\
=&\sum_{Q \in \Omega_r} c_Q\int_\bR d\xi^1\int_{\bR^{2\times 2}}dx_r^+ dx_{r^+}^+\tilde{\varphi}_{Q,1}^n(x_r^+)G_n(a^1-\xi^1;x_r^+, x_{r^+}^+ + a^+)\tilde{\varphi}_{Q,2}^n(x_{r^+}^+)g(\xi^1),
\end{align*}
whereby $G_n(a^1-\xi^1; x_r^+, x_{r+1}^+ + a^+)$ is a Green's function that solves
\begin{align*}
\left(-(c_rm_n)^2 + \sum_{i=1}^3\frac{\partial^2}{\partial a^{i,2}}\right) G_n(a^1-\xi^1; x_r^+, x_{r+1}^+ + a^+) &= \delta(a^1-\xi^1, x_r^+ - x_{r+1}^+ - a^+) \\
&\equiv \delta(a - \xi_r),
\end{align*}
whereby $\delta$ is the Dirac delta function, $a= (a^1, a^+)$ and $\xi_r = (\xi^1, \xi_r^+) \equiv (\xi^1, x_r^+ - x_{r+1}^+)$.

Write $\omega = |c_rm_n| \neq 0$ and $q = (q^1, q^2, q^3)$ be the momentum coordinates dual to $a \equiv (a^1, a^2, a^3)$. Using Fourier Transform techniques, the solution is given by \beq G_n(a^1-\xi^1; x_r^+, x_{r+1}^+ + a^+) = -\frac{1}{(2\pi)^{3/2}}\int_{\bR^3}dq\ \frac{e^{iq\cdot (a-\xi_r)}}{\omega^2 + |q|^2}. \nonumber \eeq

Note that the integrand is invariant under spatial rotation. Without any loss of generality, we can assume that $(a-\xi_r) \equiv (a^1-\xi^1,a^+ - \xi_r^+)$ lies in the $z$-axis. Hence, $q \cdot (a - \xi_r) = |q|R\cos \theta$, whereby $R = |( a - \xi_r)|$ and $\theta$ is the azimuthal angle a vector in 3-dimensional space makes with the positive $z$- axis.

Integrate out the polar and azimuthal angles in spherical coordinates, we have
\begin{align*}
G_n(a^1-\xi^1; x_r^+, x_{r+1}^+ + a^+) &= -\frac{1}{(2\pi)^{1/2}}\frac{1}{iR}\int_0^\infty  \lambda\frac{e^{iR\lambda} - e^{-iR\lambda}}{\omega^2 + \lambda^2}\ d\lambda \\
&= -\frac{1}{(2\pi)^{1/2}}\frac{1}{iR}\int_{-\infty}^\infty  \frac{\lambda e^{iR\lambda}}{(\lambda- i\omega)(\lambda + i\omega)}\ d\lambda.
\end{align*}

We will evaluate this integral using Cauchy's Residue Theorem, by considering an upper semi-circle, which contains a simple pole at $z = i\omega$ in its interior. We choose the upper semi-circle because on the upper half plane, $|e^{iRz}| \leq 1$ has an exponential decay. Finally taking its radius to be infinity, the result is that the integral is equal to \beq G_n(a^1-\xi^1; x_r^+, x_{r+1}^+ + a^+) = -\frac{2\pi i}{(2\pi)^{1/2}}\frac{1}{iR} \frac{i\omega e^{-R\omega}}{2i\omega} = -\sqrt{2\pi}\frac{e^{-R|c_rm_n|}}{2R}, \nonumber \eeq with $R^2 = |a^1 -\xi^1|^2 + |a^2 + \xi_r^2|^2 + |a^3 + \xi_r^3|^2$.

Write
\begin{align*}
K&(a^+, c_rm_n) \\
&:= \int_{\bR^2}dy^+ \left|\tilde{\varphi}_{Q,1}^n(y^+)\right|\cdot \int_{\bR^2}dx^+\left|\left[(c_rm_n)^2 + \hat{\mathcal{P}}\right]\tilde{\varphi}_{Q,2}^n(x^+)\right|\cdot \sqrt\frac{\pi}{2}\frac{e^{-m_0|a^+ + x^+ -y^+|}}{|a^+ +x^+ - y^+|}.
\end{align*}
Therefore, ($G_n(a-\xi) := G_n(a^1-\xi^1; x^-, x^+ + a^+)$)
\begin{align*}
&\left|h_{12}^n(\vec{a})g\left(a^1 \right)\right| = \left|\left[\frac{\partial^2}{\partial a^{0,2}} - \hat{\mathcal{P}}\right]\Psi_n(a^0, a)\right| \\
&= \left|\sum_{Q \in \Omega_r}c_Q\int_\bR d\xi^1\int_{\bR^{4}}d\vec{x}\ \tilde{\varphi}_{Q,1}^n(x^-)G_n(a-\xi)\left[(c_rm_n)^2 + \hat{\mathcal{P}}\right] \tilde{\varphi}_{Q,2}^n(x^+)g(\xi^1) \right| \\
&\leq \sum_{Q \in \Omega_r}|c_Q|\int_{\bR}\ d\xi^1 \left[|g(\xi^1)| + \frac{1}{m_0^2}|g''(\xi^1)| \right] K(a^+, c_rm_n)   \\
&\leq C(\rho_n)^{\tilde{k}}e^{m_0\epsilon}\parallel g \parallel_{p_3, q_3} \parallel \varphi_1 \parallel_{p_1, q_1} \parallel \varphi_2 \parallel_{p_2, q_2}\cdot \frac{e^{-m_0|a^+|}}{|a^+| - \epsilon},
\end{align*}
for some natural numbers $\tilde{k}$, $p_i$'s, $q_i$'s, provided $|x^+ - y^+| \leq \epsilon < |a^+|$ is large enough. Note that all these natural numbers are independent of $n$, $a^+$ and the Schwartz functions. We can assume that $C(\rho_n) \geq 1$, otherwise we will replace it with $1 + C(\rho_n)$ in the above inequality.

Because the above inequality holds for any compact Schwartz function $g: \bR \rightarrow \bR$ and any $a^1 \in \bR$, we have the inequality \beq |h_{12}^n(\vec{a})|\leq C(\rho_n)^{\tilde{k}}e^{m_0\epsilon} \parallel \varphi_1 \parallel_{p_1, q_1} \parallel \varphi_2 \parallel_{p_2, q_2}\cdot \frac{e^{-m_0|a^+|}}{|a^+| - \epsilon}, \nonumber \eeq for $|a^+|> \epsilon$.

Since the above inequality holds for any $n \in \mathbb{N}$, we see that Equation (\ref{e.c.5}) follows immediately from triangle inequality and Equation (\ref{e.a.1}).
\end{proof}

\begin{rem}
When the space like separation $\vec{a}$ lies in the $x^2-x^3$ plane or $S_0$, then $h_{12}(\vec{a})$ is equivalent to the vacuum expectation $\left\langle T_1^n\left(\vec{y}\right) P_0 T_2^n\left( \vec{x}\right)1, 1 \right\rangle$ from Definition \ref{d.h.2}. Suppose we are given two sets of cluster points $\{\vec{x}_{\theta} \in \bR^4\}_{\theta=1}^r$, $\{\vec{y}_{\tau}\in \bR^4\}_{\tau=1}^s$, and project them orthogonally onto $S_0$, as $A = \{x_{\theta}^+ \in S_0\}_{\theta=1}^r$, $B = \{y_{\tau}^+\in S_0\}_{\tau=1}^s$ respectively. Clustering Theorem hence says that the vacuum expectation decays exponentially, at a rate dependent on the mass gap $m_0$, provided the space-like separation $a^2 e_2 + a^2e_3 \in S_0$ between $A, B \subset S_0$, is large enough. In other words,
\begin{align*}
\left|\mathscr{W}^n\left(\{\vec{x}_\tau\}_{\tau=1}^r, \{ \vec{x}_\theta + \vec{a}\}_{\theta=r^+}^{r+s}  \right)\right| \leq C(n) e^{-m_0|a^+|},
\end{align*}
for some constant $C(n)$ dependent on $n$, and $a^+$ larger than the space-like separation between the two sets in $S_0$.

Note that Equation (\ref{e.cd.1}) holds, even for space-like distance $|a^+|$ small and close to zero. When there is no minimum mass gap $m_0$, the given upper bound is not meaningful for small $|a^+|$. For $|a^+|$ small and for each component mass gap $m_n > 1$, we see that the heavier the mass gap $m_n$, the faster is the rate of decay. Of course, when $|a^+|$ is large, then the contribution of $m_0$ to the inverse power decay law is minimal.

An inverse power decay is definitely not as rapid a decline, as an exponential decay, which was proved in \cite{Araki1962}. They showed that the decay is at most $\frac{1}{|a^+|^{3/2}}e^{-m_0|a^+|}$, which is faster than our result. But do note that the authors proved exponential decay, by using a homogeneous wave equation, for the cases when $\varphi_1$, $\varphi_2$ are both compactly supported, and the space-like separation needs to be greater than some $\epsilon$.  In contrast, we used an elliptic equation, reason given in the footnote. Any of the three cases in Remark \ref{r.p.1} will imply that $c_r \neq 0$. The rate of decay $\frac{1}{|a^+|}e^{-m_0|a^+|}$ which we have shown, is asymptotically the same as the Yukawa potential, who used it as a basis to describe the nuclear force. Refer to \cite{peskin1995}.

For large space-like separation, an exponential decay will decay a lot faster than the inverse power decay given in Equation (\ref{e.cd.1}). But for small space-like separation, an inverse power decay might be better.

Local commutativity in Wightman's last axiom allows us to define a Lorentz transformation $\Lambda_n$ in Definition \ref{d.c.1}, which implies an exponential decay $e^{-m_n|a^+|}$, but depending on each $n \in \mathbb{N}$. We will need a Yang-Mills path integral, to prove the existence of a minimum positive mass gap $m_0 \leq m_n$, which gives us a best possible exponential decay $e^{-m_0|a^+|}$, independent of the representation $\rho_n$ on the simple Lie algebra $\mathfrak{g}$, for the vacuum expectation.
\end{rem}

\appendix

\section{Surface Integrals}

Let $\sigma\equiv ( \sigma_0, \sigma_1, \sigma_2, \sigma_3): [0,1]^2 \equiv I^2 \rightarrow \bR^4$ be a parametrization of a surface $S \subset \bR^4$. Here, $\sigma' = \partial \sigma/\partial s$ and $\dot{\sigma} = \partial \sigma/\partial t$.

\begin{defn}\label{d.r.1}
For $a,b=0,1,2,3$, define Jacobian matrices,
\begin{align}
J_{ab}^\sigma(s,t) = \left(
\begin{array}{cc}
\sigma_a'(s,t) & \dot{\sigma}_a(s,t) \\
\sigma_b'(s,t) & \dot{\sigma}_b(s,t) \\
\end{array}
\right),\ a \neq b, \nonumber
\end{align}
and write $|J^\sigma_{ab}| = \sqrt{[\det{J^\sigma_{ab}}]^2}$ and $W_{ab}^{ cd} := J_{cd}^\sigma J_{ab}^{\sigma, -1}$, $a, b, c, d$ all distinct. Note that $W_{cd}^{ab} = (W_{ab}^{cd})^{-1}$.

For $a,b, c, d$ all distinct, define $\rho_\sigma^{ab}: I^2 \rightarrow \bR$ by
\begin{align}
\rho_\sigma^{ab} =& \frac{1}{\sqrt{\det\left[ 1+ W_{ab}^{cd,T}W_{ab}^{cd}\right]}} \equiv \frac{|J_{ab}^\sigma|}{\sqrt{\det\left[J_{ab}^{\sigma,T}J_{ab}^\sigma + J_{cd}^{\sigma,T}J_{cd}^\sigma\right]}}. \label{e.ri.1}
\end{align}
And \beq \int_S d\rho := \sum_{0 \leq a<b \leq 3}\int_{I^2}\rho_\sigma^{ab}(s,t)|J_{ab}^\sigma|(s,t)
\ ds dt, \nonumber \eeq which is independent of the parametrization $\sigma$ used.
\end{defn}

We leave to the reader to check that $\int_S d\rho$ gives us the area of the surface $S$. Area of a surface is invariant under spatial rotation, but it is not invariant under boost.

To construct an unitary representation of the Lorentz group, we need to consider imaginary time axis. Instead of using $\rho_\sigma^{ab}$ as defined in Equation (\ref{e.ri.1}) by some parametrization $\sigma$ for $S$, we will replace the time component $\sigma_0$ with imaginary time $i \sigma_0$.

\begin{defn}\label{d.r.2}
Let $\sigma: [0,1]^2 \equiv I^2 \rightarrow \bR^4$ be a parametrization of a surface $S \subset \bR^4$.
For $a,b=0,1,2,3$ and $a < b$, define Jacobian matrices,
\begin{align}
\acute{J}_{ab}^\sigma(s,t) =
\left\{
  \begin{array}{ll}
    \ \ \ \ J_{ab}^\sigma(s,t), & \hbox{$a \neq 0$;} \\
    \left(
               \begin{array}{cc}
                 i\sigma_a'(s,t) & i\dot{\sigma}_a(s,t) \\
                 \sigma_b'(s,t) & \dot{\sigma}_b(s,t) \\
               \end{array}
             \right), & \hbox{$a = 0$.}
  \end{array}
\right.
\nonumber
\end{align}

For $a,b, c, d$ all distinct, define $\acute{\rho}_\sigma^{ab}: I^2 \rightarrow \bC$ by
\begin{align*}
\acute{\rho}_\sigma^{ab} :=& \frac{\det \acute{J}_{ab}^\sigma}{\sqrt{\det\left[\acute{J}_{ab}^{\sigma,T}\acute{J}_{ab}^\sigma + \acute{J}_{cd}^{\sigma,T}\acute{J}_{cd}^\sigma\right]}},
\end{align*}
and
\begin{align*}
\int_S\ d\acute{\rho} :=& \sum_{0\leq a < b \leq 3}\int_{I^2}\acute{\rho}_\sigma^{ab}(s,t)[\det \acute{J}_{ab}^\sigma(\hat{s})]\ d\hat{s}, \\
 \int_S d|\acute{\rho}| :=& \int_{I^2}\left| \sum_{0\leq a < b \leq 3}\acute{\rho}_\sigma^{ab}(\hat{s})[\det \acute{J}_{ab}^\sigma](\hat{s})\right|\ d\hat{s}.
\end{align*}

We will also write
\begin{align*}
\acute{\rho}_\sigma(\hat{s}) =&  \sum_{0\leq a < b \leq 3}\acute{\rho}_\sigma^{ab}(\hat{s})[\det \acute{J}_{ab}^\sigma](\hat{s}), \\
|\acute{\rho}_\sigma|(\hat{s}) =& \left| \sum_{0\leq a < b \leq 3}\acute{\rho}_\sigma^{ab}(\hat{s})[\det \acute{J}_{ab}^\sigma](\hat{s})\right|.
\end{align*}
\end{defn}

\begin{lem}\label{l.b.5}
Let $S$ be a compact time-like or space-like surface, contained in a plane. Then $\int_S d\acute{\rho}$ and $\int_S d|\acute{\rho}|$ remain invariant under any Lorentz transformation $\Lambda: S \mapsto \hat{S} = \Lambda S$, $\Lambda$ is a $4 \times 4$ Lorentz matrix.
\end{lem}

\begin{proof}
Let $\sigma: I^2 \rightarrow S$ be a parametrization of $S$. Then $\hat{\sigma} = \Lambda \sigma$ is a parametrization of $\hat{S} = \Lambda S$. We leave to the reader to verify that \beq
\sum_{0\leq a < b \leq 3}\acute{\rho}_{\hat{\sigma}}^{ab}[\det \acute{J}_{ab}^{\hat{\sigma}}] = \sqrt{[\sigma' \cdot \sigma'][\dot{\sigma} \cdot \dot{\sigma}] - [\sigma' \cdot \dot{\sigma}]^2}. \nonumber \eeq Since $\vec{x} \cdot \vec{y}$ is invariant under Lorentz transformation, we thus have
\begin{align*}
\sum_{0\leq a < b \leq 3}\acute{\rho}_{\hat{\sigma}}^{ab}[\det \acute{J}_{ab}^{\hat{\sigma}}]
&= \sum_{0\leq a < b \leq 3}\acute{\rho}_{\sigma}^{ab}[\det \acute{J}_{ab}^{\sigma}].
\end{align*}
This shows that $\int_{\hat{S}} d\acute{\rho}= \int_{S} d\acute{\rho}$ and $\int_{\hat{S}} d|\acute{\rho}|= \int_{S} d|\acute{\rho}|$.
\end{proof}

\begin{rem}\label{r.a.3}
\begin{enumerate}
  \item When $S$ is a surface in spatial $\bR^3$, we see that $\int_S d|\acute{\rho}| = \int_S d\rho$, which is the area of the surface $S$.
  \item Note that $\int_S d\acute{\rho}$ is complex valued.
  \item As a consequence, when $S$ is a space-like surface, $\int_S d\acute{\rho}$ is real; when $S$ is a time-like surface, $\int_S d\acute{\rho}$ is purely imaginary. That is, $\int_S d\acute{\rho} = i\int_S d|\acute{\rho}|$.
  \item By definition and from the proof, we see that $\acute{\rho}_{\sigma} = \acute{\rho}_{\Lambda \sigma + \vec{a}}$, for any Lorentz transformation $\Lambda$ and any vector $\vec{a}$.
\end{enumerate}
\end{rem}

\section{Lorentz transformation of a space-like vector}

\begin{lem}\label{l.b.3}
Let $S$ be a space-like rectangular surface contained in a plane. Without loss of generality, we assume a parametrization of $S$ given in Definition \ref{d.ts.1}.
From Definition \ref{d.a.1}, we can define a basis $\{\hat{f}_a\}_{a=0}^3$.

Let $0 \neq \vec{v}$ be a space-like directional vector lying in the time-like plane $\vec{x} + S^\flat$, i.e. $\vec{v} = s\hat{f}_0 + \bar{s}\hat{f}_1 \in S^\flat$, whereby $|\bar{s}| > |s|$, and $\vec{x} \in S$.
There exists a finite sequence of translations and Lorentz transformations, which depends on $\vec{x}, \vec{v}$, such that $\vec{v}\in \vec{x} + S^\flat \mapsto -\vec{v} \in \vec{x} + S^\flat$.
\end{lem}

\begin{proof}
First, assume that $S$ contains the zero vector. Without any loss of generality, we assume that $S$ is parametrized according to Definition \ref{d.ts.1}, with $\vec{a} = 0$. From this parametrization, we can define $\{\hat{f}_a\}_{a=0}^3$ as in Definition \ref{d.a.1}, and it is easy to see that we have a Lorentz transformation $\bar{\Lambda}$ which sends $\bar{\Lambda}: \hat{f}_a \mapsto e_a$, for $a=0, \cdots, 3$. Hence, $\bar{\Lambda}S \subseteq S_0$.

Then, $\bar{\Lambda}\vec{x} = \sum_{a=0}^3 x^a e_a = x^2 e_2 + x^3 e_3$.
Let $R_i$ be a rotation about $x^i$-axis, and $\Lambda$ be boost in the $e_1$ direction, that sends
\beq \Lambda: se_0 + \bar{s}e_1 \mapsto \left[\sgn(\bar{s}) \sqrt{\bar{s}^2 - s^2}\right] e_1, \nonumber \eeq whereby $\sgn(\bar{s})$ is the sign of $\bar{s}$.

Then, we have
\begin{align*}
\vec{x}& + \vec{v} = s\hat{f}_0 + \bar{s}\hat{f}_1 + x^2\hat{f}_2 +x^3\hat{f}_3  \\
&\longrightarrow_{\bar{\Lambda}} (s, \bar{s}, x^2, x^3) \equiv se_0 + \bar{s}e_1 + x^2 e_2 + x^3 e_3 \longrightarrow_{\Lambda} (0, \sgn(\bar{s}) \sqrt{\bar{s}^2 - s^2}, x^2, x^3) \\
&\longrightarrow_{R_3} (0, \sgn(\bar{s}) \sqrt{x^{2,2} + \bar{s}^2 - s^2}, 0, x^3) \longrightarrow_{R_2} -(0, \sgn(\bar{s}) \sqrt{x^{2,2} + \bar{s}^2 - s^2}, 0, x^3) \\
&\longrightarrow_{R_3^{-1}} -(0, \sgn(\bar{s}) \sqrt{\bar{s}^2 - s^2}, x^2, x^3) \longrightarrow_{\Lambda^{-1}} -(s, \bar{s}, x^2, x^3) \\
&\longrightarrow_{R_1} (-s, -\bar{s}, x^2, x^3) \longrightarrow_{\bar{\Lambda}^{-1}} -s\hat{f}_0 - \bar{s}\hat{f}_1 + x^2\hat{f}_2 +x^3\hat{f}_3.
\end{align*}
Hence, we have \beq \tilde{\Lambda}(\vec{x}, \vec{v}) := \bar{\Lambda}^{-1}R_1 \Lambda^{-1} R_3^{-1} R_2 R_3 \Lambda\bar{\Lambda}, \nonumber \eeq which sends $\vec{v} \in  \vec{x} + S^\flat \mapsto -\vec{v} \in \vec{x} + S^\flat$. Note that all the rotations depend on $\vec{x}$ and $\vec{v}$, except $R_1$ and $R_2$, which rotate by angle $\pi$ radians.

Suppose $\tilde{S}$ does not contain the origin. Then, we consider for $\vec{a} \in \tilde{S}$, $S = U(-\vec{a},1)\tilde{S} $, which contains the origin. For any space-like vector $\vec{v} \in \vec{x}-\vec{a} + S^\flat$, we can map $\vec{v} \mapsto -\vec{v}$, using the above transformations. Finally, apply $U(\vec{a}, 1)$ to translate back to $\tilde{S}$.
\end{proof}

\begin{rem}\label{r.p.3}
On the $x^2-x^3$ plane or $S_0$, any two points on it are space-like separated. Clearly, there is a Lorentz transformation (rotation) which transforms $\vec{p} \in S_0 \mapsto -\vec{p} \in S_0$. In general, we see that there is a sequence of translations and Lorentz transformations which transforms $\vec{p} \in S \mapsto -\vec{p} \in S$, for some space-like plane $S$.

On a time-like plane $S$, we do not have a rotation in the Lorentz group that maps a time-like vector $\vec{p} \in S \mapsto -\vec{p} \in S$. It is not possible to do time and space inversion simultaneously for a time-like vector using Lorentz transformations.
\end{rem}

\begin{lem}\label{l.b.4}
Refer to Definition \ref{d.a.1}. Let $S^\flat$ be a time-like plane spanned by $\{\hat{f}_0, \hat{f}_1\}$. Fix a time-like vector $u \in S^\flat$. For any time-like or space-like vector $v \in S^\flat$, there exist a $\phi$, $\theta$ depending on $u$ and $v$ respectively, such that
\begin{align*}
v \cdot u =& \sinh(\phi-\theta) = \left(
  \begin{array}{c}
    \cosh(\phi-\theta) \\
    \sinh(\phi-\theta) \\
  \end{array}
\right) \cdot
\left(
  \begin{array}{c}
    0 \\
    1 \\
  \end{array}
\right),
\end{align*}
when $v$ is space-like; and
\begin{align*}
v \cdot u =& -\cosh(\phi-\theta) = \left(
  \begin{array}{c}
    \cosh(\phi-\theta) \\
    \sinh(\phi-\theta) \\
  \end{array}
\right) \cdot
\left(
  \begin{array}{c}
    1 \\
    0 \\
  \end{array}
\right),
\end{align*}
when $v$ is time-like.
\end{lem}

\begin{proof}
Any vector on $S^\flat$ can be represented as a 2-vector $(a,b)^T$, which are the coordinates with respect to the basis $\{\hat{f}_0, \hat{f}_1\}$. Define \beq \Lambda(\theta) =
\left(
  \begin{array}{cc}
    \cosh \theta &\ \sinh \theta \\
    \sinh \theta &\ \cosh \theta \\
  \end{array}
\right). \label{e.a.3} \eeq For a time-like vector, we will write it as $\Lambda(\phi)(1,0)^T$.
For any space-like vector, we will write it as $\Lambda(\theta)(0,1)^T$. Then, we have
\begin{align*}
\Lambda(\theta)
\left(
  \begin{array}{c}
    0 \\
    1 \\
  \end{array}
\right) \cdot
\Lambda(\phi)
\left(
  \begin{array}{c}
    1 \\
    0 \\
  \end{array}
\right) =&
\left(
  \begin{array}{c}
    \sinh \theta \\
    \cosh \theta \\
  \end{array}
\right)\cdot
\left(
  \begin{array}{c}
    \cosh \phi \\
    \sinh \phi \\
  \end{array}
\right) \\
=& -\sinh\theta \cosh \phi + \cosh\theta \sinh \phi = \sinh(\phi-\theta) \\
=&
\left(
  \begin{array}{c}
    \cosh(\phi-\theta) \\
    \sinh(\phi-\theta) \\
  \end{array}
\right) \cdot
\left(
  \begin{array}{c}
    0 \\
    1 \\
  \end{array}
\right).
\end{align*}
A similar calculation will show that \beq
\Lambda(\theta)
\left(
  \begin{array}{c}
    1 \\
    0 \\
  \end{array}
\right) \cdot
\Lambda(\phi)
\left(
  \begin{array}{c}
    1 \\
    0 \\
  \end{array}
\right) = - \cosh(\theta-\phi) =
\left(
  \begin{array}{c}
    \cosh(\phi-\theta) \\
    \sinh(\phi-\theta) \\
  \end{array}
\right) \cdot
\left(
  \begin{array}{c}
    1 \\
    0 \\
  \end{array}
\right)
. \nonumber \eeq
\end{proof}

From the calculations, when $\theta = \phi$, we see that $e^{-iu \cdot v} = \cos(0)=1$, which is real. This happens when $v$ is space-like, $u$ is time-like. However, if both are time-like or both are space-like, then we have $e^{-iu \cdot v} = \cos(u \cdot v) - i\sin(u \cdot v) \neq 1$, unless $u\cdot v = n\pi$ for $n \in \mathbb{Z}$.



\begin{thebibliography}{10}

\bibitem{MR1312606}
R.~W.~R. Darling, {\em Differential forms and connections}.
\newblock Cambridge: Cambridge University Press, 1994.

\bibitem{jaffe}
A.~Jaffe and E.~Witten, ``Quantum yang-mills theory,'' {\em The millennium
  prize problems}, no.~1, p.~129, 2006.

\bibitem{streater}
R.~F.~S. und A~S~Wightman, {\em PCT, Spin Statistics, And All That}.
\newblock New York, Amsterdam: W A Benjamin Inc., 1964.

\bibitem{glimm1981}
J.~Glimm and A.~Jaffe, {\em Quantum physics: a functional integral point of
  view}.
\newblock Springer-Verlag, 1981.

\bibitem{hall2015}
B.~C. Hall, {\em Lie groups, Lie algebras, and representations: an elementary
  introduction}, vol.~222.
\newblock Springer, 2015.

\bibitem{YM-Lim02}
A.~P.~C. Lim, ``{Wilson Area Law formula on $\mathbb{R}^4$},'' {\em arXiv
  2211.07064 [math-ph] e-print}, 2022.

\bibitem{Skript2}
T.~{Weigand}, ``{Quantum Field Theory I + II},'' {\em
  http://www.thphys.uni-heidelberg.de/~weigand/QFT2-14/SkriptQFT2.pdf}.

\bibitem{YM-Lim01}
A.~P.~C. Lim, ``{Abstract Wiener measure using abelian Yang-Mills action on
  $\mathbb{R}^4$},'' {\em arXiv 1701.01529 [math.PR] e-print}, 2017.

\bibitem{peskin1995}
M.~Peskin and D.~Schroeder, {\em An Introduction to Quantum Field Theory}.
\newblock Advanced book classics, Avalon Publishing, 1995.

\bibitem{APP1977231}
T.~Appelquist, M.~Dine, and I.~Muzinich, ``The static potential in quantum
  chromodynamics,'' {\em Physics Letters B}, vol.~69, no.~2, pp.~231 -- 236,
  1977.

\bibitem{PhysRevLett31851}
S.~Coleman and D.~J. Gross, ``Price of asymptotic freedom,'' {\em Phys. Rev.
  Lett.}, vol.~31, pp.~851--854, Sep 1973.

\bibitem{grif2008}
D.~Griffiths, {\em Introduction to Elementary Particles}.
\newblock Physics textbook, Wiley, 2008.

\bibitem{MR0461643}
H.~H. Kuo, {\em Gaussian measures in {B}anach spaces}.
\newblock Lecture Notes in Mathematics, Vol. 463, Berlin: Springer-Verlag,
  1975.

\bibitem{knapp2016}
A.~Knapp, {\em Representation Theory of Semisimple Groups: An Overview Based on
  Examples (PMS-36)}.
\newblock Princeton Mathematical Series, Princeton University Press, 2016.

\bibitem{MR1968551}
E.~Wigner, ``On unitary representations of the inhomogeneous lorentz group,''
  {\em Annals of Mathematics}, vol.~40, no.~1, pp.~149--204, 1939.

\bibitem{GUT}
J.~Baez and J.~Huerta, ``{The algebra of grand unified theories},'' {\em Bull.
  Amer. Math. Soc.}, vol.~47, pp.~483--552, 2010.

\bibitem{hepph6330}
G.~P. {Lepage}, ``{What is Renormalization?},'' {\em ArXiv High Energy Physics
  - Phenomenology e-prints}, June 2005.

\bibitem{fey1981479}
R.~P. Feynman, ``The qualitative behavior of yang-mills theory in 2 + 1
  dimensions,'' {\em Nuclear Physics B}, vol.~188, no.~3, pp.~479 -- 512, 1981.

\bibitem{Ruelle}
D.~Ruelle, ``On the asymptotic condition in quantum field theory,'' vol.~35,
  pp.~147--163, 01 1962.

\bibitem{Nair}
V.~P. Nair, {\em Quantum field theory. A modern perspective}.
\newblock New York: Springer, 2005.

\bibitem{ALBEVERIO197439}
S.~Albeverio and R.~Hoegh-Krohn, ``The wightman axioms and the mass gap for
  strong interactions of exponential type in two-dimensional space-time,'' {\em
  Journal of Functional Analysis}, vol.~16, no.~1, pp.~39 -- 82, 1974.

\bibitem{Araki1962}
H.~Araki, K.~Hepp, and D.~Ruelle, ``{On the asymptotic behaviour of Wightman
  functions in space-like directions},'' {\em Helv. Phys. Acta}, vol.~35,
  no.~III, pp.~164--174, 1962.

\end{thebibliography}
\end{document}